\documentclass[review]{elsarticle}

\usepackage{lineno,hyperref}
\usepackage{amssymb}
\usepackage{color}
\usepackage{amsmath}
\usepackage{tikz}
\usepackage{amsfonts}
\usepackage{mathrsfs}
 \usepackage{amsthm}

\usepackage{bbding}
\usepackage{mathrsfs}
\usepackage{amsmath, amsfonts, amssymb, mathrsfs, txfonts}
\usepackage{graphicx, subfigure}
\usepackage{wasysym}
\usepackage{color}
\usepackage{bm}
\usepackage{enumerate}

\usepackage{pifont}
\usepackage{txfonts}
\usepackage{amsmath}
\usepackage{graphicx}
\usepackage{amsfonts}
\usepackage{amssymb}
\usepackage{mathrsfs,psfrag,eepic,epsfig}
\usepackage{epstopdf}

\modulolinenumbers[5] 

\journal{Journal of \LaTeX\ Templates}

\makeatletter \@addtoreset{equation}{section}

\newtheorem{lemma}{Lemma}

\renewcommand{\baselinestretch}{1.25}
\begin{document}

\begin{frontmatter}

\title{Riemann-Hilbert approach and soliton solutions for the higher-order dispersive nonlinear Schr\"{o}dinger equation with nonzero boundary conditions \tnoteref{mytitlenote}}
\tnotetext[mytitlenote]{
Corresponding author.\\
\hspace*{3ex}\emph{E-mail addresses}: sftian@cumt.edu.cn,
shoufu2006@126.com (S. F. Tian) }

\author{Zhi-Qiang Li, Shou-Fu Tian\tnoteref{mytitlenote}, Jin-Jie Yang}
\address{
School of Mathematics and Institute of Mathematical Physics, China University of Mining and Technology,\\ Xuzhou 221116, People's Republic of China\\
}

\begin{abstract}
In this work, the higher-order dispersive nonlinear Schr\"{o}dinger equation with non-zero boundary conditions at infinity is investigated including the simple and double zeros of the scattering coefficients. We introduce a appropriate Riemann surface and uniformization variable in order to deal with the double-valued functions occurring in the process of direct scattering. Then, the direct scattering problem is analyzed involving the analyticity, symmetries and asymptotic behaviors. Moreover, for the cases of simple and double poles, we study the discrete spectrum and residual conditions, trace foumulae and theta conditions and the inverse scattering problem which is solved via   the Riemann-Hilbert method. Finally,  for the both   cases, we construct the   soliton and breather solutions   under the condition of reflection-less potentials.  Some interesting phenomena of the soliton and  breather
solutions  are analyzed graphically   by considering the
influences of each parameters.
\end{abstract}

\begin{keyword}
The higher-order dispersive nonlinear Schr\"{o}dinger equations \sep Nonzero boundary conditions \sep Riemann-Hilbert method.
\end{keyword}

\end{frontmatter}


\section{Introduction}

Nonlinear Schr\"{o}dinger (NLS) equation reads
\begin{equation}
iu_{t}+u_{xx}+2|u|^{2}u=0
\end{equation}
is a very important soliton equation.
As basic physical models, the NLS equation and its extensions play an important role in various fields of nonlinear science such as deep water waves\cite{I-1}, plasma physics\cite{I-2,I-3}, nonlinear optical fibers\cite{I-6,I-7}, magneto-static spin waves\cite{I-8}, etc. Since many NLS equations are completely integrable infinite-dimensional Hamiltonian systems, they have rich mathematical structure. Therefore, various research on NLS equations are still popular. In particular, the studies on the exact solutions of the NLS equations have gradually become an important branch.
There are three well-known derivative
NLS equations, including the Kaup-Newell equation \cite{KN}, the Chen-Lee-Liu
equation \cite{CLL} and the Gerdjikov-Ivanov equation \cite{GI,Tian-PAMS}. It is known that these three equations
may be transformed into each other by implicit gauge transformations, and the method of
gauge transformation can also be applied to some generalized cases \cite{Fan-dnls}.
In this respect,  there are a lot of research for the extended NLS equations \cite{I-13}-\cite{I-18}.

In this work, we investigate the higher-order dispersive nonlinear Schr\"{o}dinger (HDNLS) equation\cite{I-19,I-20} with nonzero boundary conditions (NZBCs) at infinity which have extensive applications in physical fields. The HDNLS equation takes the form
\begin{gather}
iu_{t}+u_{xx}+2|u|^{2}u+\tau (u_{xxxx}+8|u|^{2}u_{xx}+2u^{2}u^{*}_{xx}+4u|u_{x}|^{2}+6u^{*}u_{x}^{2}+6|u|^{4}u )=0 ,\notag \\
\lim_{x\rightarrow \pm \infty}u(x,t)=u_{\pm}, \qquad \mid u_{\pm}\mid=u_{0}\neq 0,
\label{1.1}
\end{gather}
where $u$ is complex function of variables $x$, $t$ and represents the slowly changing envelope of the wave and $\tau=\frac{\epsilon^{2}}{12}$ is a small dimensionless parameter. The HDNLS equation plays a dominant role in controlling the ultrashort optical pulse propagation in a long-distance, high-speed fiber transmission systems, with the higher-order nonlinear effects such as fourth-order dispersion\cite{I-21}-\cite{I-23}. Meanwhile, this equation can be used to describe the nonlinear spin excitations in the one-dimensional isotropic biquadratic Heisenberg ferromagnetic spin with the octupole-dipole interaction\cite{I-24,I-25}. To study the HDNLS equation, experts and scholars have used many methods such as Darboux transformation method, modified Darboux transformation method\cite{I-26,I-27}. However, what we will study is inverse scattering transform (IST) for the HDNLS equation with nonzero boundary conditions. To the best of our knowledge, for this equation, the research using this method has not been reported yet.

The IST is first presented  in order to solve exactly the famous Korteweg-de Vries equation \cite{I-29}. Then, Zakharov and Shabat\cite{I-3} show that the method can be applied to physically significant nonlinear evolution equation, namely, the nonlinear Schr\"{o}dinger equation. After this, the study on the NLS equations using the IST is more and more popular. Recently,  the IST is extended to study the NLS equations with NZBCs \cite{I-30}-\cite{I-48}.

Two goals will be achieved in this work. On one hand, according to the lax pairs of the equation \eqref{1.1} and through the direct scattering transform and inverse scattering transform, some results are given, involving the analyticity, symmetries and asymptotic behaviors of the scattering coefficients, the establishment of a generalized Riemann-Hilbert problem and the construction of the discrete spectrum and residual conditions. On the other hand, based on the above results, the solutions of the HDNLS equation with NZBCs will be derived for the simple and double poles. In view of the complexity of the HDNLS equation and its lax pairs, some computational skills are used and a lot of calculations are made during the analysis to achieve that two goals.

The outline of this work is as follows. In section 2, we explore the simple poles of the HDNLS equation with NZBCs involving the direct scattering problem and inverse scattering problem. Meanwhile, the soliton solutions under the reflection-less potential condition are given. In section 3, we derive the inverse scattering problem of the HDNLS equation with NZBCs and double poles. In addition, the soliton solutions are presented for some special cases. In section 4, we make a comparison with the results of other experts and scholars. In section 5, some conclusions and discussions are given.

\section{The HDNLS equation with NZBCs: simple poles}
In this section, we consider the simple poles case of the HDNLS equation with NZBCs. Firstly, we seek for an appropriate transformation to simplify the boundary conditions and make the analysis process easier.

The Lax pair of the equation \eqref{1.1} reads
\begin{gather}
\Psi_{x}=X\Psi,\qquad \Psi_{t}=T\Psi, \notag \\
\Psi=(\Psi_{1},\Psi_{2})^{T}, \label{2.1}
\end{gather}
where $\Psi_{i}, (i=1,2)$ are eigenfunctions, and
\begin{gather}
X=-ik\sigma_{3}+U, \notag \\
T=[3i\tau |u|^{4}+i|u|^{2}+i\tau (u^{*}u_{xx}+uu^{*}_{xx}-|u_{x}|^{2})+8i\tau k^{4}+2k\tau (uu^{*}_{x}-u_{x}u^{*})  \notag \\
-2ik^{2}(2\tau |u|^{2}+1)]\sigma_{3}-8\tau k^{3}U-4i\tau k^{2}\sigma_{3}U_{x}+6i\tau U^{2}U_{x}\sigma_{3} \notag \\
+i\sigma_{3}U_{x}+i\tau \sigma_{3}U_{xxx}+2k(U+\tau U_{xx}-2\tau U^{3}),\label{2.2}
\end{gather}
with
\begin{equation}
U=\left(                 
  \begin{array}{cc}   
    0 & u(x,t)  \\  
    -u^{*}(x,t) & 0 \\  
  \end{array}
\right), \qquad
\sigma_{3}=\left(                 
  \begin{array}{cc}   
    1 & 0  \\  
    0 & -1 \\  
  \end{array}
\right). \notag
\end{equation}\\
It is not hard to check that $X,T$ satisfy the zero curvature equation $X_{t}-T_{x}+[X,T]=0$, which is the compatibility condition of \eqref{1.1}.

\noindent \textbf {Theorem 2.1}
\emph{Under the transformation}
\begin{align}
\begin{split}
&u= qe^{2i(3\tau q_{0}^{4}+q_{0}^{2})t},\\
&\Psi= \phi e^{(3\tau q_{0}^{4}+q_{0}^{2})\sigma_{3}}, \notag
\end{split}
\end{align}
\emph{then the equation \eqref{1.1} changes to}
\begin{gather}
iq_{t}+q_{xx}+2(|q|^{2}-3\tau q_{0}^{4}-q_{0}^{2})q+\tau (q_{xxxx}+8|q|^{2}q_{xx} \notag \\
+2q^{2}q^{*}_{xx}+4q|q_{x}|^{2}+6q^{*}q_{x}^{2}+6|q|^{4}q )=0 ,\notag \\
\lim_{x\rightarrow \pm \infty}q(x,t)=q_{\pm}, \qquad \mid q_{\pm}\mid=q_{0}\neq 0,
\label{2.3}
\end{gather}
\emph{and the equation \eqref{2.1} changes to }
\begin{gather}
\phi_{x}=X\phi,\qquad \phi_{t}=T\phi, \notag \\
\phi=(\phi_{1},\phi_{2})^{T}, \label{2.4}
\end{gather}
\emph{where $\phi_{i}, (i=1,2)$ are eigenfunctions, and}
\begin{gather}
X=-ik\sigma_{3}+Q, \notag \\
T=[3i\tau |q|^{4}-3i\tau |q_{0}|^{4}+i|q|^{2}-i|q_{0}|^{2}+\tau (u^{*}u_{xx}+uu^{*}_{xx}-|u_{x}|^{2}) +8i\tau k^{4} \notag \\
+2k\tau (uu^{*}_{x}-u_{x}u^{*})
-2ik^{2}(2\tau |u|^{2}+1)]\sigma_{3}-8\tau k^{3}U-4i\tau k^{2}\sigma_{3}U_{x} \notag \\
+6i\tau U^{2}U_{x}\sigma_{3}+i\sigma_{3}U_{x}+i\tau \sigma_{3}U_{xxx}+2k(U+\tau U_{xx}-2\tau U^{3}),\notag
\end{gather}
\emph{with}
\begin{equation}
Q=\left(                 
  \begin{array}{cc}   
    0 & q(x,t)  \\  
    -q^{*}(x,t) & 0 \\  
  \end{array}
\right), \qquad
\sigma_{3}=\left(                 
  \begin{array}{cc}   
    1 & 0  \\  
    0 & -1 \\  
  \end{array}
\right). \notag
\end{equation}

\subsection{Direct scattering problem with NZBCs}
In this section, the direct scattering problem will be studied. In the direct scattering process, the analyticity and asymptotic of the eigenfunction, asymptotic and symmetries of the scattering matrix, discrete spectrum, and residue conditions will be given. Meanwhile, the emergence of multi-valued functions will complicate the problem during the analysis process. So we introduce the two-sheeted Riemann surface to simplify the analysis process.

\subsubsection{Riemann surface and uniformization coordinate}
Letting $x\rightarrow\pm\infty$, the following asymptotic scattering problem
\begin{align}\label{2.5}
&\left\{ \begin{aligned}
&\psi_{x}=X_{\pm}\psi,\quad X_{\pm}=\lim_{x\rightarrow\pm\infty}X=-ik\sigma_{3}+Q_{\pm},\\
&\psi_{t}=T_{\pm}\psi, \quad T_{\pm}=\lim_{x\rightarrow\pm\infty}T=(-8\tau k^{3}+2k+4\tau kq^{2}_{0})X_{\pm},
     \end{aligned}  \right.
\end{align}
can be obtained from the \emph{Theorem 2.1},
where $Q_{\pm}=\lim_{x\rightarrow\pm\infty}Q=\left(\begin{array}{cc}
    0  & q_{\pm} \\
    -q^{*}_{\pm} & 0  \\
  \end{array}\right).$
It is easy to calculate that the eigenvalues of $X_{\pm}$ is $\pm i\sqrt{k^{2}+q_{0}^{2}}$. Obviously, the eigenvalues are doubly branched. So we introduce the two-sheeted Riemann surface defined by
\begin{align}\label{2.6}
\lambda^{2}=k^{2}+q_{0}^{2}
\end{align}
to simplify the problem. From the definition, we know that the branch points $k=\pm iq_{0}$ can be obtained when $\sqrt{k^{2}+q_{0}^{2}}=0$. Therefore, the two-sheeted Riemann surface completed by gluing together two copies of extended complex $k$-plane $S_{1}$ and $S_{2}$ along the cut $iq_{0}[-1,1]$ between the branch points $k=\pm iq_{0}$. And the $\lambda (k)$ is a single-value function.
Here, we introduce the local polar coordinates
\begin{align}
k+iq_{0}=r_{1}e^{i\theta_{1}},\quad k-iq_{0}=r_{2}e^{i\theta_{2}},\quad -\frac{\pi}{2}<\theta_{1},\theta_{2}<\frac{3\pi}{2}, \notag
\end{align}
then, we obtain a single-valued analytical function on the Riemann surface
\begin{align}
\lambda(k)=&\left\{\begin{aligned}
&(r_{1}r_{2})^\frac{1}{2}e^\frac{{\theta_{1}+\theta_{2}}}{2}, \quad &on\quad S_{1},\\
-&(r_{1}r_{2})^\frac{1}{2}e^\frac{{\theta_{1}+\theta_{2}}}{2}, \quad &on\quad S_{2}.
\end{aligned} \right. \notag
\end{align}
It is easy to get the relations between $\lambda$ and $k$ as
\begin{align}
&\left\{\begin{aligned}
&Im k>0 ~ of ~sheet ~S_{1} ~and ~Im k<0 ~of ~sheet ~S_{2} ~are ~mapped ~into ~Im \lambda >0,\\
&Im k<0 ~ of ~sheet ~S_{1} ~and ~Im k>0 ~of ~sheet ~S_{2} ~are ~mapped ~into ~Im \lambda <0,
\end{aligned} \right. \notag
\end{align}
by considering the imaginary part of $\lambda (k)$. Conveniently, define the uniformization variable $z$ \cite{I-50}
\begin{align}
z=k+\lambda, \notag
\end{align}
and compare with the equation \eqref{2.6}, we obtain the following two single-value functions
\begin{align}\label{2.7}
\lambda(z)=\frac{1}{2}(z+\frac{q_{0}^{2}}{z}),\quad k(z)=\frac{1}{2}(z-\frac{q_{0}^{2}}{z}).
\end{align}
We consider the imaginary part of the Joukowsky transformation
\begin{align}\label{2.8}
\lambda=\frac{z(|z|^{2}-q_{0}^{2})+2q_{0}^{2}Rez}{2|z|^{2}}
\end{align}
which comes from equation \eqref{2.7}, and obtain that the upper half of the $\lambda$-plane maps to the upper half of the $z$-plane except for the inner part of the circle with $q_{0}$ as the radius and $0$ as the center of the circle and the lower half of the $z$-plane except for the outer part of the circle with $q_{0}$ as the radius and $0$ as the center of the circle. Consistently, the lower half of the $\lambda$-plane maps to the lower half of the $z$-plane except for the inner part of the circle with $q_{0}$ as the radius and $0$ as the center of the circle and the upper half of the $z$-plane except for the outer part of the circle with $q_{0}$ as the radius and $0$ as the center of the circle.
To sum up, the $Im\lambda\gtrless0$, respectively, maps to $D_{+}$ and $D_{-}$ where
\begin{align*}
D_{+}=\left\{z\in \mathbb{C}:\left(|z|^{2}-q_{0}^{2}\right)Im z>0\right\},\quad D_{-}=\left\{z\in \mathbb{C}:\left(|z|^{2}-q_{0}^{2}\right)Im z<0\right\}.\notag
\end{align*}
Based on the above analysis, we can summarize the following illustrations.
\\

\centerline{\begin{tikzpicture}
\path [fill=green] (-4.5,0) -- (-0.5,0) to
(-0.5,2) -- (-4.5,2);
\draw[-][thick](-4.5,0)--(-2.5,0);
\draw[fill] (-2.5,0) circle [radius=0.035];
\draw[->][thick](-2.5,0)--(-0.5,0)node[above]{$Rek$};
\draw[<-][thick](-2.5,2)node[right]{$Imk$}--(-2.5,1)node[right]{$iq_{0}$};
\draw[fill] (-2.5,1) circle [radius=0.035];
\draw[-][thick](-2.5,1)--(-2.5,0);
\draw[-][thick](-2.5,0)--(-2.5,-1)node[right]{$-iq_{0}$};
\draw[fill] (-2.5,-1) circle [radius=0.035];
\draw[-][thick](-2.5,-1)--(-2.5,-2);
\draw[fill] (-2.5,-0.3) node[right]{$0$};
\draw[fill] (-1.7,0.8) circle [radius=0.035] node[right]{$z^{*}_{n}$};
\draw[fill] (-1.7,-0.8) circle [radius=0.035] node[right]{$z_{n}$};
\path [fill=green] (0.5,0) -- (4.5,0) to
(4.5,2) -- (0.5,2);
\filldraw[white, line width=0.5](3.5,0) arc (0:180:1);
\filldraw[green, line width=0.5](1.5,0) arc (-180:0:1);
\draw[->][thick](0.5,0)--(1,0);
\draw[-][thick](1,0)--(2,0);
\draw[<-][thick](2,0)--(2.5,0);
\draw[fill] (2.5,0) circle [radius=0.035];
\draw[-][thick](2.5,0)--(3,0);
\draw[<->][thick](3,0)--(4,0);
\draw[-][thick](4,0)--(4.5,0)node[above]{$Rez$};
\draw[-][thick](2.5,2)node[right]{$Imz$}--(2.5,0);
\draw[-][thick](2.5,0)--(2.5,-2);
\draw[fill] (2.5,-0.3) node[right]{$0$};
\draw[fill] (2.5,1) circle [radius=0.035];
\draw[fill] (2.5,-1) circle [radius=0.035];
\draw[fill] (2.5,1.3) node[right]{$iq_{0}$};
\draw[fill] (2.5,-1.3) node[right]{$-iq_{0}$};
\draw[fill][red] (3.8,1.5) circle [radius=0.035] node[right]{$z^{*}_{n}$};
\draw[fill] (3.8,-1.5) circle [radius=0.035] node[right]{$z_{n}$};
\draw[fill] (1.8,0.5) circle [radius=0.035] node[right]{$-\frac{q^{2}_{0}}{z^{*}_{n}}$};
\draw[fill][red] (1.8,-0.5) circle [radius=0.035] node[right]{$-\frac{q^{2}_{0}}{z_{n}}$};
\draw[-][thick](3.5,0) arc(0:360:1);
\draw[-<][thick](3.5,0) arc(0:30:1);
\draw[-<][thick](3.5,0) arc(0:150:1);
\draw[->][thick](3.5,0) arc(0:210:1);
\draw[->][thick](3.5,0) arc(0:330:1);
\end{tikzpicture}}
\noindent { \small \textbf{Figure 1.} (Color online) The left one , the first sheet of Riemann surface, shows the discrete spectrums and the two regions with $Imk > 0$ (green) and $Imk < 0$ (white). The right one, the complex $z$-plane, shows the discrete spectrums of the scattering problem i.e. the zeros of $s_{11}(z)$ (black) and the zeros of $s_{22}(z)$ (red), the regions $D_{+}$ and $D_{-}$ where $Im\lambda>0$ (green) and $Im\lambda<0$ (white), respectively, and the orientation of the contours about the Riemann Hilbert problem .}\\

Finally, we talk about the asymptotic relationship between two planes $k$ and $z$.
When $k\in S_{1}$, we have
\begin{align*}
z&=k+\sqrt{k^{2}+q_{0}^{2}}\\
&=k+k(1+\frac{q_{0}^{2}}{k^{2}}+\cdots)^{1/2}\sim 2k+o(k^{-1}), \quad k\rightarrow\infty.\notag
\end{align*}
That means $z\rightarrow\infty$ as $k\rightarrow\infty$ $(k\in S_{1})$. Similarly, the other case, $z\rightarrow 0$ as $k\rightarrow\infty$ $(k\in S_{2})$, can be obtained.
\subsubsection{Jost function}
According to the asymptotic Lax pair \eqref{2.5}, we calculate that the matrix $X_{\pm}$ has two eigenvalues i.e., $\pm i\lambda$. Therefore, based on the relationship between $X_{\pm}$ and $T_{\pm}$ in \eqref{2.5} and the property of the matrix, the matrix $T_{\pm}$ has two eigenvalues i.e. $\pm (-8\tau k^{3}+2k+4\tau kq^{2}_{0})i\lambda$. Then, $X_{\pm}$ and $T_{\pm}$ can be transformed to diagonal matrix i.e.,
\begin{align}
\begin{split}
  X_{\pm}(x,t;z)&=Y_{\pm}(z)(-i\lambda \sigma_{3})Y_{\pm}^{-1}(z),\\
  T_{\pm}(x,t;z)&=Y_{\pm}(z)[-i\lambda(-8\tau k^{3}+2k+4\tau kq^{2}_{0})\sigma_{3}]Y_{\pm}^{-1}(z),\notag
  \end{split}
\end{align}
where
\begin{align}
Y_{\pm}(z)=\left(
  \begin{array}{cc}
     1 & -\frac{iq_{\pm}}{z} \\
     -\frac{iq_{\pm}^{*}}{z} & 1 \\
  \end{array}
\right)=\mathbb{I}-(i/z)\sigma_{3}Q_{\pm}. \notag
\end{align}
We know that all the values of $k (on~S_{1}, S_{2})$ , included in the continuous spectrum $\Sigma_{k}$, meet $\lambda(k)\in\mathbb{R}$ which means $\Sigma_{k}=\mathbb{R}\cup i[-q_{0},q_{0}]$. After the introduction of the uniformization variable, the $\Sigma_{k}$ is transformed into the $\Sigma_{z}=\mathbb{R}\cup C_{0}$. The subscript $z$ indicates that the set is in complex $z$-plane and $C_{0}$ is a circle with $0$ as the center and $q_{0}$ as the radius. For the simplicity of the following analysis, we omit the subscript i.e., $\Sigma_{z}\rightarrow\Sigma$. Then, according to the Lax pair \eqref{2.5}, the Jost solutions can be constructed and satisfy that
\begin{align}
\phi_{\pm}(x,t;z)\thicksim\psi_{\pm}(x,t;z)=Y_{\pm}(z)e^{-i\theta(x,t;z)\sigma_{3}},
\quad x\rightarrow \pm\infty \notag
\end{align}
where $\theta(x,t;z)=\lambda(z)[x+(-8\tau k^{3}+2k+4\tau kq^{2}_{0})t]$.
Introducing the $\mu(x,t;z)$ which satisfies that
\begin{align}\label{E-1}
\mu_{\pm}(x,t;z)=\phi_{\pm}(x,t;z)e^{i\theta(x,t;z)\sigma_{3}}  \sim  Y_{\pm}(z), \quad x\rightarrow\pm\infty,
\end{align}
the equivalent Lax pair can be obtained as
\begin{gather}\label{2.10}
(Y_{\pm}^{-1}(z)\mu_{\pm}(z))_{x}-i\lambda [Y_{\pm}^{-1}(z)\mu_{\pm}(z),\sigma_{3}]=Y_{\pm}^{-1}(z)\Delta Q_{\pm}(z)\mu_{\pm}(z),\notag \\
(Y_{\pm}^{-1}(z)\mu_{\pm}(z))_{t}-i\lambda(-8\tau k^{3}+2k+4\tau kq^{2}_{0})[Y_{\pm}^{-1}(z)\mu_{\pm}(z),\sigma_{3}]=Y_{\pm}^{-1}(z)\Delta T_{\pm}(z)\mu_{\pm}(z),
\end{gather}
where $\Delta Q_{\pm}(z)=Q-Q_{\pm}$ and $\Delta T_{\pm}(z)=T-T_{\pm}$. Through calculations, Lax pair \eqref{2.10} can be written in full derivative form. Therefore, based on the properties of the differential, we can select two special integration paths i.e., $(-\infty,t)\rightarrow(x,t)$ and $(+\infty,t)\rightarrow(x,t)$. Then, we get the following linear integral equations
\begin{align}\label{2.11}
\begin{matrix}
\mu_{-}(x,t;z)=Y_{-}+\int_{-\infty}^{x}Y_{-}e^{-i\lambda(x-y)\hat{\sigma}_{3}}[Y_{-}^{-1}\Delta Q_{-}(y,t)\mu_{-}(y,t;z)]\, dy,\\
\mu_{+}(x,t;z)=Y_{+}-\int_{x}^{\infty}Y_{+}e^{-i\lambda(x-y)\hat{\sigma}_{3}}[Y_{+}^{-1}\Delta Q_{+}(y,t)\mu_{+}(y,t;z)]\, dy.
\end{matrix}
\end{align}
Therefore, the analyticity of the function $\mu_{\pm}$ can be derived.

\noindent \textbf {Theorem 2.2}
\emph{The functions $\mu_{-,1}, \mu_{+,2}$ are analytic in $D_{+}$ and $\mu_{-,2}, \mu_{+,1}$ are analytic in $D_{-}$ and they can be recorded as $\mu^{+}_{-,1}, \mu^{+}_{+,2}, \mu^{-}_{-,2}, \mu^{-}_{+,1}$, respectively.
The functions $\mu_{\pm,j} (j=1,2)$ is the $j$-th column of $\mu_{\pm}$. For conveniently, $\mu_{+}$ and $\mu_{-}$ can be rewritten as}
\begin{align}
\mu_{+}=(\mu^{-}_{+,1}, \mu^{+}_{+,2}), \quad \mu_{-}=(\mu^{+}_{-,1}, \mu^{-}_{-,2}). \notag
\end{align}
\begin{proof}
By simple calculations, we can get that
\begin{align}
Y^{-1}_{-}\mu_{-,1}(y,t;z)=\left(                 
  \begin{array}{c}   
    1  \\  
    0  \\  
  \end{array}
\right)+\int_{-\infty}^{x}G_{0}(x-y,z)\Delta Q_{-}(y)\mu_{-,1}(y,t;z)\, dy \notag
\end{align}
where $G_{0}(x-y,z)=\frac{1}{\gamma}\left(                 
  \begin{array}{cc}   
    1 & \frac{iq_{-}}{z}  \\  
    \frac{iq^{*}_{-}}{z}e^{2i\lambda(x-y)} & e^{2i\lambda(x-y)} \\ 
  \end{array}
\right)$ with $\gamma=det(Y_{\pm})=1+\frac{q^{2}_{0}}{z}$.
Considering the  $e^{2i\lambda\left(x-y\right)}=e^{2i\left(x-y\right)Re\lambda}e^{-2\left(x-y\right)Im\lambda}$ , we can make the conclusion that $x-y>0$, so $\mu_{-,1}$ is analytic in $D_{+}$ when $Im\lambda>0$. Then, $\mu_{-,1}$ can be recorded as $\mu^{+}_{-,1}$ which means the first column of $\mu_{-}$ is analytic in $D_{+}$.
Similar to the above analysis, the analyticity of $\mu_{-,2}, \mu_{+,1}, \mu_{+,2}$ can be obtained.
Now we finish the proof.
\end{proof}
\subsubsection{Scattering matrix}
In this section, the scattering matrix will be discussed. Before the discussion, we, firstly, introduce a lemma.
\begin{lemma}
If $A(x)$, $Y(x)$ are $n$-order matrix matrices, satisfying $Y_{x}=AY$, then $(\det Y)_{x}=tr(A)\det Y$ and $\det Y(x)=[\det Y(x_{0})]e^{\int_{x_{0}}^{x}tr[A(y)]\,dy}$.\label{lemma1}
\end{lemma}
\begin{proof}
Introduce
\begin{gather}
A =\left(
      \begin{array}{cccc}
        a_{11} & a_{12} & \cdot\cdot\cdot & a_{1n} \\
        a_{21} & a_{22} & \cdot\cdot\cdot & a_{2n} \\
        \cdot\cdot\cdot & \cdot\cdot\cdot & \cdot\cdot\cdot & \cdot\cdot\cdot \\
        a_{n1} & a_{n2} & \cdot\cdot\cdot & a_{nn} \\
      \end{array}
    \right),  \qquad Y=\left(
                 \begin{array}{c}
                   Y_{1} \\
                   Y_{2} \\
                   \cdot\cdot\cdot \\
                   Y_{n} \\
                 \end{array}
               \right), \notag
\end{gather}
where $Y_{i} (i=1,2,\cdot\cdot\cdot,n)$ is the $i$-th row of matrix $Y$. From $Y_{x}=AY$, we obtain
\begin{gather}
Y_{i,x}=a_{i1}Y_{1}+a_{i2}Y_{2}+\cdot\cdot\cdot+a_{in}Y_{n}, \qquad i=1,2,\cdot\cdot\cdot,n. \notag
\end{gather}
Then, we denote
\begin{gather}
\left(\det Y\right)_{x}=\sum_{i=1}^{n}\det \left(
                 \begin{array}{c}
                   Y_{1} \\
                   \cdot\cdot\cdot \\
                   Y_{i,x} \\
                   \cdot\cdot\cdot \\
                   Y_{n} \\
                 \end{array}
               \right)=\sum_{i=1}^{n}\det \left(
                 \begin{array}{c}
                   Y_{1} \\
                   \cdot\cdot\cdot \\
                   a_{i1}Y_{1}+\cdot\cdot\cdot+a_{ii}Y_{i}+\cdot\cdot\cdot+a_{in}Y_{n} \\
                   \cdot\cdot\cdot \\
                   Y_{n} \\
                 \end{array}
               \right) \notag
\end{gather}
which means
\begin{gather}
\left(\det Y\right)_{x}=\sum_{i=1}^{n}a_{ii}\det Y=tr(A)\det Y. \notag
\end{gather}
Furthermore, through integral, we can obtain
\begin{gather}
\left(\det Y\right)=Ce^{\int_{x_{0}}^{x}tr[A(y)]\,dy}. \notag
\end{gather}
Selecting $x=x_{0}$, we can calculate that $C=\left(\det Y\right)_{x_{0}}$. So, we have
\begin{gather}
\left(\det Y\right)=\left(\det Y\right)_{x_{0}}e^{\int_{x_{0}}^{x}tr[A(y)]\,dy}. \notag
\end{gather}
Now, we finish the proof.
\end{proof}
According to the $\mathbf{Lemma 1}$ and $tr(X)=tr(T)=0$ in equation \eqref{2.4}, we denote that $(\det \phi)_{x}=(\det \phi)_{t}=0$. So, for any $z\in \Sigma$, we obtain that $\det\phi_{\pm}(x,t;z)=\det Y_{\pm}(z)=\gamma(z)$. Meanwhile, for any $z\in \Sigma_{0}$ defined as $\Sigma_{0}:=\Sigma\setminus{\pm iq_{0}}$, $\phi_{\pm}$ are two fundamental matrix solutions of scattering problem. Then, we can denote the relationship between the $\phi_{+}$ and $\phi_{-}$ as follows
\begin{align}\label{2.12}
 \phi_{+}(x,t;z)=\phi_{-}(x,t;z)S(z),
\end{align}
where $S(z)=(s_{ij})_{2\times2}$ is a constant matrix and is independent of the variable $x$ and $t$.
In addition, based on the equation \eqref{2.12}, we can get the following relation
\begin{align}
&s_{11}(z)=\frac{Wr\left(\phi_{+,1},\phi_{-,2}\right)}{\gamma},\label{2.13}\quad
s_{22}(z)=\frac{Wr\left(\phi_{-,1},\phi_{+,2}\right)}{\gamma},\\
&s_{12}(z)=\frac{Wr\left(\phi_{+,2},\phi_{-,2}\right)}{\gamma},\quad
s_{21}(z)=\frac{Wr\left(\phi_{-,1},\phi_{+,1}\right)}{\gamma},
\end{align}
where the subscript of $\phi_{\pm,j}$ mean the $j$-column of $\phi_{\pm}$ and $\gamma=1+\frac{q^{2}_{0}}{z}$.
Then, we discuss the analytical properties of the scattering matrix $S(z)$.

\noindent \textbf {Theorem 2.3}
\emph{The function $s_{11}$ is analytic in $D_{-}$ and $s_{22}$ is analytic in $D_{+}$. However, the functions $s_{12}$ and $s_{21}$ are nowhere analytic.}
\begin{proof}
With the equation \eqref{E-1}, we can derive that
\begin{align}\label{2.14}
\mu_{+}=\mu_{-}e^{-i\theta(z)\hat{\sigma}_{3}}S(z)
\end{align}
and
\begin{align}\label{2.15}
(\det \mu_{\pm})_{x}=(\det \phi_{\pm})_{x}=0.
\end{align}
The equation \eqref{2.15} implies that $\det \mu_{\pm}$ are independent of variable $x$. Thus, we have
\begin{align}
\det \mu_{\pm}=\det(\lim_{x\rightarrow\pm\infty}\mu_{\pm}=\det Y_{\pm})_{x}=\gamma\neq0, \notag
\end{align}
which means $\mu_{\pm}$ is reversible. Then, equation \eqref{2.14} can be transformed into $e^{-i\theta(z)\hat{\sigma}_{3}}S(z)=\mu_{-}^{-1}\mu_{+}$.
Then, based on the \emph{Theorem 2.2} and the expression of $S(z)$, the corresponding analytical properties of the $s(z)_{ij}(i,j=1,2)$ can be derived.
\end{proof}
Here, we introduce the reflection coefficients
\begin{align}\label{2.16}
\rho(z)=s_{21}(z)/s_{11}(z),\quad \tilde{\rho}(z)=s_{12}(z)/s_{22}(z),\quad \forall z\in\Sigma.
\end{align}
which will be important in the reverse problem.
\subsubsection{Symmetries}
Here, the symmetry of the scattering matrix will be analyzed. This property will play an important role in the following analysis and make the analysis process more elegant. Different from the case of zero boundary conditions(ZBCs), we have to deal with not only the map $k\rightarrow k^{*}$ but also the sheets of the Riemann surface. On one hand, $z\rightarrow z^{*}$($\in$ $z$-plane) implies $(k,\lambda)\rightarrow(k^{*},\lambda^{*})$($\in$ $k$-plane), on the other hand, $z\rightarrow -q_{0}^{2}/z$($\in$ $z$-plane) implies $(k,\lambda)\rightarrow(k,-\lambda)$($\in$ $k$-plane). The corresponding symmetries of the scattering problem are these two transformations. Then, we discuss the symmetries of the scattering problem.

Firstly, the symmetries of $\mu_{\pm}(x,t;z)$ which correspond to the Jost solutions $\phi_{\pm}(x,t;z)$ are given as follows.

\noindent \textbf {Theorem 2.4}
\emph{The two symmetries:}
\begin{align}
\mu_{\pm}(x,t;z)&=-\sigma\mu_{\pm}^{*}(x,t;z^{*})\sigma, \label{E-2}\\
\mu_{\pm}(x,t;z)&=-\frac{i}{z}\mu_{\pm}(x,t;-\frac{q_{0}^{2}}{z})\sigma_{3}Q_{\pm},\label{E-3}
\end{align}
\emph{where $\sigma=\left(
      \begin{array}{cc}
        0 & 1  \\
        -1 & 0  \\
      \end{array}
    \right)$.}
\emph{Furthermore, the \eqref{E-2} and \eqref{E-3} can be written as the form of the components of each column as follows}
\begin{align}
\mu_{\pm,1}(x,t;z)&=\sigma\mu_{\pm,2}^{*}(x,t;z^{*}),\quad \qquad
\mu_{\pm,2}(x,t;z)=-\sigma\mu_{\pm,1}^{*}(x,t;z^{*}),\label{E-4}\\
\mu_{\pm,1}(x,t;z)&=(-\frac{iq_{\pm}^{*}}{z})\mu_{\pm,2}(x,t;-\frac{q_{0}^{2}}{z}),\quad
\mu_{\pm,2}(x,t;z)=(-\frac{iq_{\pm}}{z})\mu_{\pm,1}(x,t;-\frac{q_{0}^{2}}{z}).\label{E-5}
\end{align}
\begin{proof}
Let $\omega(x,t;z)=\sigma\mu_{\pm}^{*}(x,t;z^{*})$. Notice that $\sigma\sigma=-I$, $\sigma\sigma_{3}\sigma=\sigma_{3}$, $\sigma Y_{\pm}^{*}(z^{*})\sigma=-Y_{\pm}^{-1}(z)$ and $\sigma\Delta Q_{\pm}^{*}\sigma=\Delta Q_{\pm}$, then, it is not hard to calculate that $\omega$ is the solution of the equation \eqref{2.11} when the $\mu_{\pm}(x,t;z)$ is the solution of the equation \eqref{2.11}. Based on the equation \eqref{2.10}, we can obtain that
\begin{align}
\omega_{\pm}(x,t;z)\sigma=-(Y_{\pm}(z)+o(1)) \notag
\end{align}
as $x\rightarrow\pm\infty$.
Since the solution of the scattering problem with given boundary conditions is unique, the $\omega_{\pm}\sigma=\mu_{\pm}$ is obtained. So, the \eqref{E-2} is proved.\\
Equation \eqref{E-3} can be proved similarly.
\end{proof}
Then, we discuss the symmetry of the scattering matrix $S(z)$.

\noindent \textbf {Theorem 2.5}
\emph{The symmetries of the scattering matrix $S(z)$ are expressed as follows}
\begin{align}
S(z)&=-\sigma S^{*}(z^{*})\sigma, \label{E-6}\\
S(z)&=(\sigma_{3}Q_{-})^{-1}S(-\frac{q_{0}^{2}}{z})\sigma_{3}Q_{+}.\label{E-7}
\end{align}
\emph{Furthermore, according to \eqref{E-6} and \eqref{E-7}, the relationship among the elements of $S(z)$ can be denoted as follows}
\begin{align}
s_{22}(z)=s^{*}_{11}(z^{*}), ~ s_{12}(z)=-s^{*}_{21}(z^{*}), \label{E-8}\\
s_{11}(z)=(q_{+}^{*}/q_{-}^{*})s_{22}(-q_{0}^{2}/z),\label{E-9}\\
s_{12}(z)=(q_{+}/q_{-}^{*})s_{21}(-q_{0}^{2}/z).\label{E-10}
\end{align}
\begin{proof}
Based on the \emph{Theorem 2.4} and $S(z)=e^{i\theta(z)\hat{\sigma}_{3}}\mu_{-}^{-1}(x,t;z)\mu_{+}(x,t;z)$, we can obtain that
\begin{align}
-\sigma S^{*}(z^{*})\sigma&=\sigma e^{-i\theta(z)\sigma_{3}}\sigma\sigma(\mu_{-}^{-1})^{*}(x,t;z^{*})\sigma\sigma
\mu^{*}_{+}(x,t;z^{*})\sigma\sigma e^{i\theta(z)\sigma_{3}}\sigma, \notag\\
&=e^{i\theta(z)\hat{\sigma}_{3}}\mu_{-}^{-1}(x,t;z)\mu_{+}(x,t;z)=S(z).\notag
\end{align}
So, the \eqref{E-6} is proved.
Equation \eqref{E-7} can be proved similarly.
\end{proof}
In addition, the relationship between the reflection coefficients can be denoted as
\begin{align}
\rho(z)=-\tilde{\rho}^{*}(z^{*})=(q^{*}_{-}/q_{-})\tilde{\rho}(-q_{0}^{2}/z).\notag
\end{align}
\subsubsection{Discrete spectrum and residue condition }
The discrete spectrum of the scattering problem is set that contains all values $z\in\mathbb{C}\setminus\Sigma$ which make the eigenfunctions exist in $L^{2}(\mathbb{R})$. These discrete spectrum are the values $z\in D_{-}$ and $z\in D_{+}$ such that $s_{11}(z)=0$ and $s_{22}(z)=0$, respectively. Here, we assume that $z_{n}(\in D_{-}\cap\{z\in\mathbb{C}: Imz<0\}, n=1,2,...,N)$ are the simple poles of $s_{11}(z)$ i.e. $s_{11}(z_{n})=0$ but $s'_{11}(z_{n})\neq0$, $n=1, 2,..., N$. Then, based on the \emph{Theorem 2.5} , we can acquire that
\begin{align}
s_{22}(z_{n}^{*})=s_{22}(-q_{0}^{2}/z_{n})=s_{11}(-q_{0}^{2}/z_{n}^{*})=0. \notag
\end{align}
So, the set of the discrete spectrum can be obtained as follows
\begin{align}
\mathbb{Z}=\left\{z_{n}, -\frac{q_{0}^{2}}{z_{n}^{*}},
  z_{n}^{*}, -\frac{q_{0}^{2}}{z_{n}}\right\},\quad s_{11}(z_{n})=0, \quad n=1,2,...,N. \notag
\end{align}

Then, we pay attention to the residue condition that will be useful in the inverse problem. According to the equation \eqref{2.13} and the nature of the determinant, the following relation can be derived as
\begin{align}\label{E-11}
\phi_{+,1}(z_{n})=b_{n}(z_{n})\phi_{-,2}(z_{n}),
\end{align}
where $b_{n}$ is a constant. Contacting the equation \eqref{E-1}, we can get
\begin{align}\label{E-12}
\mu_{+,1}(z_{n})=b_{n}(z_{n})e^{2i\theta(z_{n})}\mu_{-,2}(z_{n}).
\end{align}
Therefore, we have
\begin{align}\label{E-13}
\mathop{Res}_{z=z_{n}}\left[\frac{\mu_{+,1}(z)}{s_{11}(z)}\right]=
\frac{\mu_{+,1}(z_{n})}{s'_{11}(z_{n})}=\frac{b_{n}(z_{n})}{s'_{11}(z_{n})}
e^{2i\theta(z_{n})}\mu_{-,2}(z_{n}).
\end{align}
Through a similar analysis process, we derive that
\begin{align}\label{E-14}
\mathop{Res}_{z=z_{n}^{*}}\left[\frac{\mu_{+,2}(z)}{s_{22}(z)}\right]=
\frac{\mu_{+,2}(z_{n}^{*})}{s'_{22}(z_{n}^{*})}=
\frac{d_{n}(z_{n}^{*})}{s'_{22}(z_{n}^{*})}
e^{-2i\theta(z_{n}^{*})}\mu_{-,1}(z_{n}^{*}),
\end{align}
where $d_{n}$ is a constant and satisfies $\phi_{+,2}(z_{n}^{*})=d_{n}(z_{n}^{*})\phi_{-,1}(z_{n}^{*})$ which can be denoted as
\begin{align}\label{E-15}
\mu_{+,2}(z_{n}^{*})=d_{n}(z_{n}^{*})e^{-2i\theta(z_{n}^{*})}\mu_{-,1}(z_{n}^{*}).
\end{align}
According to the symmetries, the constants $b_{n}$ and $d_{n}$ have a fixed relationship. By applying the \emph{Theorem 2.4} to \eqref{E-12} and comparing with \eqref{E-15}, it is easy to acquire that
\begin{align}\label{E-16}
-b_{n}^{*}(z_{n})=d_{n}(z_{n}^{*}).
\end{align}
For the following analysis more convenient, we use the following transformation
\begin{align}
 C_{n}[z_{n}]=\frac{b_{n}(z_{n})}{s'_{11}(z_{n})},\quad
\tilde{C}_{n}[z_{n}]=\frac{d_{n}(z_{n}^{*})}{s'_{22}(z_{n}^{*})}. \notag
\end{align}
Thus, the relationship between $C_{n}[z_{n}]$ and $\tilde{C}_{n}[z_{n}]$ is that
\begin{align}\label{E-17}
-C^{*}_{n}[z_{n}]=\tilde{C}_{n}[z_{n}^{*}].
\end{align}
Finally, with the similar analysis, we can acquire the residue of the remaining two points of the eigenvalue quartet as follows
\begin{align}
\mathop{Res}_{z=-\frac{q_{0}^{2}}{z^{*}_{n}}}\left[\frac{\mu_{+,1}(z)}{s_{11}(z)}\right]&=
C_{N+n}e^{2i\theta(-\frac{q_{0}^{2}}{z^{*}_{n}})}\mu_{-,2}(-\frac{q_{0}^{2}}{z^{*}_{n}}), \notag \\
\mathop{Res}_{z=-\frac{q_{0}^{2}}{z_{n}}}\left[\frac{\mu_{+,2}(z)}{s_{22}(z)}\right]&=
\tilde{C}_{N+n}e^{-2i\theta(-\frac{q_{0}^{2}}{z_{n}})}\mu_{-,1}(-\frac{q_{0}^{2}}{z_{n}}), \notag
\end{align}
where $C_{N+n}=\frac{q_{-}^{*}d_{n}(z_{n}^{*})}{q_{+}s'_{11}(-q_{0}^{2}/z^{*}_{n})}$ and $\tilde{C}_{N+n}=\frac{q_{-}b_{n}(z_{n})}{q^{*}_{+}s'_{22}(-q_{0}^{2}/z_{n})}$.
\subsubsection{Analysis of asymptotic behavior}
In this subsection, we will analyse the asymptotic behaviors of the eigenfunction and the scattering matrix which will make an important impact on the construction of the Riemann Hilbert problem in inverse problem.

From the above analysis, we can see that when the limit $k\rightarrow\infty$, there are two cases  corresponding to it, namely, $z\rightarrow\infty$ and $z\rightarrow0$. Then, the asymptotic properties of the Jost function are given.

\noindent \textbf {Theorem 2.6}
\emph{The asymptotic properties are given as}
\begin{align}\label{E-18}
u_{\pm}(x,t;z)=
\left\{
\begin{aligned}
&I+\frac{i}{z}\sigma_{3}Q+O(z^{-2}),\quad\quad &z\rightarrow\infty,\\
&-\frac{i}{z}\sigma_{3}Q_{\pm}+O(1),\quad &z\rightarrow0.
\end{aligned}
\right.
\end{align}
\begin{proof}
Firstly, we verify the case $z\rightarrow\infty$. According to the Lax pair \eqref{2.10}, we let
\begin{align}\label{E-19}
Y^{-1}_{\pm}\mu_{\pm}=\chi_{\pm}^{(0)}+\chi_{\pm}^{(1)}/z+o(1/z)=
Y^{-1}_{\pm}(\mu_{\pm}^{(0)}+\mu_{\pm}^{(1)}/z+o(1/z)),\quad z\rightarrow\infty.
\end{align}
We know that $Y_{\pm}(z)=\mathbb{I}-(i/z)\sigma_{3}Q_{\pm}$. It is easy to calculate that $Y^{-1}_{\pm}(z)=\frac{1}{1+q_{0}^{2}/z^{2}}(\mathbb{I}+(i/z)\sigma_{3}Q_{\pm})$.
Then, from \eqref{E-19}, the relationship between $\chi_{\pm}^{(i)}$ and $\mu_{\pm}^{(i)}$ $(i=0,1,2)$ can be obtained. Furthermore, substituting \eqref{E-19} into the Lax pair \eqref{2.10}, comparing the same power coefficients of $z$, and combining with \eqref{E-1} and the relationship between $\chi_{\pm}^{(i)}$ and $\mu_{\pm}^{(i)}$, it is not hard to calculate that $\mu_{\pm}^{(0)}=\mathbb{I}$ and $\mu_{\pm}^{(1)}=i\sigma_{3}Q$. So, the asymptotic behavior of $\mu_{\pm}$ is obtained, i.e., $u_{\pm}(x,t;z)=I+\frac{i}{z}\sigma_{3}Q+O(z^{-2})$ as $z\rightarrow\infty$.\\
The other case $z\rightarrow0$ can be proved similarly.
\end{proof}
Meanwhile, the asymptotic behavior of scattering matrix is also needed in the following analysis.

\noindent \textbf {Theorem 2.7}
\emph{The asymptotic properties of the scattering matrix are given as}
\begin{align}
S(z)=
\left\{
\begin{aligned}
&I+O(1/z),\quad\quad\quad\quad &z\rightarrow\infty,\\
&diag(q_{-}/q_{+},q_{+}/q_{-})+O(z),\quad &z\rightarrow0.
\end{aligned}
\right.
\end{align}
\begin{proof}
Based on the \eqref{2.14} and \emph{Theorem 2.6}, the conclusion of \emph{Theorem 2.7} can be proved easily.
\end{proof}
\subsection{Inverse scattering problem}
\subsubsection{Generalized Riemann-Hilbert problem}
In this subsection, we are going to construct a generalized Riemann-Hilbert problem(RHP). From the above analysis and \eqref{2.12}, we know the relationship between eigenfunctions which are analytic in $D_{+}$ and analytic in $D_{-}$. Based on the \eqref{2.14}, we can calculate that
\begin{align}\label{E-22}
\mu_{+,1}(z)&=s_{11}(z)\mu_{-,1}(z)+s_{21}(z)e^{2i\theta(z)}\mu_{-,2}(z), \\
\mu_{+,2}(z)&=s_{12}(z)e^{-2i\theta(z)}\mu_{-,1}(z)+s_{22}(z)\mu_{-,2}(z).
\end{align}
According to the analytical properties of $\mu_{\pm}$ and scattering matrix, we introduce the sectionally  meromorphic matrices
\begin{align}\label{Matrix}
M(x,t;z)=\left\{\begin{aligned}
&M^{+}(x,t;z)=\left(\mu_{-,1}(x,t;z),\frac{\mu_{+,2}(x,t;z)}{s_{22}(z)}\right), \quad z\in D^{+},\\
&M^{-}(x,t;z)=\left(\frac{\mu_{+,1}(x,t;z)}{s_{11}(z)},\mu_{-,2}(x,t;z)\right), \quad z\in D^{-}.
\end{aligned}\right.
\end{align}
Here the superscripts $\pm$ mean analyticity in $D_{+}$ and $D_{-}$, respectively. From \eqref{E-22}, we can obtain the relationship between $M^{+}$ and $M^{-}$ as follows
\begin{align}\label{E-23}
M^{+}(x,t;z)=M^{-}(x,t;z)(\mathbb{I}-G(x,t;z)),
\end{align}
where the jump matrix is
\begin{align*}
G(x,t;z)=e^{i\theta(z)\hat{\sigma}_{3}}\left(
\begin{array}{ccc}
  0 & -\tilde{\rho}(z) \\
  \rho(z) & \rho(z)\tilde{\rho}(z)
\end{array} \right). \notag
\end{align*}
Based on the \emph{Theorem 2.6} and \emph{Theorem 2.7}, the asymptotic behavior of $M^{\pm}$ can be obtained easily, i.e.,
\begin{align}
M^{\pm}(x,t;z)=\left\{
\begin{aligned}
\mathbb{I}+O(1/z), \quad z\rightarrow\infty,\\
-\frac{i}{z}\sigma_{3}Q_{-}+O(1), z\rightarrow0.
\end{aligned}
\right.
\end{align}
From the above analysis and \eqref{E-13} and \eqref{E-14}, we obtain a generalized RHP.

\noindent \textbf {Theorem 2.8}
\emph{The generalized Riemann-Hilbert problem}
\begin{itemize}
  \item  $M(x,t;z)$ is meromorphic in $C\setminus\Sigma$;
  \item  $M^{+}(x,t;z)=M^{-}(x,t;z)(\mathbb{I}-G(x,t;z))$,~~~$z\in\Sigma$;
  \item  $M(x,t;z)$ satisfies residue conditions at zero points $\{z| s_{11}(z)=s_{22}(z)=0\}$;
  \item  $M^{\pm}(x,t;z)\thicksim\mathbb{I}+O(1/z)$, ~~$z\rightarrow\infty$;
  \item  $M^{\pm}(x,t;z)\thicksim-\frac{i}{z}\sigma_{3}Q_{-}+O(1)$,~~$z\rightarrow0$.
\end{itemize}
Here, we introduce a transformation to make the following analysis more convenient. The transformation are that
\begin{align}\label{E-24}
\xi_{n}=z_{n}, \quad \xi_{N+n}=-\frac{q_{0}^{2}}{z^{*}_{n}}, \xi^{*}_{n}=z^{*}_{n},\quad \xi^{*}_{N+n}=-\frac{q_{0}^{2}}{z_{n}}.
\end{align}
So, the $\xi_{n}, (n=1,2,\cdots, 2N)$ are the poles in $D_{-}$ and $\xi^{*}_{n}, (n=1,2,\cdots, 2N)$ are the poles in $D_{+}$.
Then, the residue conditions can be translated into
\begin{align}
\mathop{Res}_{z=\xi_{n}}\left[\frac{\mu_{+,1}(z)}{s_{11}(z)}\right]&=
\frac{\mu_{+,1}(\xi_{n})}{s'_{11}(\xi_{n})}=C_{n}[\xi_{n}]e^{2i\theta(\xi_{n})}\mu_{-,2}(\xi_{n}), \\
\mathop{Res}_{z=\xi^{*}_{n}}\left[\frac{\mu_{+,2}(z)}{s_{22}(z)}\right]&=
\frac{\mu_{+,2}(\xi^{*}_{n})}{s'_{22}(\xi^{*}_{n})}=\tilde{C}_{n}[\xi^{*}_{n}]e^{-2i\theta(\xi^{*}_{n})}\mu_{-,1}(\xi^{*}_{n}).
\end{align}
Furthermore, we have
\begin{align}\label{R-1}
\mathop{Res}_{z=\xi^{*}_{n}}M^{+}=(0,\tilde{C}_{n}[\xi^{*}_{n}]e^{-2i\theta(\xi^{*}_{n})}\mu_{-,1}(\xi^{*}_{n})),\quad
n=1,2,\cdots,2N,\notag \\
\mathop{Res}_{z=\xi_{n}}M^{-}=(C_{n}[\xi_{n}]e^{2i\theta(\xi_{n})}\mu_{-,2}(\xi_{n}),0), \quad
n=1,2,\cdots,2N.
\end{align}
Now, we need to solve the the Riemann-Hilbert problem. Firstly, via subtracting out the asymptotic behavior and the pole contributions, we obtain a regular RHP. Then, we have
\begin{align}\label{RR-1}
\begin{split}
&M^{+}(x,t;z)-\mathbb{I}+\frac{i}{z}\sigma_{3}Q_{-}-\sum_{n=1}^{2N}\frac
{\mathop{Res}_{z=\xi^{*}_{n}}M^{+}(z)}{z-\xi^{*}_{n}}-\sum_{n=1}^{2N}\frac
{\mathop{Res}_{z=\xi_{n}}M^{-}(z)}{z-\xi_{n}}\\=&
M^{-}(x,t;z)-\mathbb{I}+\frac{i}{z}\sigma_{3}Q_{-}-\sum_{n=1}^{2N}\frac
{\mathop{Res}_{z=\xi^{*}_{n}}M^{+}(z)}{z-\xi^{*}_{n}}-\sum_{n=1}^{2N}\frac
{\mathop{Res}_{z=\xi_{n}}M^{-}(z)}{z-\xi_{n}}-M^{-}(z)G(z).
\end{split}
\end{align}
From the above analysis, it is easy to know that the left side of \eqref{RR-1} is analytic in $D_{+}$ and the right side of \eqref{RR-1}, apart from the item $M^{-}(z)G(z)$, is analytic in $D_{-}$. Meanwhile, both sides of the equation \eqref{RR-1} have the asymptotic behavior that are $O(1/z)(z\rightarrow\infty)$ and $O(1)(z\rightarrow0)$. From the \emph{Theorem 2.7}, we can obtain the asymptotic behavior of $G(x,t;s)$, i.e., $O(1/z)(z\rightarrow\infty)$ and $O(1)(z\rightarrow0)$. Then, we introduce the Cauchy projectors $P_{\pm}$ over $\Sigma$ by
\begin{align}\label{Cauchy}
P_{\pm}[f](z)=\frac{1}{2i\pi}\int_{\Sigma}\frac{f(\zeta)}{\zeta-(z\pm i0)}\,d\zeta,
\end{align}
where the $\int_{\Sigma}$ implies the integral along the oriented contour shown in Fig. 1 and the $z\pm i0$ mean the limit is taken from the left/right of $z(z\in\Sigma)$. Using this Cauchy projectors, then, we can obtain the solution of the RHP as follows
\begin{align}\label{E-25}
\begin{split}
M(x,t;z)=&\mathbb{I}-\frac{i}{z}\sigma_{3}Q_{-}+\sum_{n=1}^{2N}\frac
{\mathop{Res}_{z=\xi^{*}_{n}}M^{+}(z)}{z-\xi^{*}_{n}}+\sum_{n=1}^{2N}\frac
{\mathop{Res}_{z=\xi_{n}}M^{-}(z)}{z-\xi_{n}}\\
&+\frac{1}{2i\pi}\int_{\Sigma}\frac{M(x,t;s)^{-}G(x,t;s)}{s-z}\,ds,\quad
z\in\mathbb{C}\setminus\Sigma,
\end{split}
\end{align}
where the $\int_{\Sigma}$ implies the contour shown in Fig. 1.
\subsubsection{Reconstruct the formula for potential}
To reduce to an algebraic integral system, the expression of the residue which emerge in \eqref{E-25} is needed. Coincidentally, the explicit $\mathop{Res}_{z=\xi_{n}}M^{+}$ and $\mathop{Res}_{z=\xi^{*}_{n}}M^{-}$ have been shown in \eqref{R-1}. We, therefore, evaluate the second column of the \eqref{E-25} at $z=\xi_{n}$ in $D_{-}$ and  get
\begin{align}\label{E-26}
u_{-,2}(x,t;\xi_{n})=\left(\begin{array}{cc}
                       -iq_{-}/\xi_{n} \\
                        1
                     \end{array}\right)
+\left(\sum_{k=1}^{2N}\frac{\mathop{Res}_{z=\xi^{*}_{k}}M^{+}(z)}{\xi_{n}-\xi^{*}_{k}}\right)_{2}
+\frac{1}{2i\pi}\int_{\Sigma}\frac{(M^{-}G)_{2}(x,t;\xi)}{s-\xi_{n}}\,ds \notag \\
=\left(\begin{array}{cc}
                       -iq_{-}/\xi_{n} \\
                        1
                     \end{array}\right)+\sum_{k=1}^{2N}
\frac{\tilde{C}_{k}[\xi^{*}_{k}]e^{-2i\theta(\xi^{*}_{k})}}{\xi_{n}-\xi^{*}_{k}}
u_{-,1}(x,t;\xi^{*}_{k})
+\frac{1}{2i\pi}\int_{\Sigma}\frac{(M^{-}G)_{2}(x,t;\xi)}{s-\xi_{n}}\,ds,
\end{align}
for $n=1,2,\cdots,2N$. Similarly, evaluating the first column of the \eqref{E-25} at $z=\xi^{*}_{n}$ in $D_{+}$ , we have
\begin{align}\label{E-27}
u_{-,1}(x,t;\xi^{*}_{n})=\left(\begin{array}{cc}
                            1 \\
                        -iq^{*}_{-}/\xi^{*}_{n}
                     \end{array}\right)+\sum_{j=1}^{2N}
\frac{C_{j}[\xi_{j}]e^{2i\theta(\xi_{j})}}{\xi_{n}^{*}-\xi_{j}}
u_{-,2}(x,t;\xi_{j})
+\frac{1}{2i\pi}\int_{\Sigma}\frac{(M^{-}G)_{1}(x,t;\xi)}{s-\xi^{*}_{n}}\,ds,
\end{align}
for $n=1,2,\cdots,2N$. The $(M^{-}G)_{j}$ means the $j-th$ column of $(M^{-}G)$. Additionally, a closed algebraic integral system for the solution of the RHP can be obtained by evaluating the $M^{-}$ through \eqref{E-25} with the equation \eqref{E-26} and \eqref{E-27}. Then, considering the \eqref{E-25}, the asymptotic behavior can be obtained as
\begin{align}\label{E-28}
M(x,t;z)=&\mathbb{I}+\frac{1}{z}\{-i\sigma_{3}Q_{-}+\sum_{n=1}^{2N}\mathop{Res}_{z=\xi^{*}_{n}}M^{+}(z)
+\sum_{n=1}^{2N}\mathop{Res}_{z=\xi_{n}}M^{-}(z) \notag \\
&-\frac{1}{2i\pi}\int_{\Sigma}M^{-}(x,t;s)G(x,t;s)\,ds\}+O(z^{-2}),\quad
z\rightarrow\infty.
\end{align}
Finally, through taking $M=M^{-}$ and combining the $1,2$ element of \eqref{E-28} and the \emph{Theorem 2.6}, the reconstruction formula for the potential can be acquired as
\begin{align}\label{E-29}
q(x,t)=-q_{-}-i\sum_{n=1}^{2N}\tilde{C}_{n}[\xi^{*}_{n}]e^{-2i\theta(x,t;\xi^{*}_{n})}\mu_{-,11}(x,t;\xi^{*}_{n})
+\frac{1}{2\pi}\int_{\Sigma}(M^{-}(x,t;s)G(x,t;s))_{12}\,ds.
\end{align}
\subsubsection{Trace formulate and theta condition}
Based on the \emph{Theorem 2.3}, we know the analytic properties of $s_{11}$ and $s_{22}$. The discrete spectrum we have analysed in the subsection \emph{2.1.5} are that $z_{n}, -\frac{q_{0}^{2}}{z_{n}^{*}}, z_{n}^{*}, -\frac{q_{0}^{2}}{z_{n}}, n=1,2,\cdots,N.$ Now, we construct the following function
\begin{align}\label{TT-1}
\zeta^{-}_{1}(z)=s_{11}(z)\prod_{n=1}^{N}\frac{(z-z_{n}^{*})(z+q_{0}^{2}/z_{n})}
{(z-z_{n})(z+q_{0}^{2}/z_{n}^{*})},\notag\\
\zeta^{+}_{1}(z)=s_{22}(z)\prod_{n=1}^{N}\frac{(z-z_{n})(z+q_{0}^{2}/z_{n}^{*})}
{(z-z_{n}^{*})(z+q_{0}^{2}/z_{n})}.
\end{align}
Therefore, the $\zeta^{-}_{1}$ and $\zeta^{+}_{1}$ are analytic in $D_{\mp}$ corresponding to the analytic properties of $s_{11}$ and $s_{22}$ and have no zeros. Meanwhile, according to the asymptotic behavior of $S(z)$ in \emph{Theorem 2.7}, we obtain that $\zeta^{\mp}_{1}(z)\rightarrow1$ as $z\rightarrow\infty$. Then, considering that $\det S(z)=1$ and the expression of the reflection coefficients, we obtain that
\begin{align}\label{TT-2}
\zeta^{-}_{1}(z)\zeta_{1}^{+}(z)=\frac{1}{1-\rho(z)\tilde{\rho}(z)},\quad z\in\Sigma.
\end{align}
Furthermore, via taking the logarithm of the above formula and applying the Plemelj's formulae and Cauchy projectors, we obtain that
\begin{align}\label{TT-3}
\log\zeta^{\mp}_{1}(z)=\pm\frac{1}{2\pi i}\int_{\Sigma}
\frac{\log[1-\rho(s)\tilde{\rho}(s)]}{s-z}\,ds,\quad z\in D_{\mp}.
\end{align}
Then, the trace formula can be obtained by substituting the \eqref{TT-3} into \eqref{TT-1} as
\begin{align}
s_{11}(z)&=exp\left(\frac{1}{2\pi i}\int_{\Sigma}\frac{\log[1-\rho(s)\tilde{\rho}(s)]}{s-z}
\,ds\right)\prod_{n=1}^{N}\frac{(z-z_{n})(z+q_{0}^{2}/z_{n}^{*})}
{(z-z_{n}^{*})(z+q_{0}^{2}/z_{n})},\label{TT-4}\\
s_{22}(z)&=exp\left(-\frac{1}{2\pi i}\int_{\Sigma}\frac{\log[1-\rho(s)\tilde{\rho}(s)]}{s-z}
\,ds\right)\prod_{n=1}^{N}\frac{(z-z_{n}^{*})(z+q_{0}^{2}/z_{n})}
{(z-z_{n})(z+q_{0}^{2}/z_{n}^{*})}.\label{TT-5}
\end{align}
Based on the \emph{Theorem 2.7}, we know that $s_{11}(z)\rightarrow q_{-}/q_{+}$ as $z\rightarrow0$. At the same time, since
\begin{align}
\prod_{n=1}^{N}\frac{(z-z_{n}^{*})(z+q_{0}^{2}/z_{n})}
{(z-z_{n})(z+q_{0}^{2}/z_{n}^{*})}\rightarrow 1, \quad as \quad z\rightarrow0, \notag
\end{align}
we have the theta condition
\begin{align}
\arg\frac{q_{-}}{q_{+}}=\frac{1}{2\pi}\int_{\Sigma}
\frac{\log[1-\rho(s)\tilde{\rho}(s)]}{s}\,ds+4\sum_{n=1}^{N}\arg z_{n}.
\end{align}

\subsubsection{Reflection-less potentials}
In this subsection, it is interesting to study a type solutions which are that the reflection coefficients $\rho(z)$ and $\tilde{\rho}(z)$ disappear. So, the jump matrix from $M^{-}$ to $M^{+}$ also vanishes i.e., $G(x,t;z)=0$. Under this conditions, we can acquire that
\begin{align}\label{E-30}
u_{-,12}(x,t;\xi_{j})
=-\frac{iq_{-}}{\xi_{j}}+\sum_{k=1}^{2N}
\frac{\tilde{C}_{k}[\xi^{*}_{k}]e^{-2i\theta(\xi^{*}_{k})}}{\xi_{j}-\xi^{*}_{k}}
u_{-,11}(x,t;\xi^{*}_{k}), \notag\\
u_{-,11}(x,t;\xi^{*}_{n})=1+\sum_{j=1}^{2N}
\frac{C_{j}[\xi_{j}]e^{2i\theta(\xi_{j})}}{\xi_{n}^{*}-\xi_{j}}
u_{-,12}(x,t;\xi_{j}).
\end{align}
from the \eqref{E-26} and \eqref{E-27}. For convenience, we introduce a transformation
\begin{align}\label{E-31}
c_{j}(x,t;z)=\frac{C_{j}[\xi_{j}]e^{2i\theta(\xi_{j})}}{z-\xi_{j}},\quad j=1,2,\cdots,2N.
\end{align}
Then, comprehensive \eqref{E-30} and \eqref{E-31}, we can get the following formula
\begin{align}\label{E-32}
u_{-,11}(x,t;\xi^{*}_{n})=1-iq_{-}\sum_{j=1}^{2N}\frac{c_{j}(\xi^{*}_{n})}{\xi_{j}}-
\sum_{j=1}^{2N}\sum_{k=1}^{2N}c_{j}(\xi_{n}^{*})c_{k}^{*}(\xi_{j}^{*})\mu_{-,11}(x,t;\xi_{k}^{*}).
\end{align}
In order to better express the solution, we introduce that
\begin{gather}
X_{n}=\mu_{-,11}(x,t;\xi_{n}),\quad X=(X_{1},\cdots,X_{2N})^{T}, \notag \\
U_{n}=1-iq_{+}\sum_{j=1}^{2N}\frac{c_{j}(\xi^{*}_{n})}{\xi_{j}}, \quad U=(U_{1},\cdots,U_{2N})^{T},\notag\\
P=(P_{n,k})_{2N\times2N},\quad P_{n,k}=\sum_{j=1}^{2N}c_{j}(\xi_{n}^{*})c_{k}^{*}(\xi_{j}^{*}),\notag\\
n,k=1,2,\cdots,2N, \notag
\end{gather}
where the superscript $T$ implies transposition.
Then, letting $M=\mathbb{I}+P$, we have $MX=U$. So, the solution for the potential can be expressed as
\begin{align}\label{E-33}
q(x,t)=-q_{-}+i\frac{\det M^{\sharp}}{\det M},
\end{align}
where $\det M^{\sharp}=\left(\begin{array}{cc}
                      0 & \Upsilon \\
                      U & M
                    \end{array}\right)$,
$\Upsilon=(\Upsilon_{1},\cdots,\Upsilon_{2N})$ and $\Upsilon_{n}=\tilde{C}_{n}[\xi^{*}_{n}]e^{-2i\theta(x,t;\xi^{*}_{n})}$ $(n=1,2,\cdots,2N)$.
\subsection{Soliton solutions}
In this section, for further studying the dynamic behavior of the soliton solutions, we make a graphical analysis. Based on the equation \eqref{E-33} and selecting appropriate parameters, some figures are illustrated. Firstly, we consider the case that the dimensionless parameter $\epsilon$ is zero. Therefore, the HDNLS equation is reduced to the classical nonlinear Schr\"{o}dinger equation. Then, when $N=1$, we apply the appropriate parameters and get the following images.\\

{\rotatebox{0}{\includegraphics[width=3.6cm,height=3.0cm,angle=0]{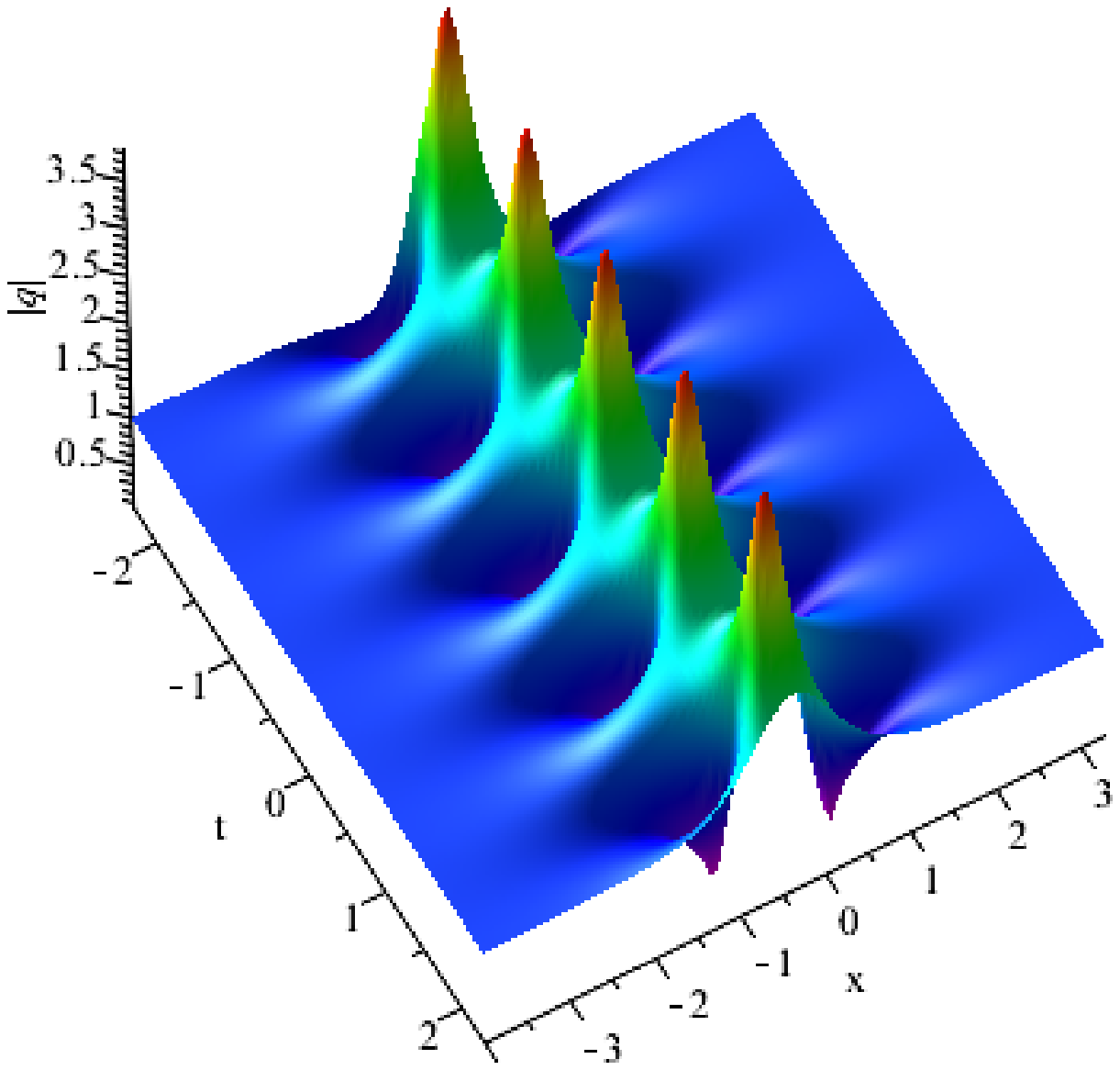}}}
~~~~
{\rotatebox{0}{\includegraphics[width=3.6cm,height=3.0cm,angle=0]{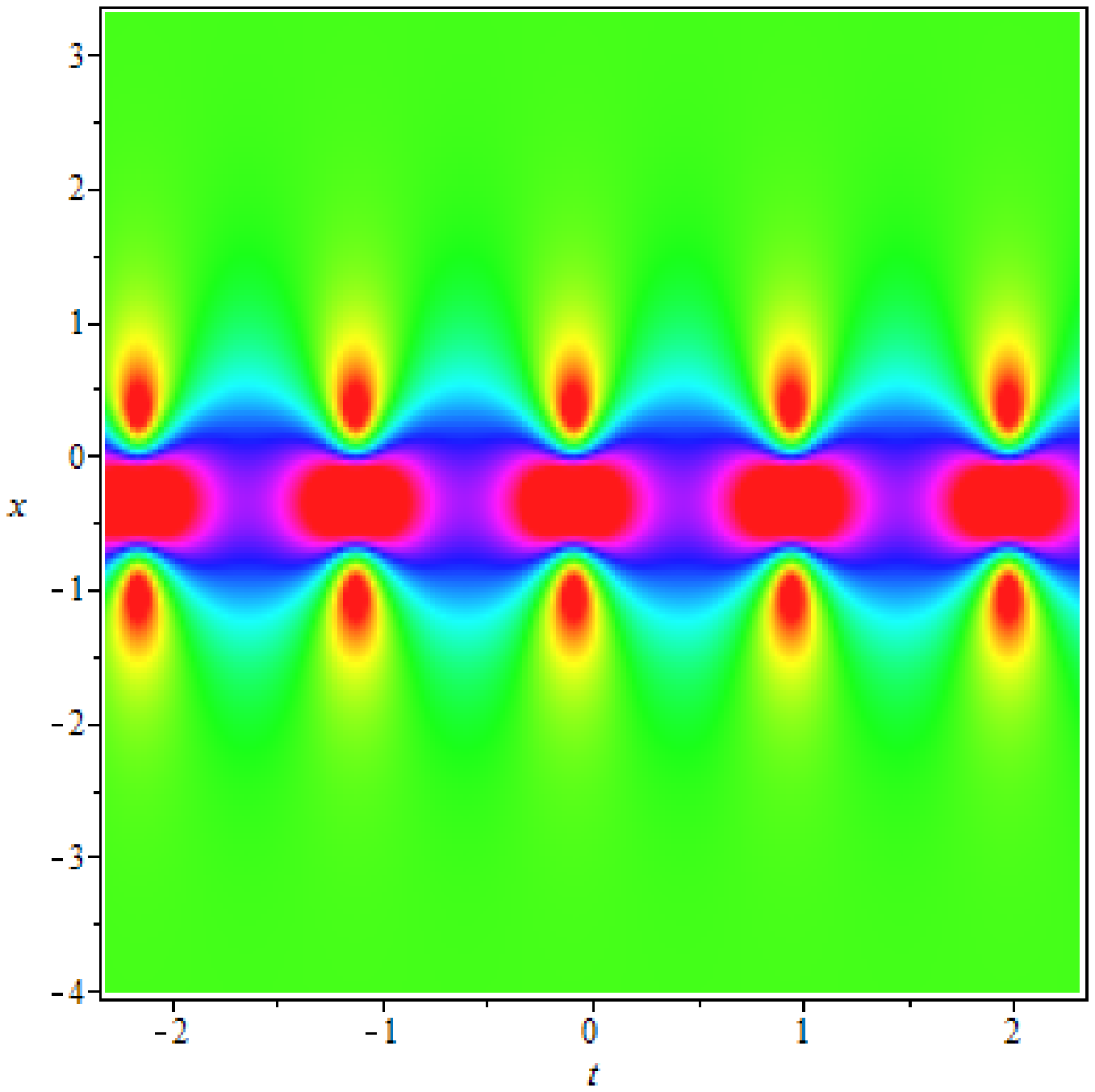}}}
~~~~
{\rotatebox{0}{\includegraphics[width=3.6cm,height=3.0cm,angle=0]{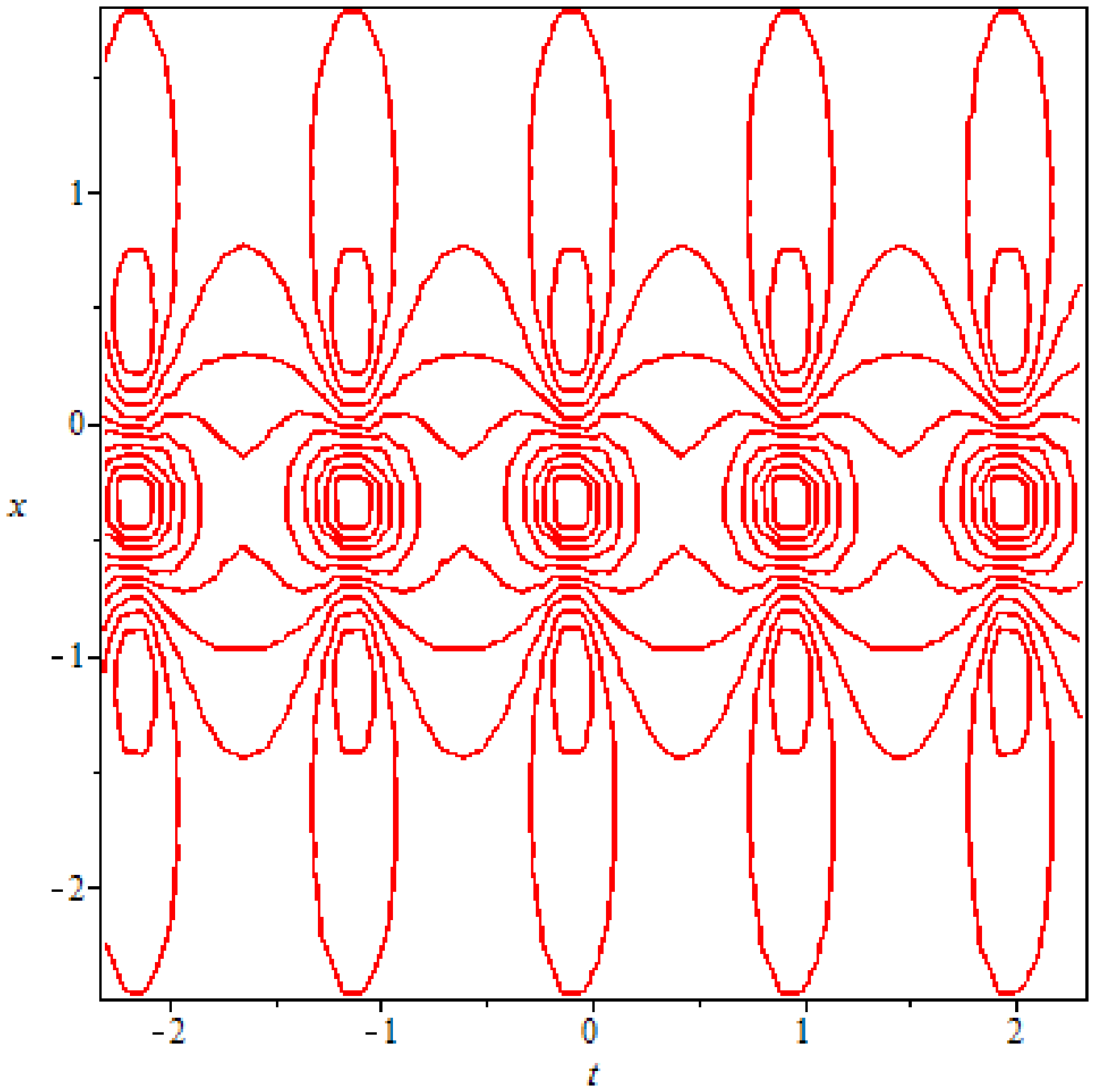}}}

$\qquad~~~~~~~~~(\textbf{a})\qquad \ \qquad\qquad\qquad\qquad~~~(\textbf{b})
\qquad\qquad\qquad\qquad\qquad~(\textbf{c})$\\
\noindent { \small \textbf{Figure 2.} (Color online) Plots of the breather solution of the equation  with the parameters $\epsilon=0$, $q_{-}=1$, $\xi_{1}=-2.5i$ and $b_{1}=e^{2+i}$.
$\textbf{(a)}$: the breather solution ,
$\textbf{(b)}$: the density plot ,
$\textbf{(c)}$: the contour line of the breather solution.} \\

From the illustrations, it is easy to find that the propagation of the solution is parallel to the time axis. This solution is called the stationary breather solution. Then, we change the parameter $\xi$ and obtain the following graphs.\\

{\rotatebox{0}{\includegraphics[width=3.6cm,height=3.0cm,angle=0]{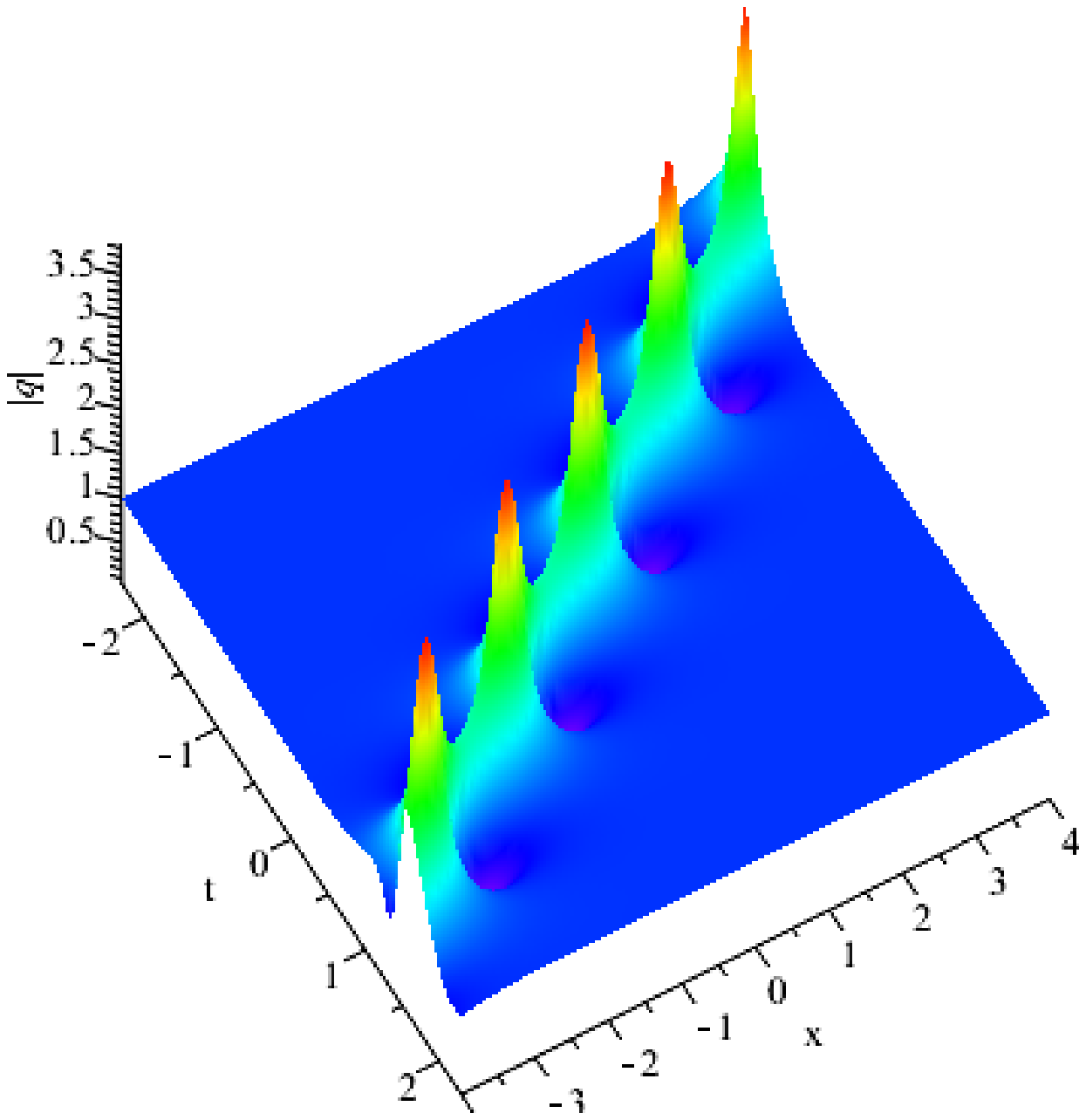}}}
~~~~
{\rotatebox{0}{\includegraphics[width=3.6cm,height=3.0cm,angle=0]{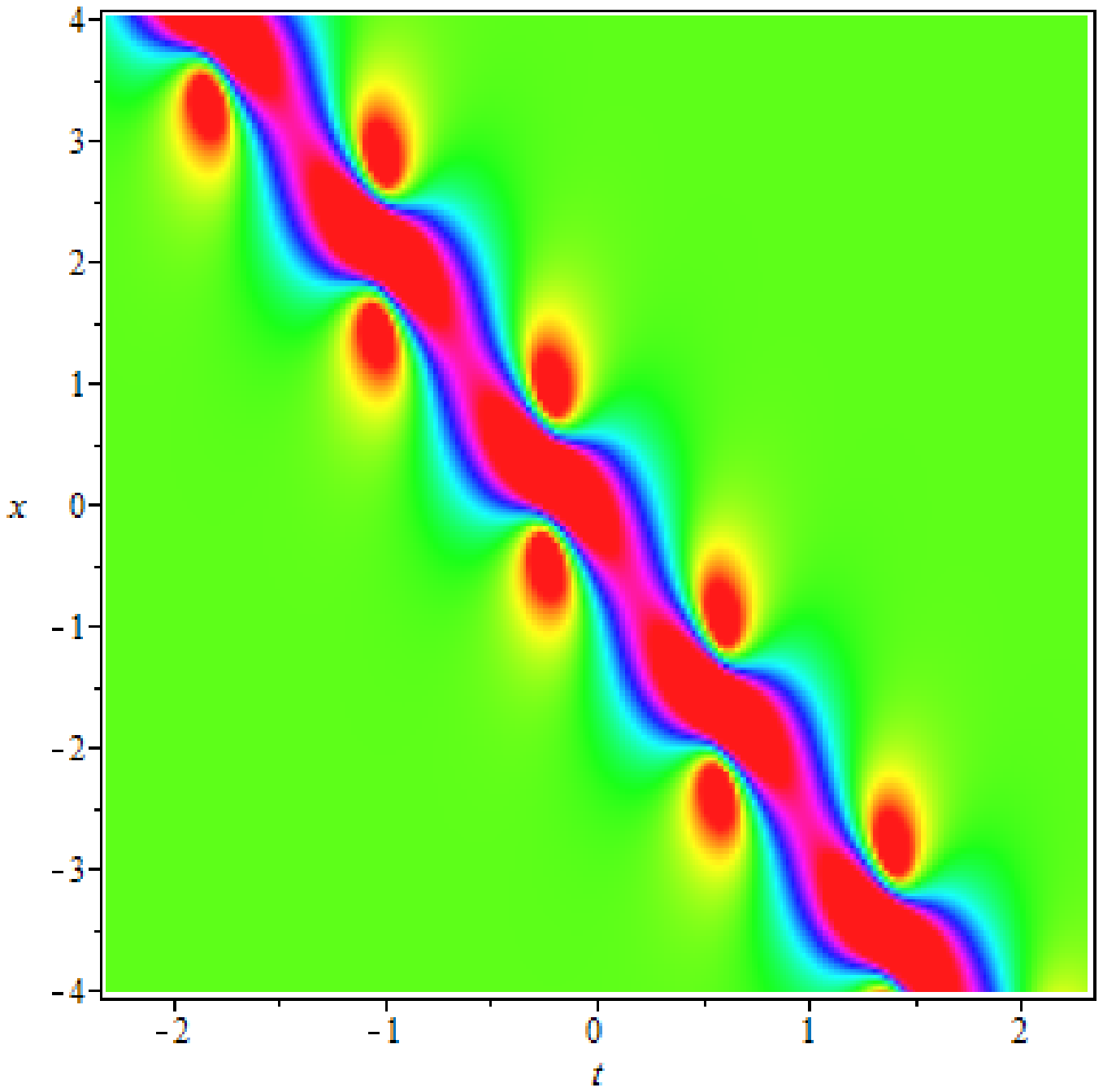}}}
~~~~
{\rotatebox{0}{\includegraphics[width=3.6cm,height=3.0cm,angle=0]{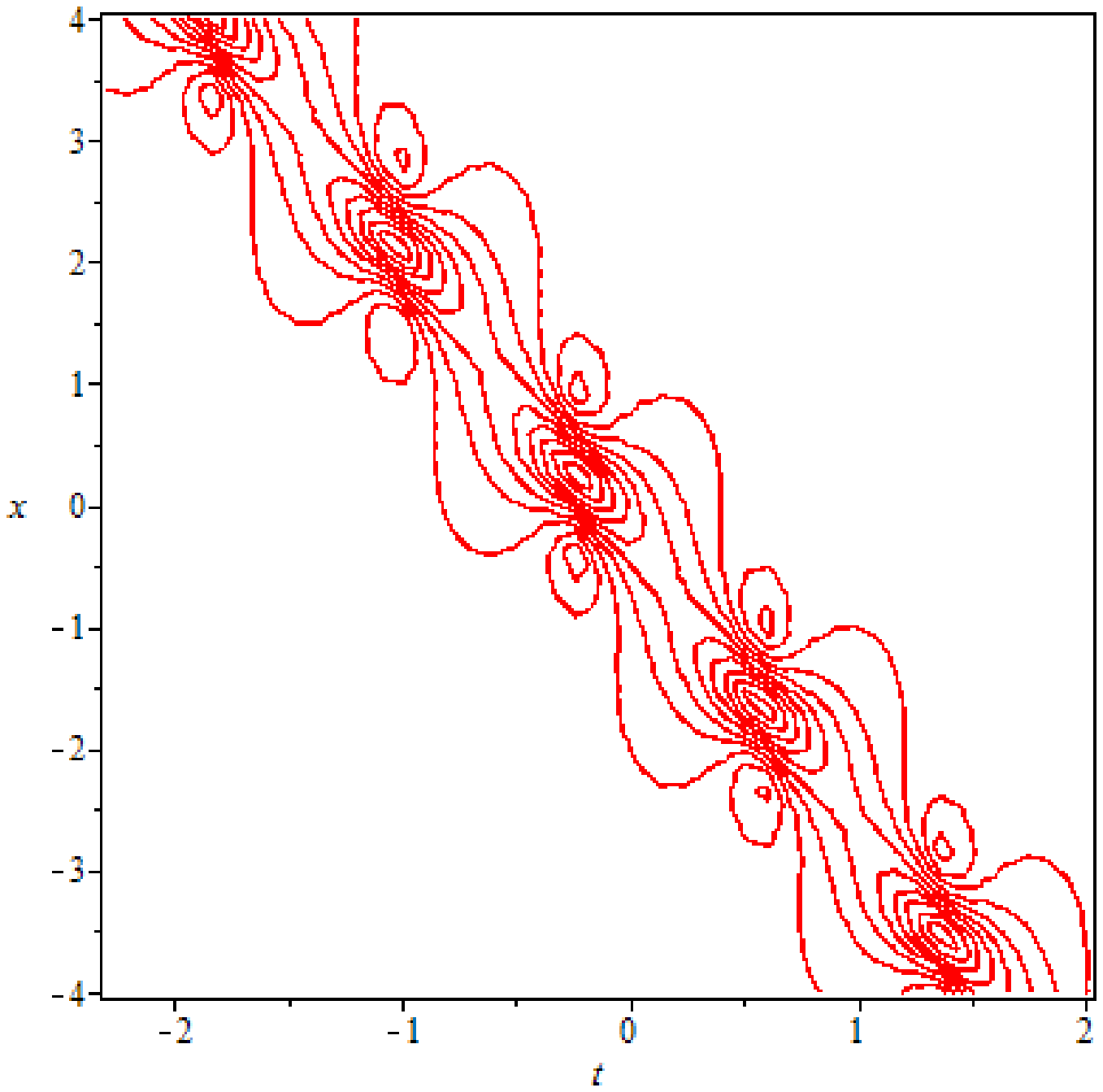}}}

$\qquad~~~~~~~~~(\textbf{a})\qquad \ \qquad\qquad\qquad\qquad~~~(\textbf{b})
\qquad\qquad\qquad\qquad\qquad~(\textbf{c})$\\
\noindent { \small \textbf{Figure 3.} (Color online) Plots of the breather solution of the equation  with the parameters $\epsilon=0$, $q_{-}=1$, $\xi_{1}=1-2.5i$ and $b_{1}=e^{2+i}$.
$\textbf{(a)}$: the breather solution,
$\textbf{(b)}$: the density plot ,
$\textbf{(c)}$: the contour line of the breather solution.}\\

It is obvious that the propagation of the solution is parallel to neither the $x$-axis nor the time axis. This solution can be called the non-stationary breather soliton solution. By comparing Figure 2. and Figure 3, we can find that when the discrete spectrum $\xi_{n}$ are pure imaginary, the stationary breather solution can be obtained. While, when the discrete spectrum $\xi_{n}$ have real parts, the non-stationary breather soliton solution is generated.

Furthermore, we change the boundary value $q_{-}$ and get the following graphs.\\

{\rotatebox{0}{\includegraphics[width=3.6cm,height=3.0cm,angle=0]{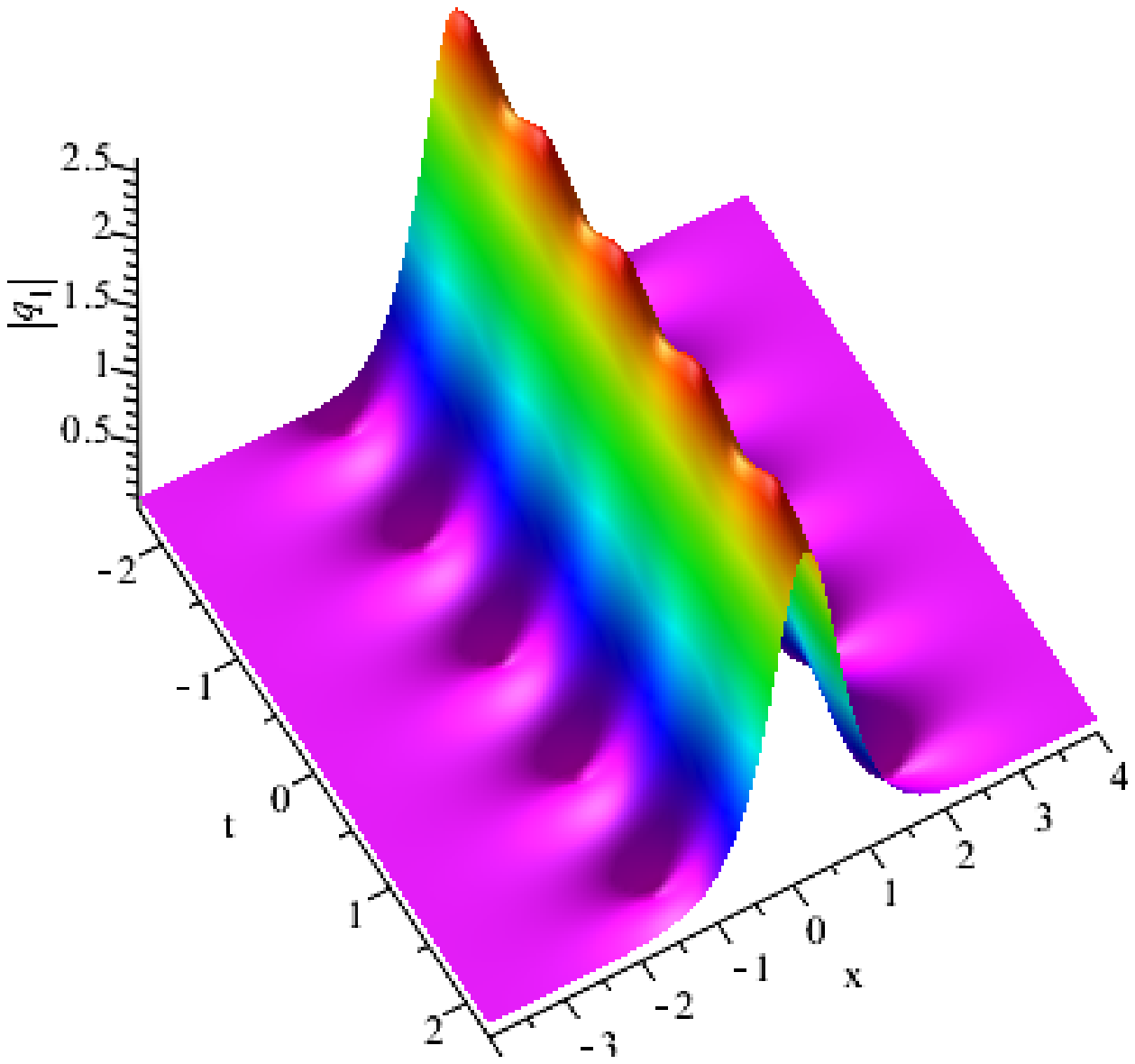}}}
~~~~
{\rotatebox{0}{\includegraphics[width=3.6cm,height=3.0cm,angle=0]{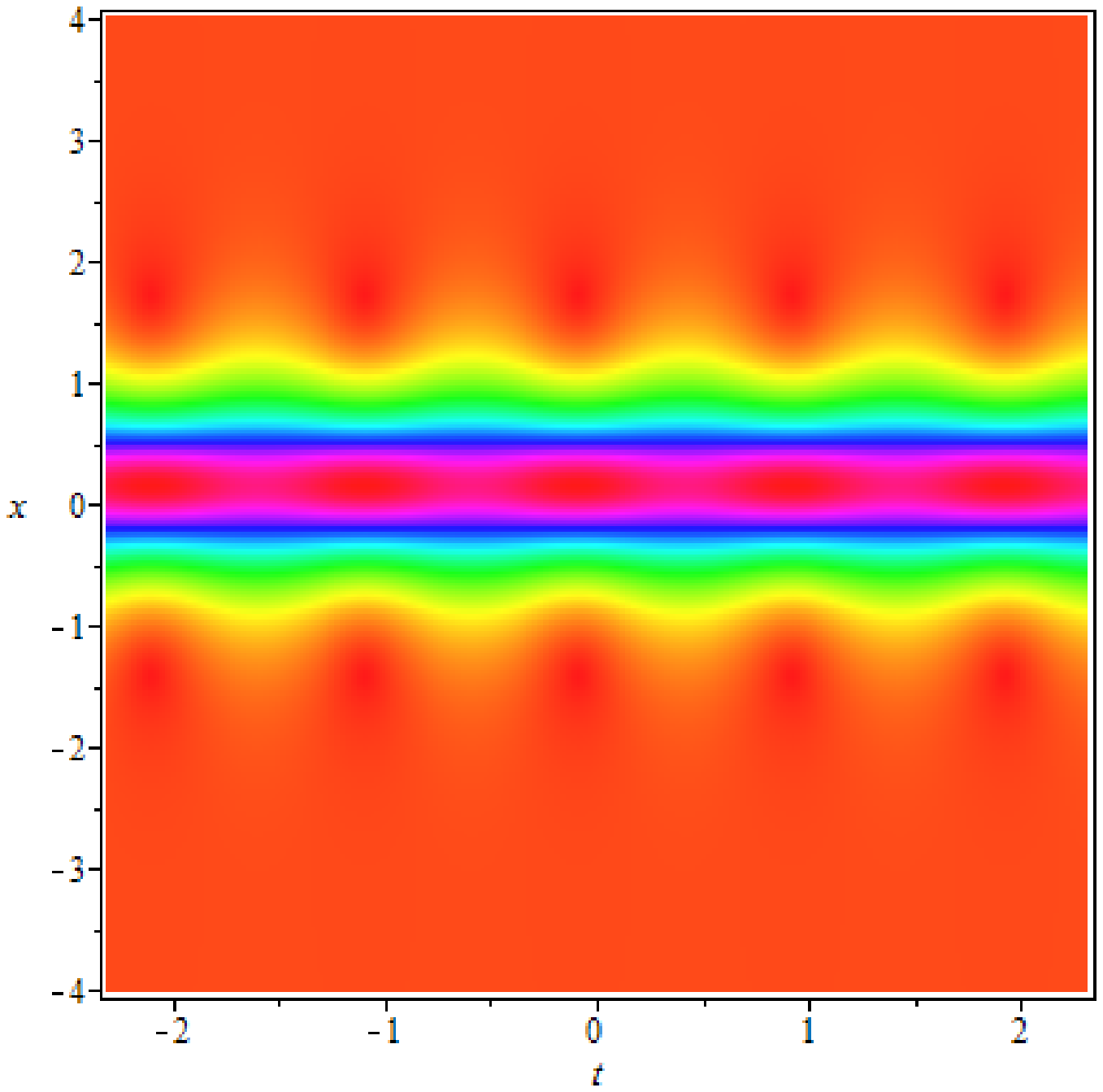}}}
~~~~
{\rotatebox{0}{\includegraphics[width=3.6cm,height=3.0cm,angle=0]{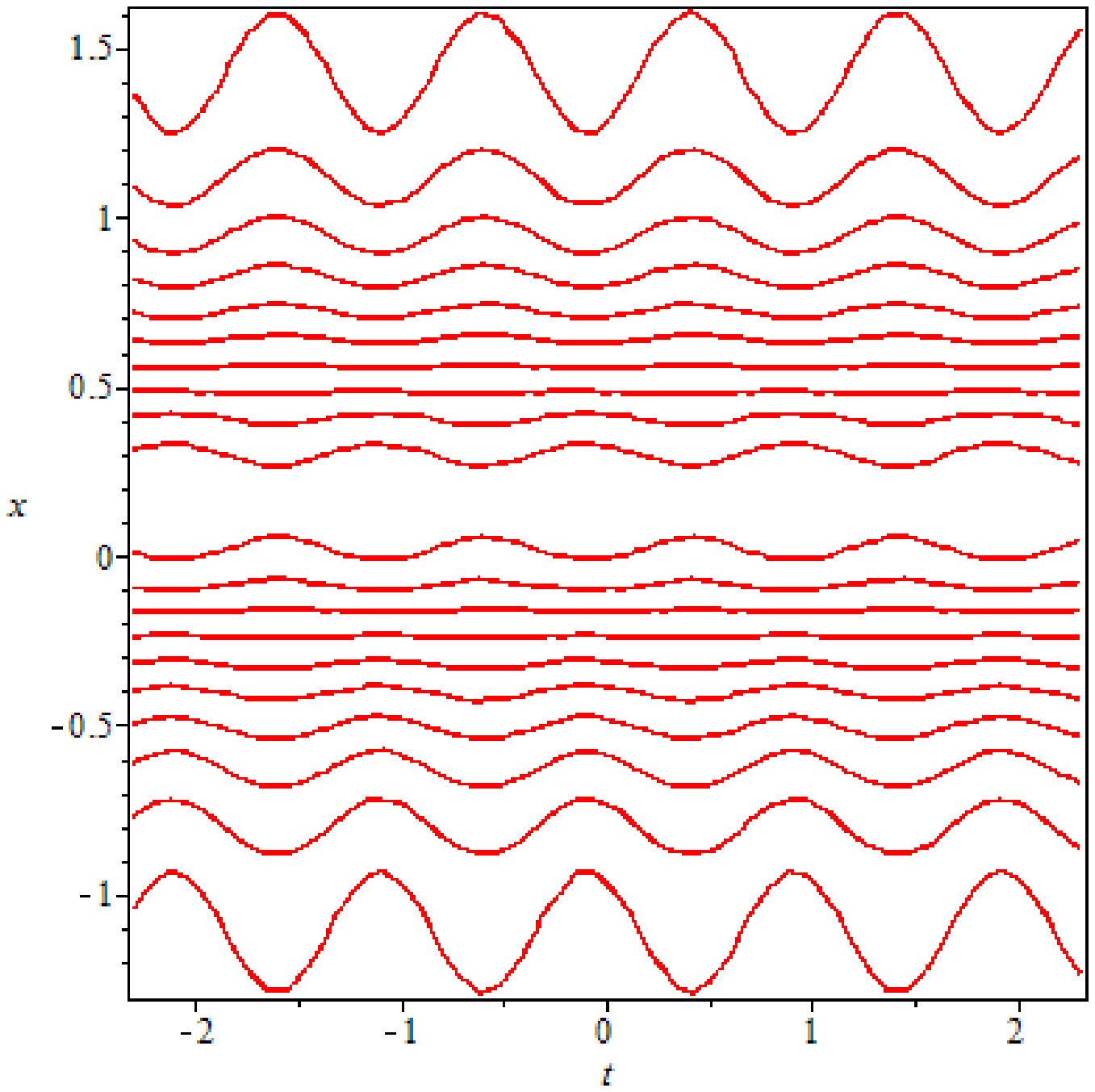}}}

$\qquad~~~~~~~~~(\textbf{a})\qquad \ \qquad\qquad\qquad\qquad~~~(\textbf{b})
\qquad\qquad\qquad\qquad\qquad~(\textbf{c})$\\

{\rotatebox{0}{\includegraphics[width=3.6cm,height=3.0cm,angle=0]{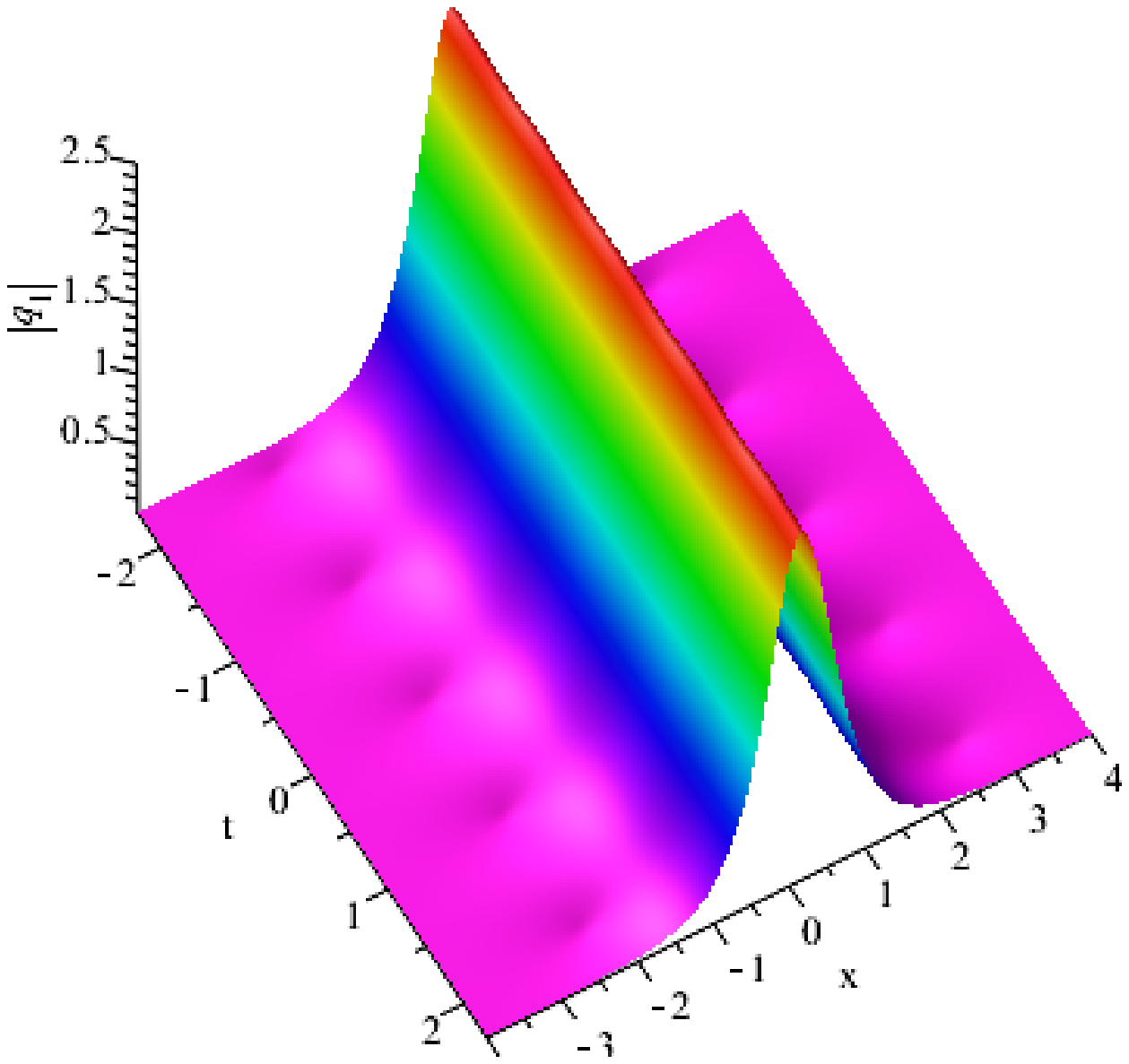}}}
~~~~
{\rotatebox{0}{\includegraphics[width=3.6cm,height=3.0cm,angle=0]{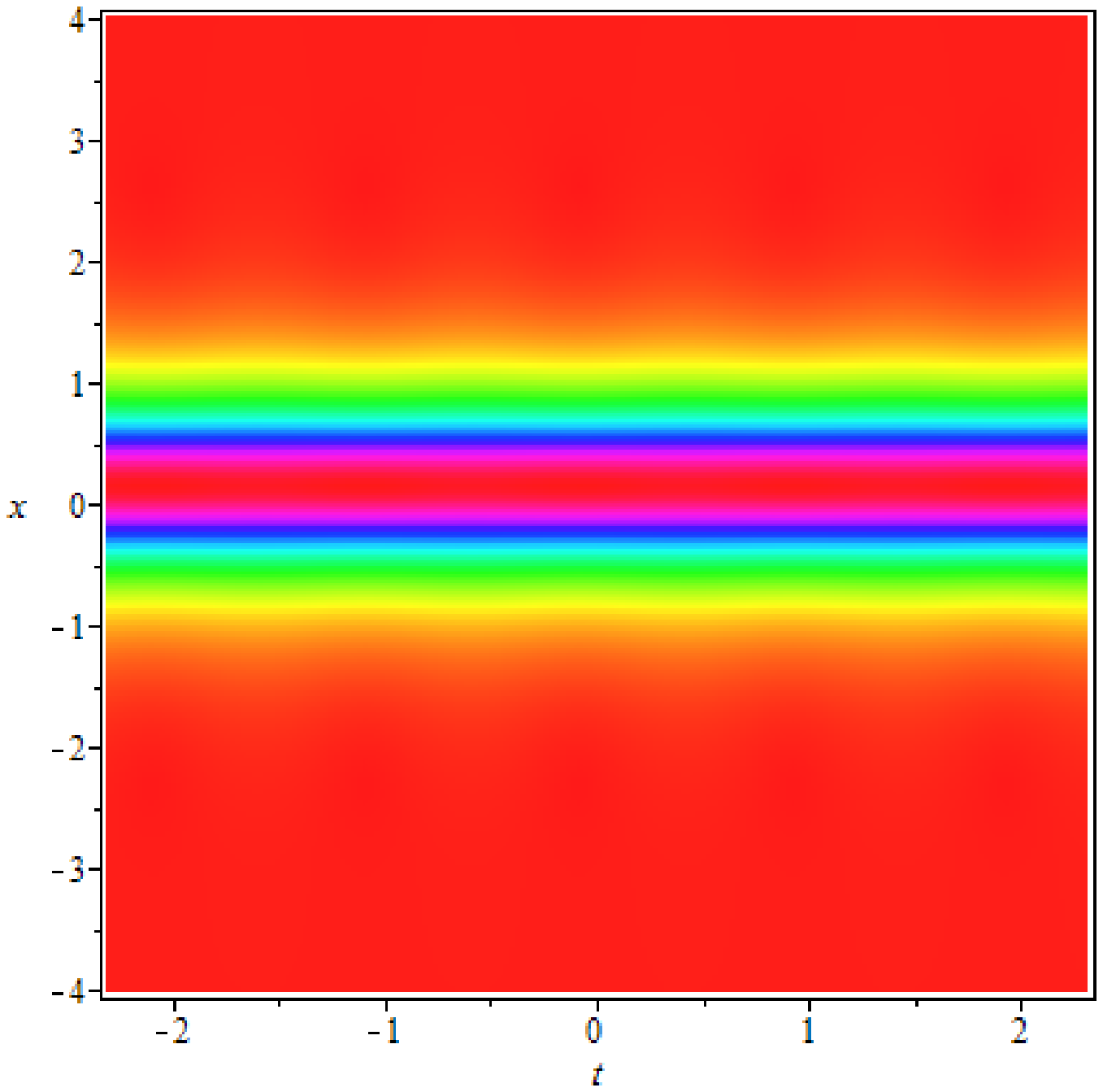}}}
~~~~
{\rotatebox{0}{\includegraphics[width=3.6cm,height=3.0cm,angle=0]{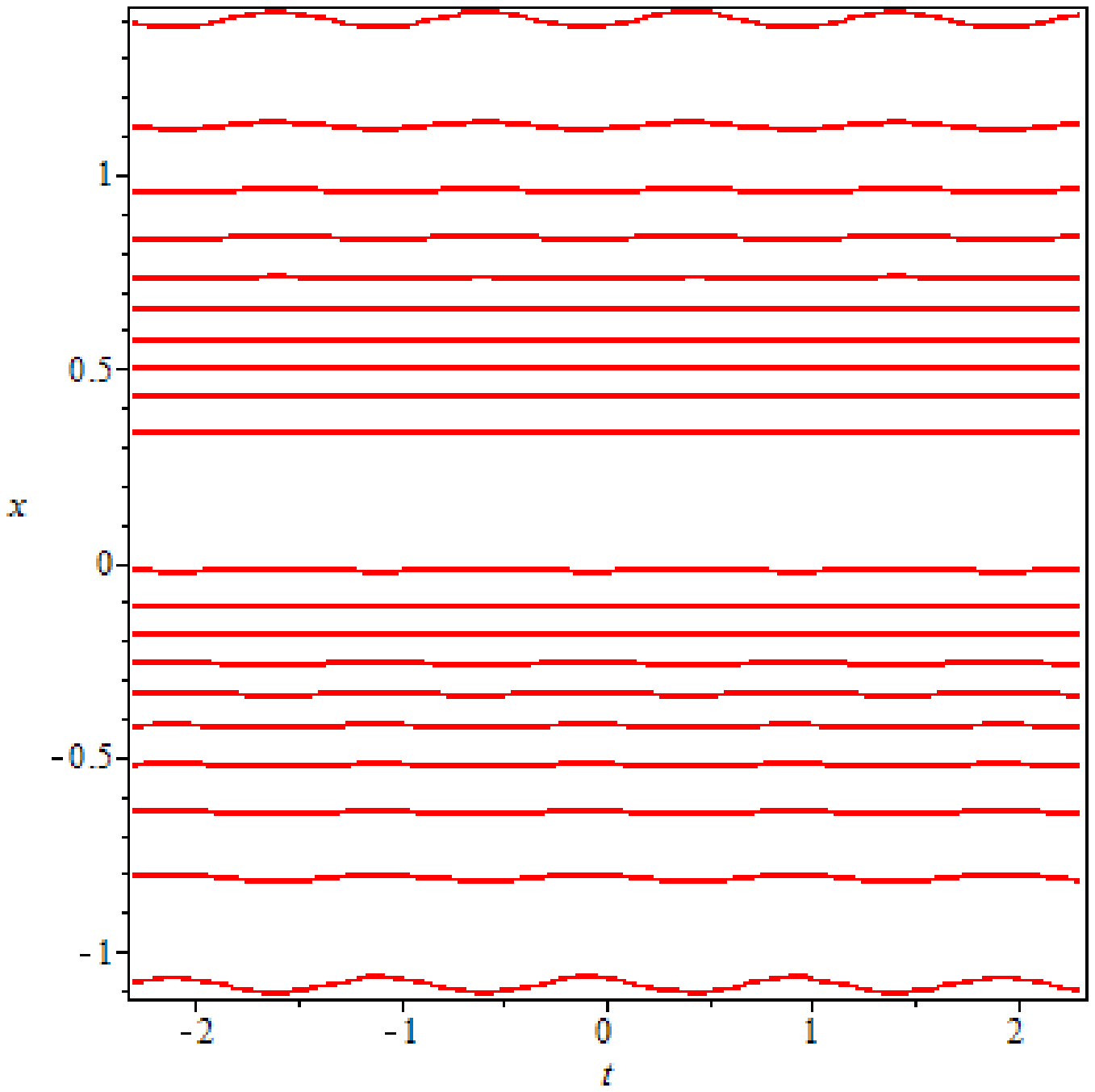}}}

$\qquad~~~~~~~~~(\textbf{d})\qquad \ \qquad\qquad\qquad\qquad~~~(\textbf{e})
\qquad\qquad\qquad\qquad\qquad~(\textbf{f})$\\
\noindent { \small \textbf{Figure 4.} (Color online) Plots of the soliton solution of the equation  with the parameters $\epsilon=0$, $\xi_{1}=-2.5i$ and $b_{1}=e^{2+i}$.
$\textbf{(a)}$: the soliton solution with $q_{-}=0.1$,
$\textbf{(b)}$: the density plot corresponding to $(a)$,
$\textbf{(c)}$: the contour line corresponding to $(a)$,
$\textbf{(d)}$: the soliton solution with $q_{-}=0.01$,
$\textbf{(e)}$: the density plot corresponding to $(d)$,
$\textbf{(f)}$: the contour line corresponding to $(d)$.}\\

Figure 4. reveals that when the boundary value $q_{-}$ becomes smaller and smaller, the stationary breather solution is getting closer a bell soliton solution. Of course, it just looks like a a bell soliton solution and some properties might be different, but it is really an interesting phenomenon.

Now, we consider the case that the dimensionless parameter $\epsilon$ is non-zero. Similarly, the following images can be obtained by selecting appropriate parameters and changing $\epsilon$.\\

{\rotatebox{0}{\includegraphics[width=3.6cm,height=3.0cm,angle=0]{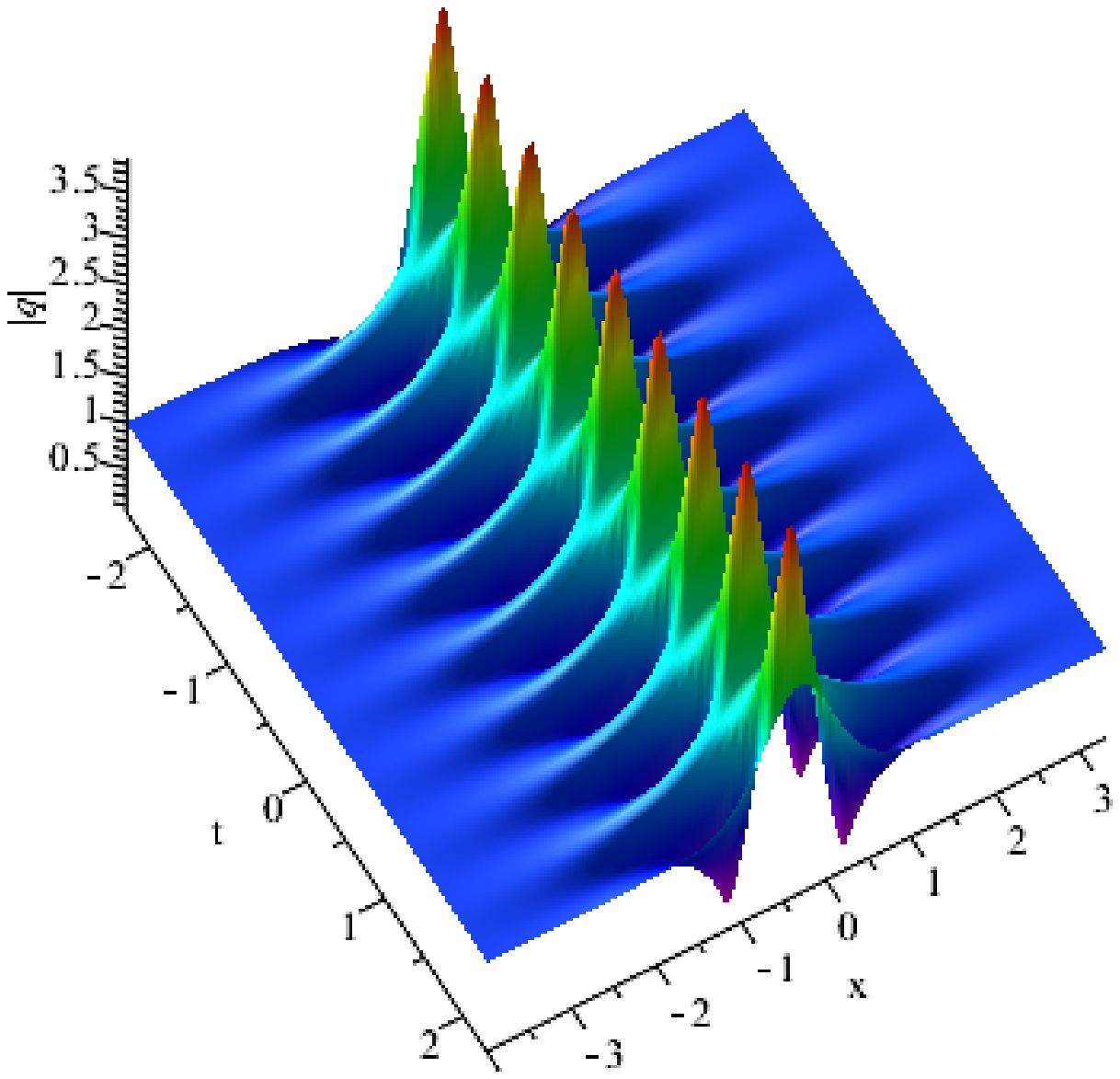}}}
~~~~
{\rotatebox{0}{\includegraphics[width=3.6cm,height=3.0cm,angle=0]{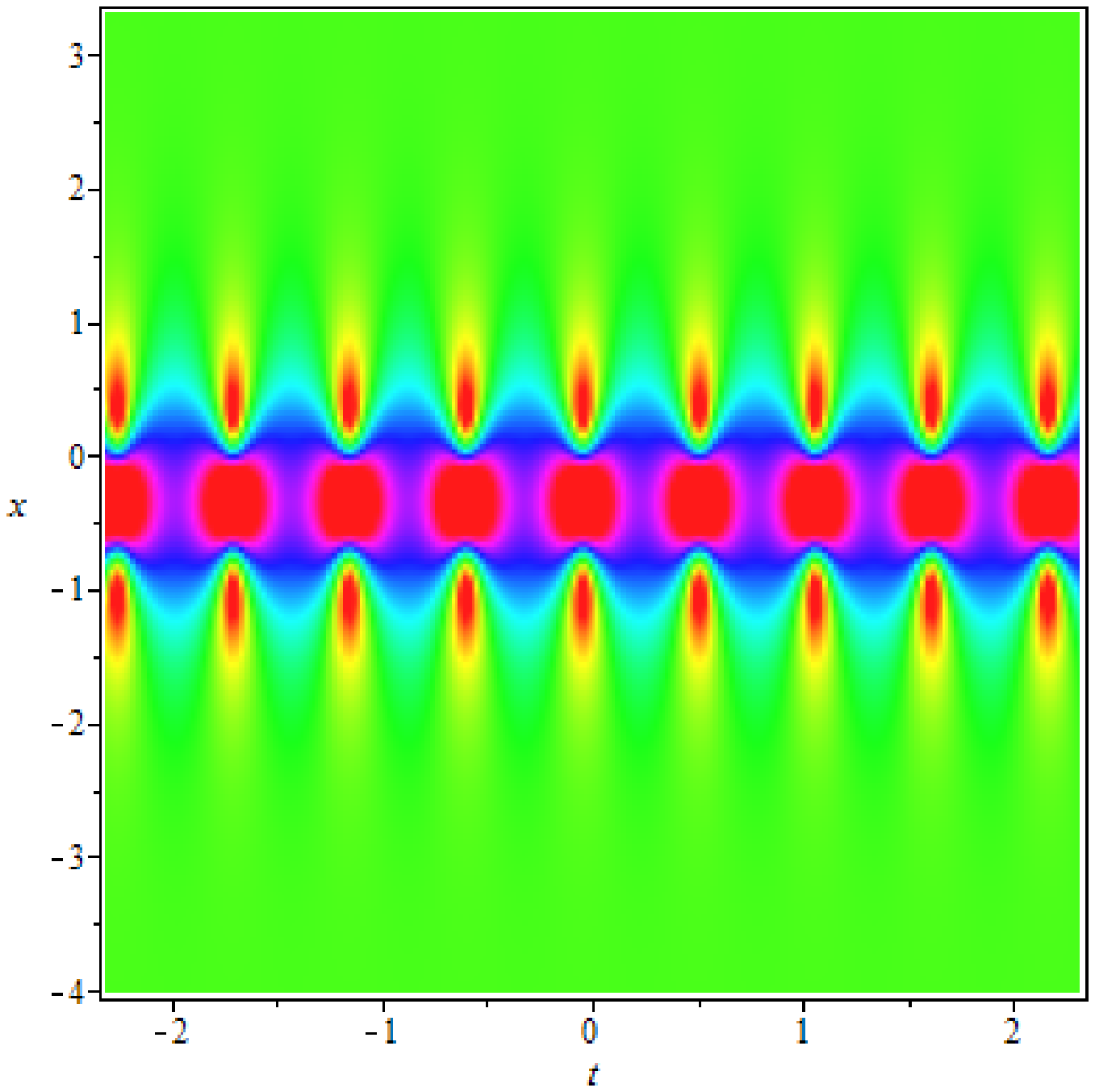}}}
~~~~
{\rotatebox{0}{\includegraphics[width=3.6cm,height=3.0cm,angle=0]{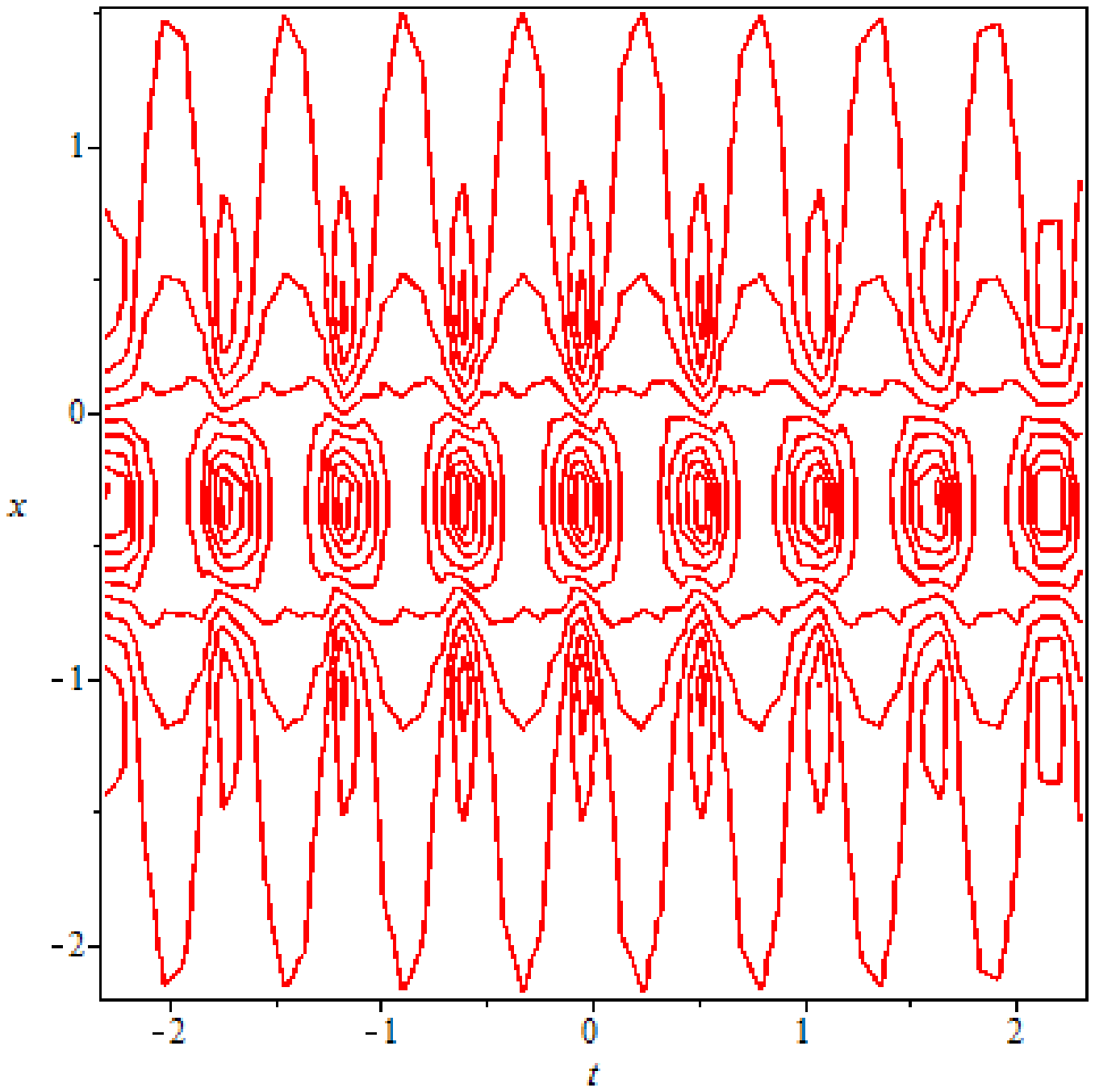}}}

$\qquad~~~~~~~~~(\textbf{a})\qquad \ \qquad\qquad\qquad\qquad~~~(\textbf{b})
\qquad\qquad\qquad\qquad\qquad~(\textbf{c})$\\

{\rotatebox{0}{\includegraphics[width=3.6cm,height=3.0cm,angle=0]{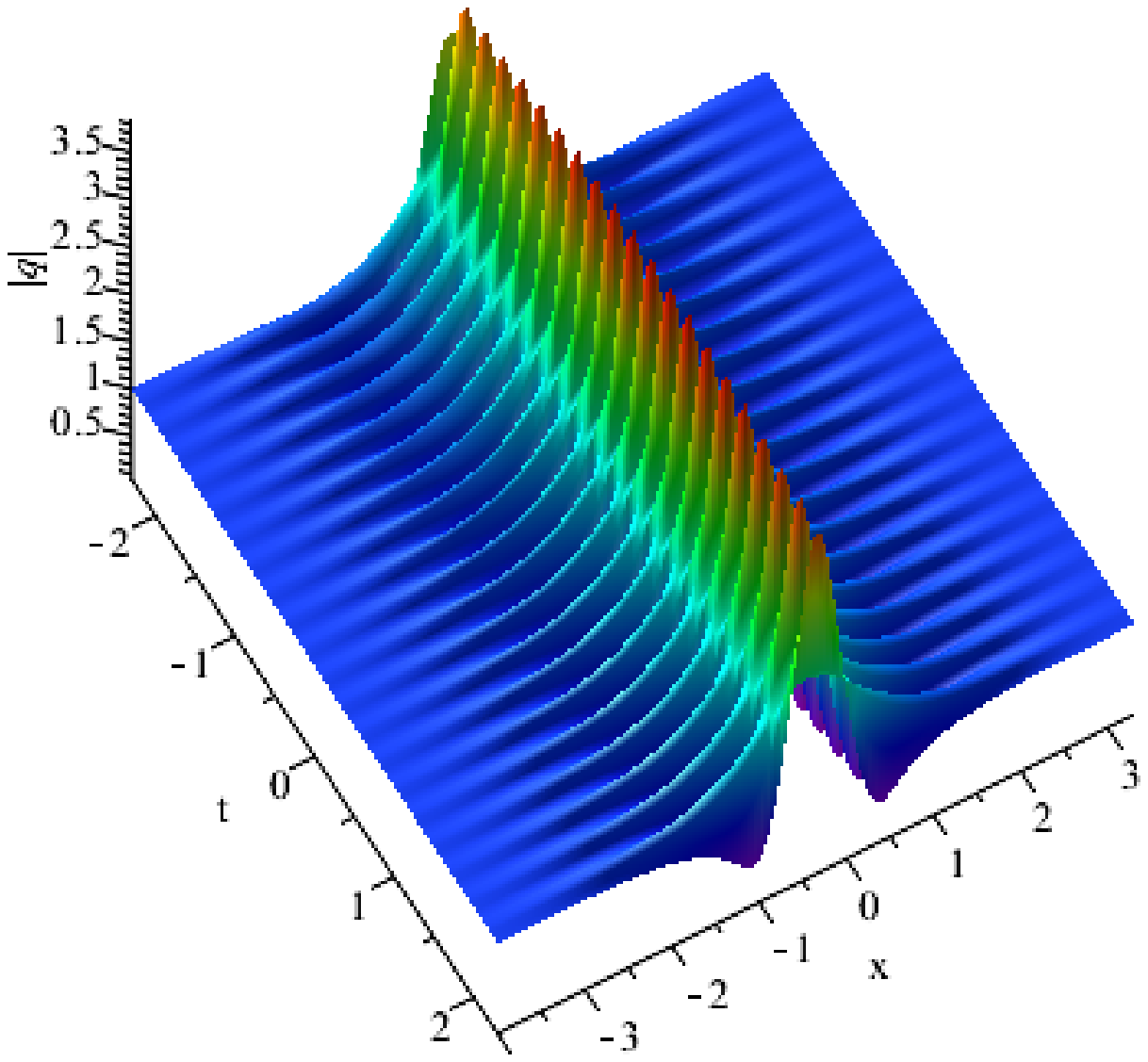}}}
~~~~
{\rotatebox{0}{\includegraphics[width=3.6cm,height=3.0cm,angle=0]{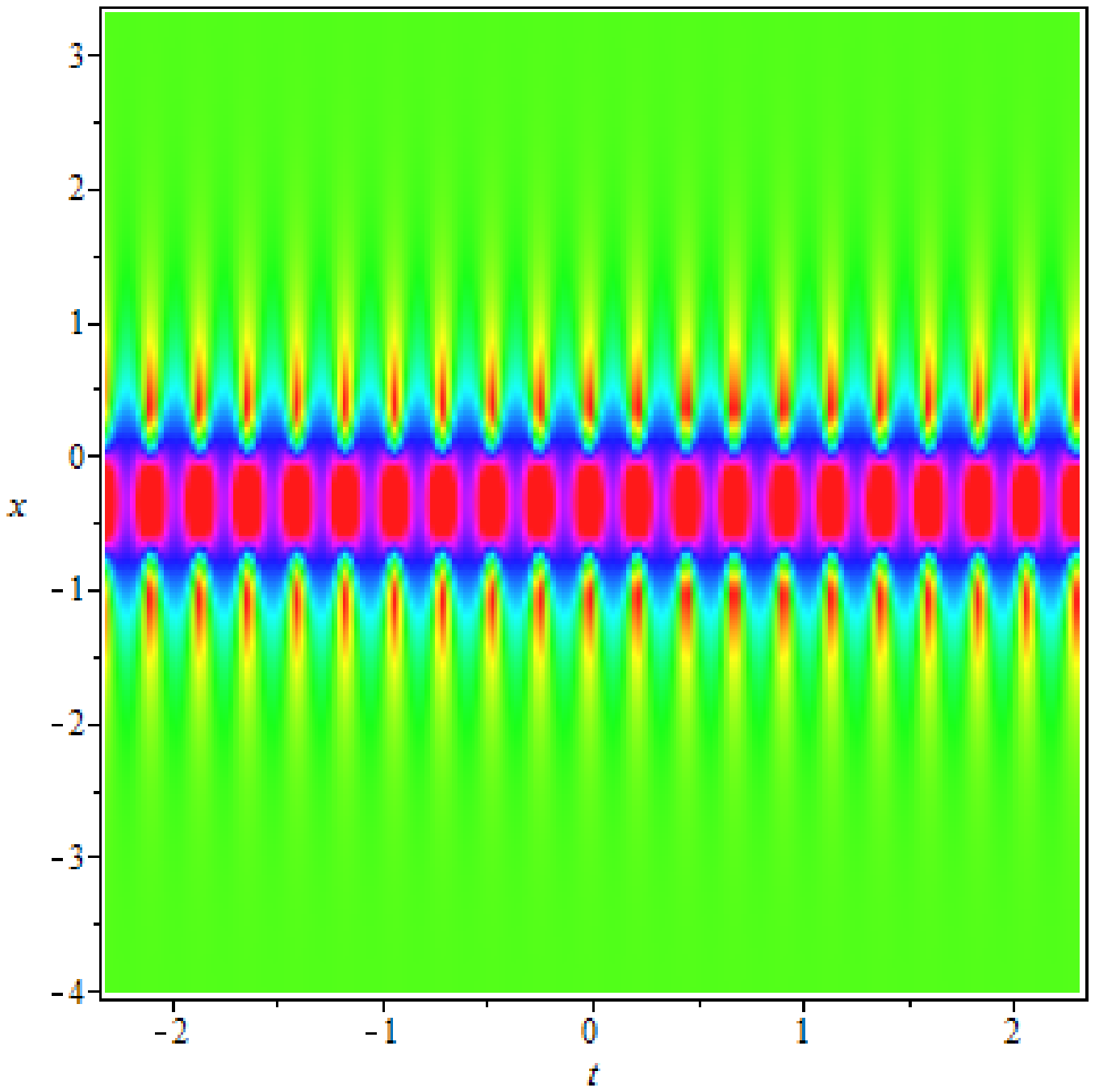}}}
~~~~
{\rotatebox{0}{\includegraphics[width=3.6cm,height=3.0cm,angle=0]{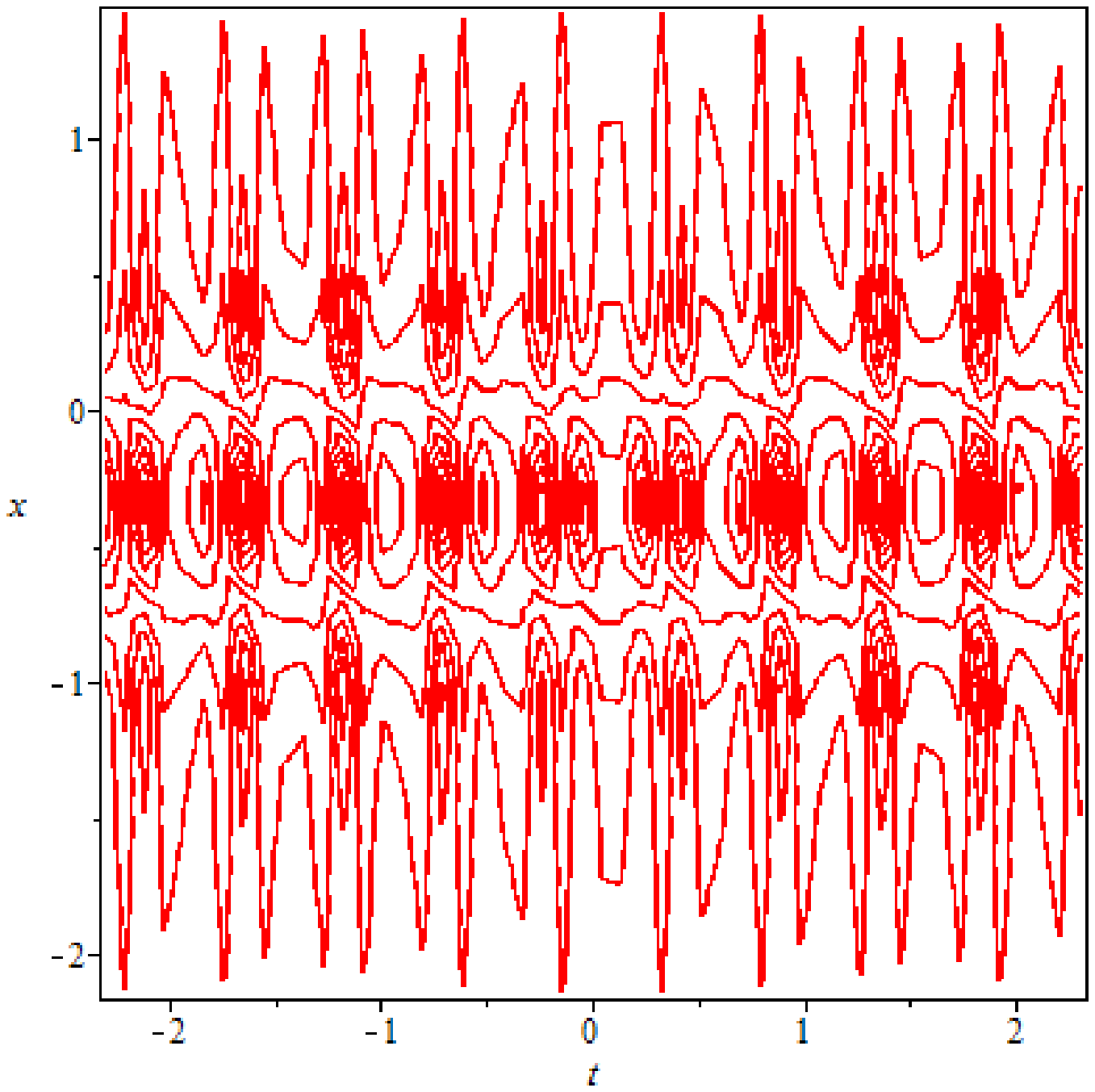}}}

$\qquad~~~~~~~~~(\textbf{d})\qquad \ \qquad\qquad\qquad\qquad~~~(\textbf{e})
\qquad\qquad\qquad\qquad\qquad~(\textbf{f})$\\

{\rotatebox{0}{\includegraphics[width=3.6cm,height=3.0cm,angle=0]{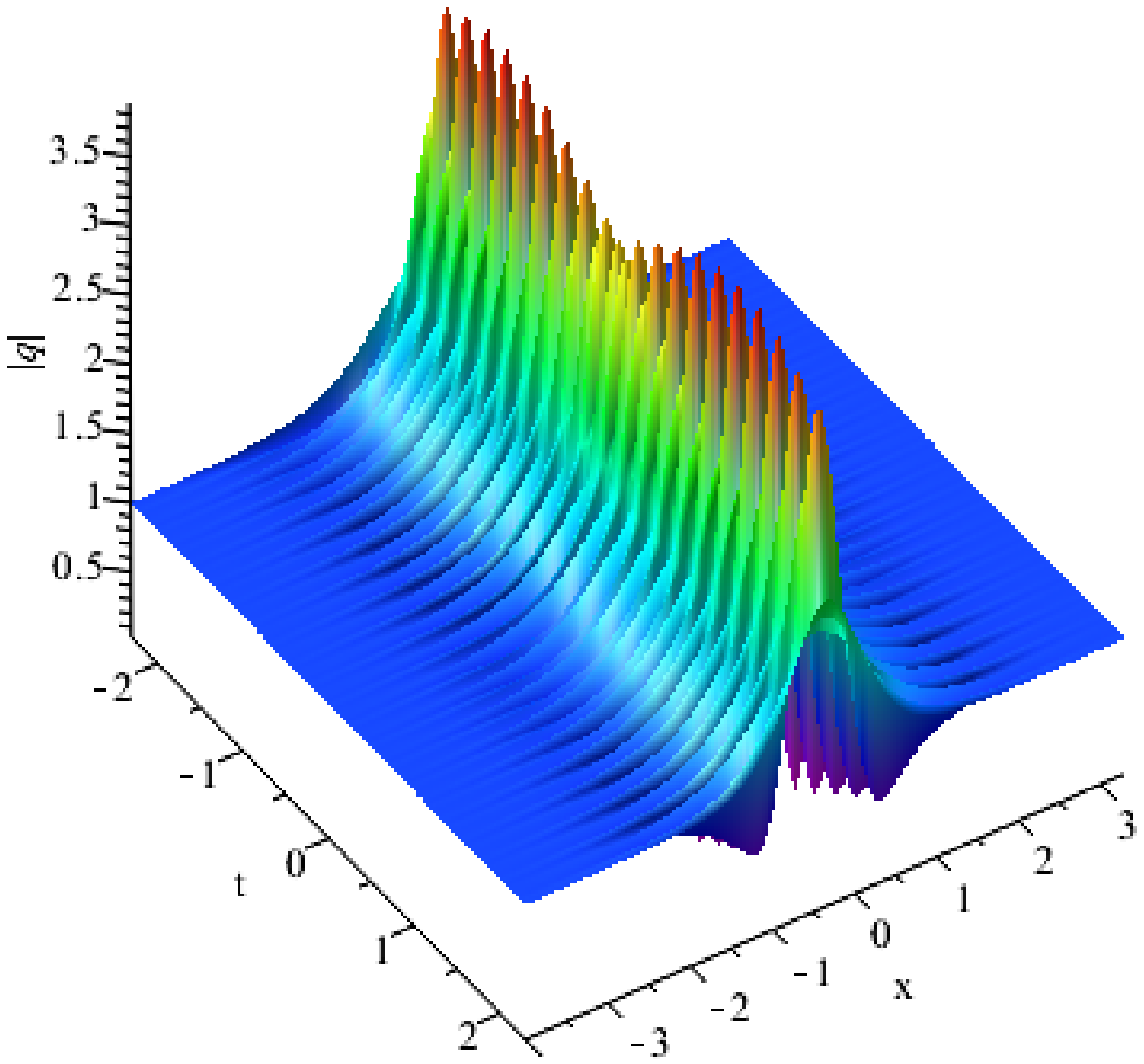}}}
~~~~
{\rotatebox{0}{\includegraphics[width=3.6cm,height=3.0cm,angle=0]{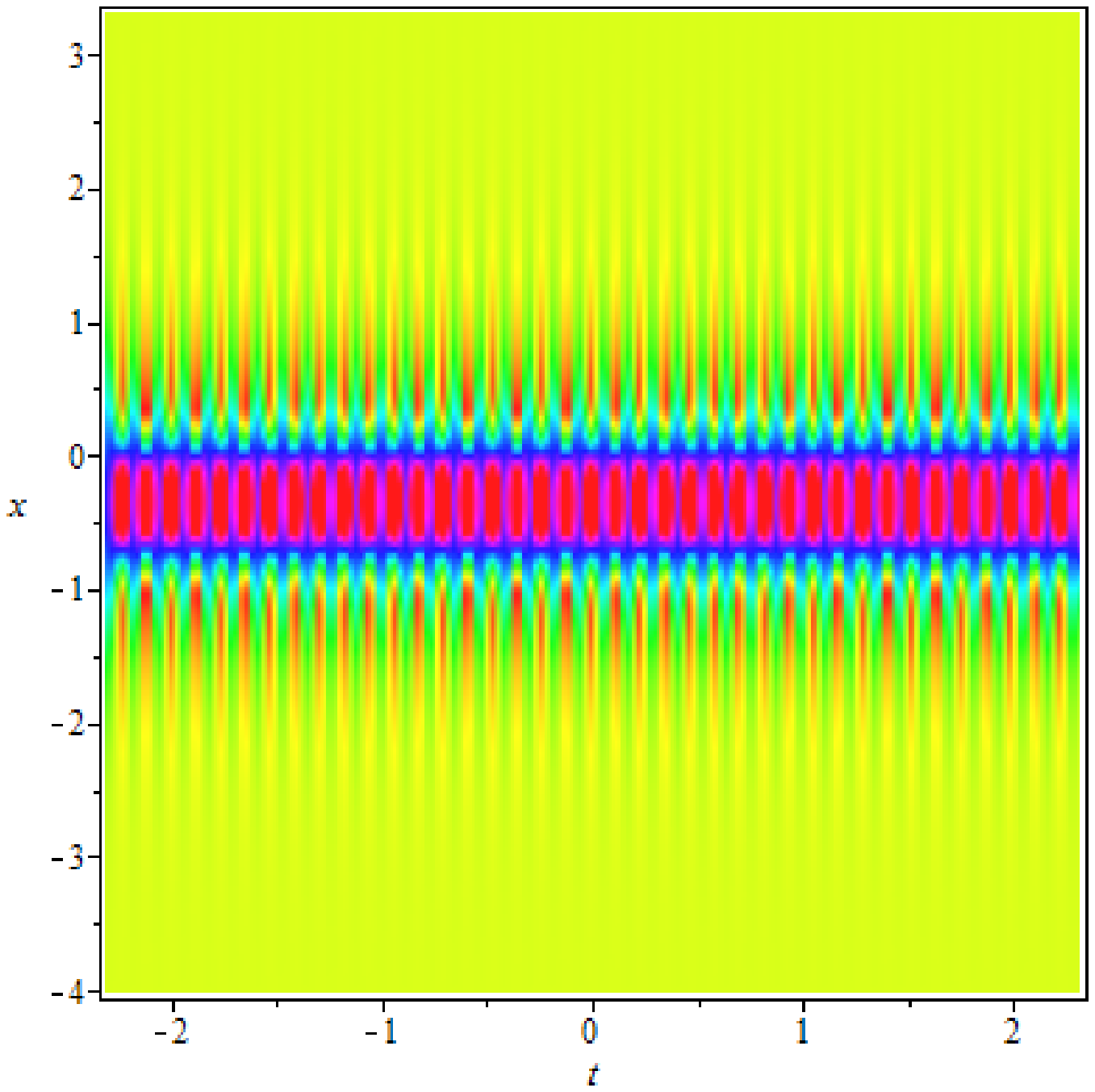}}}
~~~~
{\rotatebox{0}{\includegraphics[width=3.6cm,height=3.0cm,angle=0]{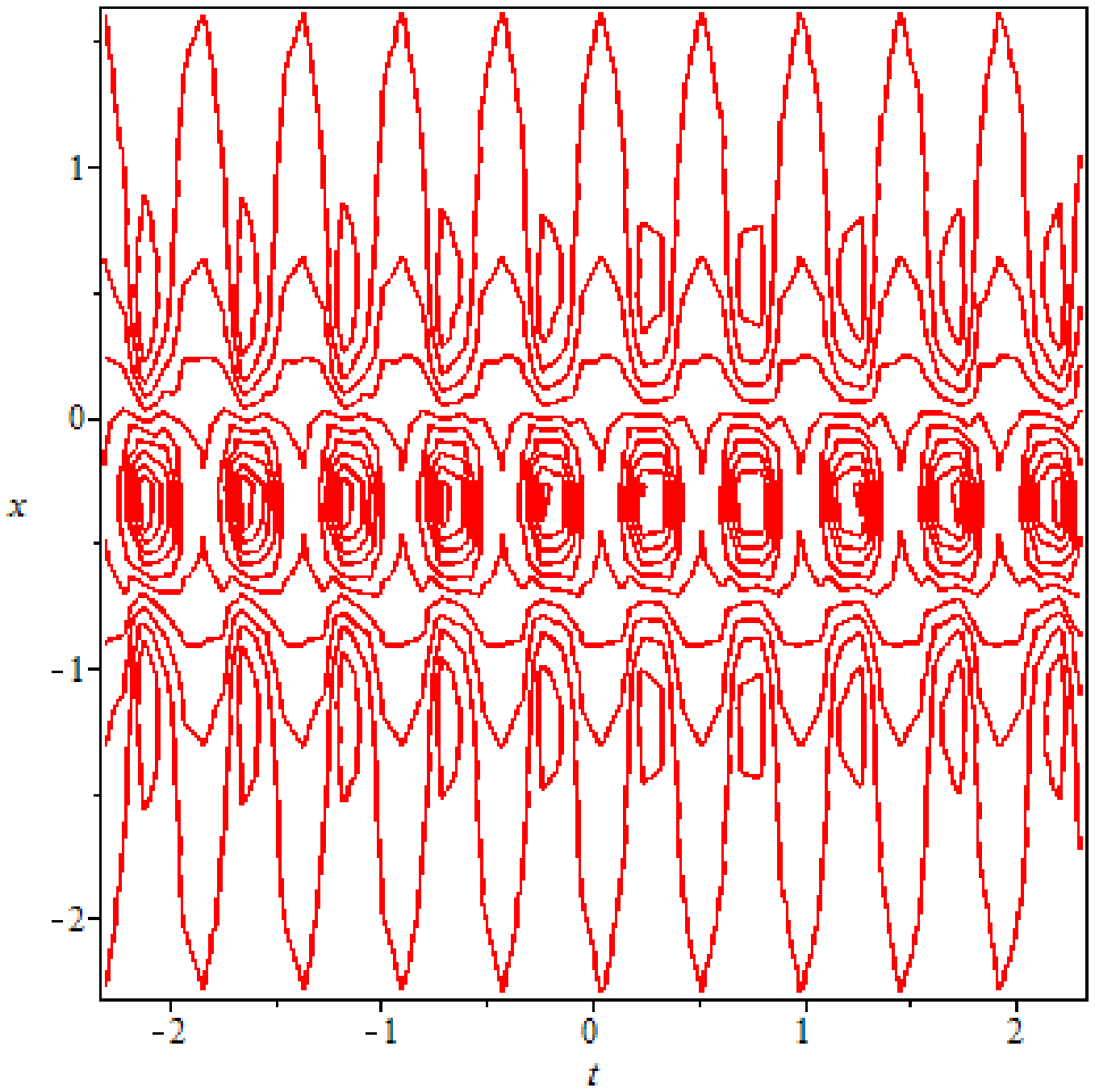}}}

$\qquad~~~~~~~~~(\textbf{g})\qquad \ \qquad\qquad\qquad\qquad~~~(\textbf{h})
\qquad\qquad\qquad\qquad\qquad~(\textbf{i})$\\
\noindent { \small \textbf{Figure 5.} (Color online) Plots of the soliton solution of the equation  with the parameters $q_{-}=1$, $\xi_{1}=-2.5i$ and $b_{1}=e^{2+i}$.
$\textbf{(a)}$: the soliton solution with $\epsilon=1$,
$\textbf{(b)}$: the density plot corresponding to $(a)$,
$\textbf{(c)}$: the contour line of the soliton solution corresponding to $(a)$,
$\textbf{(d)}$: the soliton solution with $\epsilon=2$,
$\textbf{(e)}$: the density plot corresponding to $(d)$,
$\textbf{(f)}$: the contour line of the soliton solution corresponding to $(d)$,
$\textbf{(g)}$: the soliton solution with $\epsilon=3$,
$\textbf{(h)}$: the density plot corresponding to $(g)$,
$\textbf{(i)}$: the contour line of the soliton solution corresponding to $(g)$.}\\

From the Figure 5. and comparing it with Figure 1, it is obvious that when the dimensionless parameter $\epsilon$ becomes larger, the soliton solutions are more closely arranged. We can learn it more clearly from the density maps. Thus, compared to the classical Schr\"{o}dinger equation, the addition of $\epsilon$ will change the shape of the solution but will not change the structure of the solution.

The above analysis reveals the dynamic behavior of the soliton solutions and shows various interesting phenomena by selecting appropriate parameters and change some parameters. It is hoped that these results can contribute to the physical field.

\section{The HDNLS equation with NZBCs: double poles}
The case of single zeros of the analytic scattering coefficients has been investigated. In this section, we, then, study the double zeros of the analytic scattering coefficients.
\subsection{Discrete spectrum with double poles and residue conditions}
From the above analysis, we know that the discrete spectrum of the scattering problem is set that contains all values $z\in\mathbb{C}\setminus\Sigma$ which make the eigenfunctions exist in $L^{2}(\mathbb{R})$. We assume that $z_{n}(\in D_{-}\cap\{z\in\mathbb{C}: Imz<0\}, n=1,2,...,N)$ are the double zeros of $s_{11}(z)$ i.e. $s_{11}(z_{n})=s'_{11}(z_{n})=0$ but $s''_{11}(z_{n})\neq0$, $n=1, 2,\cdots, N$. Then, based on the \emph{Theorem 2.5} , the discrete spectrum of the scattering problem can be acquired as
\begin{align}
\left\{z_{n}, -\frac{q_{0}^{2}}{z_{n}^{*}},
  z_{n}^{*}, -\frac{q_{0}^{2}}{z_{n}}\right\}, \quad n=1,2,...,N. \notag
\end{align}
For convenience, we use the transformation that $\hat{z}_{n}=-\frac{q^{2}_{0}}{z_{n}}$ and $\check{z}_{n}=-\frac{q^{2}_{0}}{z^{*}_{n}}$. Then, according to the \eqref{E-1} and \eqref{2.13} we have the following theorem.

\noindent \textbf {Theorem 3.1}
\emph{When $z_{n}(\in D_{-}\cap\{z\in\mathbb{C}: Imz<0\}, n=1,2,...,N)$ are the double zeros of $s_{11}(z)$ , the }
\begin{align} \label{E-34}
\mu_{+,1}(x,t;z_{n})=e_{n}e^{2i\theta(x,t;z_{n})}\mu_{-,2}(x,t;z_{n})
\end{align}
\emph{and}
\begin{align}\label{E-35}
\mu'_{+,1}(x,t;z_{n})=e^{2i\theta(x,t;z_{n})}[(h_{n}+2ie_{n}\theta'(x,t;z_{n}))\mu_{-,2}(x,t;z_{n})
+e_{n}\mu'_{-,2}(x,t;z_{n})],
\end{align}
\emph{where $e_{n}$ and $h_{n}$ are constants and independent of $x$ and $t$.}
\begin{proof}
Based on the \eqref{2.13} and $z_{n}$ are the double zeros of $s_{11}(z)$, we can derive that
\begin{align}\label{E-36}
\phi_{+,1}(z_{n})=e_{n}\phi_{-,2}(z_{n}).
\end{align}
Then, using the \eqref{E-1}, we have the \eqref{E-34}.
Considering that $s'_{11}(z_{n})=0$, we can derive that
\begin{align}
Wr\left(\phi'_{+,1}(z_{n})-e_{n}\phi'_{-,2}(z_{n}),\phi_{-,2}(z_{n})\right)=0.\notag
\end{align}
So, we calculate that
\begin{align}\label{E-37}
\phi'_{+,1}(z_{n})=e_{n}\phi'_{-,2}(z_{n})+h_{n}\phi_{-,2}(z_{n}).
\end{align}
The $e_{n}$ and $h_{n}$ in \eqref{E-36} and \eqref{E-37}, respectively, are constant and independent of $x$ and $t$.
Therefore, based on the \eqref{E-1}, we can derive the \eqref{E-35}.
\end{proof}
Similar to the \emph{Theorem 3.1}, we have the following relationship.
\begin{subequations}
\begin{gather}
\mu_{+,1}(x,t;\check{z}_{n})=\check{e}_{n}e^{2i\theta(x,t;\check{z}_{n})}\mu_{-,2}(x,t;\check{z}_{n}),\label{E-38a}\\
\mu'_{+,1}(x,t;\check{z}_{n})=e^{2i\theta(x,t;\check{z}_{n})}
[(\check{h}_{n}+2i\check{e}_{n}\theta'(x,t;\check{z}_{n}))\mu_{-,2}(x,t;\check{z}_{n})
+\check{e}_{n}\mu'_{-,2}(x,t;\check{z}_{n})];\label{E-38b}\\
\mu_{+,2}(x,t;z^{*}_{n})=\tilde{e}_{n}e^{-2i\theta(x,t;z^{*}_{n})}\mu_{-,1}(x,t;z^{*}_{n}),\label{E-39a}\\
\mu'_{+,2}(x,t;z^{*}_{n})=e^{-2i\theta(x,t;z^{*}_{n})}
[(\tilde{h}_{n}-2i\tilde{e}_{n}\theta'(x,t;z^{*}_{n}))\mu_{-,1}(x,t;z^{*}_{n})
+\tilde{e}_{n}\mu'_{-,1}(x,t;z^{*}_{n})];\label{E-39b}\\
\mu_{+,2}(x,t;\hat{z}_{n})=\hat{e}_{n}e^{-2i\theta(x,t;\hat{z}_{n})}\mu_{-,1}(x,t;\hat{z}_{n}),\label{E-40a}\\
\mu'_{+,2}(x,t;\hat{z}_{n})=e^{-2i\theta(x,t;\hat{z}_{n})}
[(\hat{h}_{n}-2i\hat{e}_{n}\theta'(x,t;\hat{z}_{n}))\mu_{-,1}(x,t;\hat{z}_{n})
+\hat{e}_{n}\mu'_{-,1}(x,t;\hat{z}_{n})];\label{E-40b}
\end{gather}
\end{subequations}
where the $\check{e}_{n}$, $\check{h}_{n}$, $\tilde{e}_{n}$, $\tilde{h}_{n}$, $\hat{e}_{n}$, and $\hat{h}_{n}$ are constants and independent of $x$ and $t$. Furthermore, the constants have the fixed relationship.

\noindent \textbf {Theorem 3.2}
\emph{The relationship among the constants are that }
\begin{gather}
\tilde{e}_{n}=-e^{*}_{n}, \quad \hat{e}_{n}=e_{n}\frac{q_{-}}{q^{*}_{+}}, \quad \check{e}_{n}=-\hat{e}^{*}_{n}, \notag\\
\tilde{h}_{n}=-h^{*}_{n},\quad \hat{h}_{n}=h_{n}\frac{q_{-}z^{2}_{n}}{q^{*}_{+}q^{2}_{0}},\quad
\check{h}_{n}=-\hat{h}^{*}_{n}.\notag
\end{gather}
\begin{proof}
According to the \eqref{E-4}, \eqref{E-39a} and \eqref{E-39b}, we have that
\begin{align} \label{E-41}
-\sigma\mu^{*}_{+,1}(x,t;z_{n})=\tilde{e}_{n}e^{-2i\theta(x,t;z^{*}_{n})}\sigma\mu^{*}_{-,2}(x,t;z_{n}).
\end{align}
Applying $\theta(x,t;z^{*}_{n})=\theta^{*}(x,t;z_{n})$, multiplying $\sigma$ to both ends of equation \eqref{E-41} and taking the complex conjugate, we obtain that
\begin{align}
\mu_{+,1}(x,t;z_{n})=-\tilde{e}^{*}_{n}e^{2i\theta(x,t;z_{n})}\mu_{-,2}(x,t;z_{n}).
\end{align}
Therefore, the $\tilde{e}_{n}=-e^{*}_{n}$ has been proved.\\
The other relationships among the constants can be proved similarly.
\end{proof}
Then, we pay attention to the residue condition that will be useful in the inverse problem.

\noindent \textbf {Proposition 1.}
\emph{If the functions $f$ and $g$ are analytic in a complex region $\Omega\in\mathbb{C}$ , $g$ has a double poles at $z_{0}\in\Omega$ i.e. $g(z_{0})=g'(z_{0})=0$, $g''(z_{0})\neq0$, and $f(z_{0})\neq0$. Thus the residue of $f/g$ can be calculated by the Laurent expansion at $z=z_{0}$, namely}
\begin{align}\label{E-42}
\mathop{Res}_{z=z_{0}}\left[\frac{f}{g}\right]=\frac{2f'(z_{0})}{g''(z_{0})}-
\frac{2f(z_{0})g'''(z_{0})}{3(g''(z_{0}))^{2}},\quad
\mathop{P_{-2}}_{z=z_{0}}\left[\frac{f}{g}\right]=\frac{2f(z_{0})}{g''(z_{0})}.
\end{align}

According to the \emph{Theorem 3.1} and the formulae \eqref{E-42}, we can derive the residue of the $u_{+,1}(x,t;z)/s_{11}(z)$ as
\begin{subequations}
\begin{align}
&\mathop{P_{-2}}_{z=z_{n}}\left[\frac{\mu_{+,1}(x,t;z)}{s_{11}(z)}\right]=
\frac{2\mu_{+,1}(x,t;z_{n})}{s''_{11}(z_{n})}=
\frac{2e_{n}}{s''_{11}(z_{n})}e^{2i\theta(x,t;z_{n})}\mu_{-,2}(x,t;z_{n}),\label{E-43a}\\
&\mathop{Res}_{z=z_{n}}\left[\frac{\mu_{+,1}(x,t;z)}{s_{11}(z)}\right]=
\frac{2\mu'_{+,1}(x,t;z_{n})}{s''_{11}(z_{n})}
-\frac{2\mu_{+,1}(x,t;z_{n})a'''_{11}(z_{n})}{3(s''_{11}(z_{n}))^{2}} \notag\\
&=\frac{2e_{n}}{s''_{11}(z_{n})}e^{2i\theta(x,t;z_{n})}\left[\mu'_{-,2}(x,t;z_{n})+
\left(\frac{h_{n}}{e_{n}}+2i\theta'(x,t;z_{n})-\frac{s'''_{11}(z_{n})}{3s''_{11}(z_{n})}
\right)\mu_{-,2}(x,t;z_{n})\right].\label{E-43b}
\end{align}
\end{subequations}
In order to make the analysis more convenient, we introduce the following representation
\begin{align}
E_{n}=\frac{2e_{n}}{s''_{11}(z_{n})}, \quad H_{n}=\frac{h_{n}}{e_{n}}-\frac{s'''_{11}(z_{n})}{3s''_{11}(z_{n})}.\notag
\end{align}
Then, the \eqref{E-43a} and \eqref{E-43b} can be transformed into
\begin{subequations}
\begin{align}
&\mathop{P_{-2}}_{z=z_{n}}\left[\frac{\mu_{+,1}(x,t;z)}{s_{11}(z)}\right]=
E_{n}e^{2i\theta(x,t;z_{n})}\mu_{-,2}(x,t;z_{n}),\label{E-44a}\\
&\mathop{Res}_{z=z_{n}}\left[\frac{\mu_{+,1}(x,t;z)}{s_{11}(z)}\right]=
E_{n}e^{2i\theta(x,t;z_{n})}\left[\mu'_{-,2}(x,t;z_{n})+
\left(H_{n}+2i\theta'(x,t;z_{n})\right)\mu_{-,2}(x,t;z_{n})\right].\label{E-44b}
\end{align}
\end{subequations}
From the above analysis, we know that $z=\check{z}_{n}$ are also the double zeros of $s_{11}$ and $z=z_{n}^{*}, z=\hat{z}_{n}$ are the double zeros of $s_{22}$. Similar to the \eqref{E-43a}, \eqref{E-43b}, \eqref{E-44a} and \eqref{E-44b}, we can derive that
\begin{subequations}
\begin{align}
&\mathop{P_{-2}}_{z=\check{z}_{n}}\left[\frac{\mu_{+,1}(z)}{s_{11}(z)}\right]=
\check{E}_{n}e^{2i\theta(\check{z}_{n})}\mu_{-,2}(\check{z}_{n}),\\
&\mathop{Res}_{z=\check{z}_{n}}\left[\frac{\mu_{+,1}(z)}{s_{11}(z)}\right]=
\check{E}_{n}e^{2i\theta(\check{z}_{n})}\left[\mu'_{-,2}(\check{z}_{n})+
\left(\check{H}_{n}+2i\theta'(\check{z}_{n})\right)\mu_{-,2}(\check{z}_{n})\right];\\
&\mathop{P_{-2}}_{z=z^{*}_{n}}\left[\frac{\mu_{+,2}(z)}{s_{22}(z)}\right]=
\tilde{E}_{n}e^{-2i\theta(z^{*}_{n})}\mu_{-,1}(z^{*}_{n}),\\
&\mathop{Res}_{z=z^{*}_{n}}\left[\frac{\mu_{+,2}(z)}{s_{22}(z)}\right]=
\tilde{E}_{n}e^{-2i\theta(z^{*}_{n})}\left[\mu'_{-,1}(z^{*}_{n})+
\left(H_{n}-2i\theta'(z^{*}_{n})\right)\mu_{-,1}(z^{*}_{n})\right];\\
&\mathop{P_{-2}}_{z=\hat{z}_{n}}\left[\frac{\mu_{+,2}(z)}{s_{22}(z)}\right]=
\hat{E}_{n}e^{-2i\theta(\hat{z}_{n})}\mu_{-,1}(\hat{z}_{n}),\\
&\mathop{Res}_{z=\hat{z}_{n}}\left[\frac{\mu_{+,2}(z)}{s_{22}(z)}\right]=
\hat{E}_{n}e^{-2i\theta(\hat{z}_{n})}\left[\mu'_{-,1}(\hat{z}_{n})+
\left(\hat{H}_{n}-2i\theta'(\hat{z}_{n})\right)\mu_{-,1}(\hat{z}_{n})\right];
\end{align}
\end{subequations}
where $\check{E}_{n}$, $\check{H}_{n}$, $\tilde{E}_{n}$, $\tilde{H}_{n}$, $\hat{E}_{n}$ and $\hat{H}_{n}$ are the transformations corresponding to the constants $\check{e}_{n}$, $\check{h}_{n}$, $\tilde{e}_{n}$, $\tilde{h}_{n}$, $\hat{e}_{n}$, and $\hat{h}_{n}$, respectively.

\subsection{RH problem and reconstruction formula for the potential}
In the analysis of simple poles, we have constructed a generalized RHP and obtained the asymptotic behavior of the $M^{\pm}$ and the jump matrix $G$. Then, for convenience, we introduce a substitution
\begin{align}
\xi_{n}:=z_{n},\quad \xi_{N+n}:=\check{z}_{n},\quad \hat{\xi}_{n}:=\hat{z}_{n}, \quad \hat{\xi}_{N+n}:=z_{n}^{*}.
\end{align}
So, the $\xi_{n}, n=1,2,\cdots, 2N$ are the double poles in $D_{-}$ and $\xi^{*}_{n}, n=1,2,\cdots, 2N$ are the double poles in $D_{+}$. Via subtracting out the asymptotic behavior and the double poles contributions, we obtain a regular RHP. Then, we have
\begin{align}\label{E-45}
\begin{split}
 M^{+}-&\mathbb{I}+\frac{i}{z}\sigma_{3}Q_{-}-\sum_{n=1}^{2N}\left\{
\frac{\mathop{Res}\limits_{z=\hat{\xi}_{n}}M^{+}}{z-\hat{\xi}_{n}}
+\frac{\mathop{P_{-2}}\limits_{z=\hat{\xi}_{n}}M^{+}}{(z-\hat{\xi}_{n})^{2}}
+\frac{\mathop{Res}\limits_{z=\xi_{n}}M^{-}}{z-\xi_{n}}
+\frac{\mathop{P_{-2}}\limits_{z=\xi_{n}}M^{-}}{(z-\xi_{n})^{2}}\right\}\\&=
M^{-}-\mathbb{I}+\frac{i}{z}\sigma_{3}Q_{-}-\sum_{n=1}^{2N}\left\{
\frac{\mathop{Res}\limits_{z=\hat{\xi}_{n}}M^{+}}{z-\hat{\xi}_{n}}
+\frac{\mathop{P_{-2}}\limits_{z=\hat{\xi}_{n}}M^{+}}{(z-\hat{\xi}_{n})^{2}}
+\frac{\mathop{Res}\limits_{z=\xi_{n}}M^{-}}{z-\xi_{n}}
+\frac{\mathop{P_{-2}}\limits_{z=\xi_{n}}M^{-}}{(z-\xi_{n})^{2}}\right\}
-M^{-}G.
\end{split}
\end{align}
According to the \eqref{Matrix}, we acquire that
\begin{align}
\mathop{Res}_{z=\xi_{n}}[M^{-}]=(\mathop{Res}_{z=\xi_{n}}\left[\frac{\mu_{+,1}(x,t;z)}{s_{11}(z)}\right],0), \quad \mathop{P_{-2}}_{z=\xi_{n}}M^{-}=(\mathop{P_{-2}}_{z=\xi_{n}}\left[\frac{\mu_{+,1}(x,t;z)}{s_{11}(z)}\right],0), \notag \\
\mathop{Res}_{z=\hat{\xi}_{n}}[M^{+}]=(0,\mathop{Res}_{z=\hat{\xi}_{n}}\left[\frac{\mu_{+,2}(x,t;z)}
{s_{22}(z)}\right]), \quad \mathop{P_{-2}}_{z=\hat{\xi}_{n}}M^{+}=(0,\mathop{P_{-2}}_{z=\hat{\xi}_{n}}\left[\frac{\mu_{+,2}(x,t;z)}
{s_{22}(z)}\right]). \notag
\end{align}
It is easy to verify that the left side of \eqref{E-45} is analytic in $D_{+}$ and the right side of \eqref{E-45}, apart from the item $M^{-}(z)G(z)$, is analytic in $D_{-}$. Meanwhile, both sides of the equation \eqref{E-45} have the asymptotic behavior that are $O(1/z)(z\rightarrow\infty)$ and $O(1)(z\rightarrow0)$. We also have the asymptotic behavior of $G(x,t;s)$, i.e., $O(1/z)(z\rightarrow\infty)$ and $O(1)(z\rightarrow0)$. Therefore, applying the cauchy projectors \eqref{Cauchy}, the solution of the RHP can be obtained as
\begin{align}\label{E-46}
M(x,t;z)=&\mathbb{I}-\frac{i}{z}\sigma_{3}Q_{-}+\sum_{n=1}^{2N}\left\{
\frac{\mathop{Res}\limits_{z=\hat{\xi}_{n}}M^{+}}{z-\hat{\xi}_{n}}
+\frac{\mathop{P_{-2}}\limits_{z=\hat{\xi}_{n}}M^{+}}{(z-\hat{\xi}_{n})^{2}}
+\frac{\mathop{Res}\limits_{z=\xi_{n}}M^{-}}{z-\xi_{n}}
+\frac{\mathop{P_{-2}}\limits_{z=\xi_{n}}M^{-}}{(z-\xi_{n})^{2}}\right\} \notag\\
&+\frac{1}{2i\pi}\int_{\Sigma}\frac{M^{-}(x,t;s)G(x,t;s)}{s-z}\,ds,\quad z\in\mathbb{C}\setminus\Sigma.
\end{align}
To obtain a closed algebraic integral system, the expression of the residue which emerge in \eqref{E-46} is necessary. We have shown this in the above analysis. We, therefore, evaluate the second column of the \eqref{E-46} at $z=\xi_{k}$ in $D_{-}$. Before this, we introduce the notation $\hat{E}_{N+n}:=\tilde{E}_{n}$ and $\hat{H}_{N+n}:=\tilde{H}_{n}$, and define that
\begin{align}
C_{n}(x,t;z)=\frac{\hat{E}_{n}}{z-\hat{\xi}_{n}}e^{-2i\theta(x,t;\hat{\xi}_{n})},\quad
D_{n}(x,t)=\hat{H}_{n}-2i\theta'(x,t;\hat{\xi}_{n}).\notag
\end{align}
It can be derived that $C'_{n}(x,t;z)=-\frac{C_{n}(x,t;z)}{z-\hat{\xi}_{n}}$.
Then, we have
\begin{align}\label{E-47}
u_{-,2}(x,t;\xi_{k})=&\left(\begin{array}{cc}
                       -\frac{iq_{-}}{\xi_{k}} \\
                        1
                     \end{array}\right)
+\frac{1}{2i\pi}\int_{\Sigma}\frac{(M^{-}G)_{2}(s)}{s-\xi_{k}}\,ds \notag \\
&+\sum_{n=1}^{2N}C_{n}(\xi_{k})
\left[\mu'_{-,1}(x,t;\hat{\xi}_{n})+\left(D_{n}(x,t)+\frac{1}{\xi_{k}-\hat{\xi}_{n}}\right)
\mu_{-,1}(x,t;\hat{\xi}_{n})\right].
\end{align}
Based on the \eqref{E-5}, we have $\mu_{-,2}(x,t;\xi_{k})=-\frac{iq_{-}}{\xi_{k}}\mu_{-,1}(x,t;\hat{\xi}_{k})$ and substitute it into the \eqref{E-47}. The following formulae,
\begin{align}\label{E-48}
&\sum_{n=1}^{2N}\left(C_{n}(\xi_{k})\mu'_{-,1}(x,t;\hat{\xi}_{n})+
\left[C_{n}(\xi_{k})\left(D_{n}(x,t)+\frac{1}{\xi_{k}-\hat{\xi}_{n}}\right)+
\frac{iq_{-}}{\xi_{k}}\delta_{kn}\right]\mu_{-,1}(x,t;\hat{\xi}_{n})\right) \notag\\
&+\left(\begin{array}{cc}
        -\frac{iq_{-}}{\xi_{k}} \\
         1
       \end{array}\right)
+\frac{1}{2i\pi}\int_{\Sigma}\frac{(M^{-}G)_{2}(s)}{s-\xi_{k}}\,ds=0,
\end{align}
can be obtained. Here, the $\delta_{ij}$ is Kronecker delta. Then, by taking the first-order derivative of $\mu_{-,2}(x,t;z)$ and \eqref{E-5} with respect to $z$, and evaluating at $z=\xi_{k}$, we can obtain that
\begin{align}\label{E-49}
&\sum_{n=1}^{2N}\left(\left[\frac{C_{n}(\xi_{k})}{\xi_{k}-\hat{\xi}_{n}}-
\frac{iq_{-}q_{0}^{2}}{\xi_{k}^{3}}\delta_{kn}\right]\mu'_{-,1}(\hat{\xi}_{n})+
\left[\frac{C_{n}(\xi_{k})}{\xi_{k}-\hat{\xi}_{n}}\left(D_{n}(x,t)+\frac{2}{\xi_{k}-\hat{\xi}_{n}}\right)
+\frac{iq_{-}}{\xi_{k}^{2}}\delta_{kn}\right]\mu_{-,1}(\hat{\xi}_{n})\right) \notag\\
&+\left(\begin{array}{cc}
        -\frac{iq_{-}}{\xi_{k}^{2}} \\
         0
       \end{array}\right)
-\frac{1}{2i\pi}\int_{\Sigma}\frac{(M^{-}G)_{2}(s)}{(s-\xi_{k})^{2}}\,ds=0.
\end{align}
Finally, considering the asymptotic behavior of the \eqref{E-46}, we obtain that
\begin{align}\label{E-50}
M(x,t;z)=\mathbb{I}&+\frac{1}{z}\{-i\sigma_{3}Q_{-}+\sum_{n=1}^{2N}\left[
\mathop{Res}_{z=\hat{\xi}_{n}}M^{+}
+\mathop{Res}_{z=\xi_{n}}M^{-}\right] \notag\\
&-\frac{1}{2i\pi}\int_{\Sigma}M^{-}(x,t;s)G(x,t;s)\,ds\}+O(z^{-2}),\quad z\rightarrow\infty.
\end{align}
Through taking $M=M^{-}$ and combining the $1,2$ element of \eqref{E-50} and the \emph{Theorem 2.6}, the reconstruction formula for the potential can be acquired as
\begin{align}\label{E-51}
q(x,t)=&-q_{-}+\frac{1}{2\pi}\int_{\Sigma}(M^{-}(x,t;s)G(x,t;s))_{12}\,ds \notag \\
&-i\sum_{n=1}^{2N}\hat{E}_{n}e^{-2i\theta(x,t;\hat{\xi}_{n})}
[\mu'_{-,11}(x,t;\hat{\xi}_{n})+\mu_{-,11}(x,t;\hat{\xi}_{n})D_{n}(x,t)].
\end{align}
\subsection{Trace formulate and theta condition}
Here, we will deduce the trace formulate and theta condition for the double poles case. According to the above analysis, we know that the $z_{n}, -\frac{q_{0}^{2}}{z_{n}^{*}} (n=1,2,\cdots,N)$ are the double zeros of $s_{11}$ and the $z_{n}^{*}, -\frac{q_{0}^{2}}{z_{n}} (n=1,2,\cdots,N)$ are the double zeros of $s_{22}$. Then, we construct the following function
\begin{align}\label{dTT-1}
\zeta^{-}_{2}(z)=s_{11}(z)\prod_{n=1}^{N}\frac{(z-z_{n}^{*})^{2}(z+q_{0}^{2}/z_{n})^{2}}
{(z-z_{n})^{2}(z+q_{0}^{2}/z_{n}^{*})^{2}},\notag\\
\zeta^{+}_{2}(z)=s_{22}(z)\prod_{n=1}^{N}\frac{(z-z_{n})^{2}(z+q_{0}^{2}/z_{n}^{*})^{2}}
{(z-z_{n}^{*})^{2}(z+q_{0}^{2}/z_{n})^{2}}.
\end{align}
The analytic properties of the $\zeta^{-}_{2}$ and $\zeta^{+}_{2}$ correspond to the $s_{11}$ and $s_{22}$. However, they have no zeros. Furthermore, considering the $\det S(z)=1$ and applying the expression of the reflection coefficients, i.e. $\rho(z)=s_{21}(z)/s_{11}(z), \tilde{\rho}(z)=s_{12}(z)/s_{22}(z)$, we obtain that
\begin{align}\label{dTT-2}
\zeta^{-}_{2}(z)\zeta_{2}^{+}(z)=\frac{1}{1-\rho(z)\tilde{\rho}(z)},\quad z\in\Sigma.
\end{align}
Based on the asymptotic behavior of $S(z)$ which we have given in \emph{Theorem 2.7}, we can obtain that $\zeta^{\mp}_{2}(z)\rightarrow1$ as $z\rightarrow\infty$. Then, via taking the logarithm of the \eqref{dTT-2} and applying the Plemelj's formulae and Cauchy projectors, we obtain that
\begin{align}\label{dTT-3}
\log\zeta^{\mp}_{2}(z)=-\frac{1}{2\pi i}\int_{\Sigma}
\frac{\log[1-\rho(s)\tilde{\rho}(s)]}{s-(z\pm i0)}\,ds,\quad z\in D_{\mp}.
\end{align}
Therefore, the trace formula can be obtained as
\begin{align}
s_{11}(z)&=exp\left(-\frac{1}{2\pi i }\int_{\Sigma}\frac{\log[1-\rho(s)\tilde{\rho}(s)]}{s-z}
\,d\zeta\right)\prod_{n=1}^{N}\frac{(z-z_{n})^{2}(z+q_{0}^{2}/z_{n}^{*})^{2}}
{(z-z_{n}^{*})^{2}(z+q_{0}^{2}/z_{n})^{2}},\\
s_{22}(z)&=exp\left(-\frac{1}{2\pi i }\int_{\Sigma}\frac{\log[1-\rho(s)\tilde{\rho}(s)]}{s-z}
\,ds\right)\prod_{n=1}^{N}\frac{(z-z_{n}^{*})^{2}(z+q_{0}^{2}/z_{n})^{2}}
{(z-z_{n})^{2}(z+q_{0}^{2}/z_{n}^{*})^{2}}.
\end{align}
Finally, based on the asymptotic behavior of $s_{11}$, the theta condition can be obtained as
\begin{align}
\arg\frac{q_{-}}{q_{+}}=\arg q_{-}-\arg q_{+}=8\sum_{n=1}^{N}\arg z_{n}+\frac{1}{2\pi}\int_{\Sigma}
\frac{\log[1-\rho(s)\tilde{\rho}(s)]}{s}\,ds.
\end{align}

\subsection{Reflection-less potentials}
Here, we are interested to investigate a type solutions which are that the reflection coefficients $\rho(z)$ and $\tilde{\rho}(z)$ disappear. Therefore, the jump matrix from $M^{-}$ to $M^{+}$ also vanishes, i.e. $G(x,t;z)=0$. Under this conditions, from the \eqref{E-51},  we have
\begin{align}\label{E-52}
q(x,t)=-q_{-}-i\sum_{n=1}^{2N}\hat{E}_{n}e^{-2i\theta(x,t;\hat{\xi}_{n})}
[\mu'_{-,11}(x,t;\hat{\xi}_{n})+\mu_{-,11}(x,t;\hat{\xi}_{n})D_{n}(x,t)],
\end{align}
where the $\mu_{-,11}(x,t;\hat{\xi}_{n})$ and $\mu'_{-,11}(x,t;\hat{\xi}_{n})$ can be solved by the following equation
\begin{align}\label{E-53}
MX=V,
\end{align}
where $X_{n}=\mu_{-,11}(x,t;\hat{\xi}_{n})$, $X_{2N+n}=\mu'_{-,11}(x,t;\hat{\xi}_{n})$, $V_{n}=\frac{iq_{-}}{\xi_{n}}$, $V_{2N+n}=\frac{iq_{-}}{\xi^{2}_{n}}$ $(n=1,2,\cdots,2N)$ and M is a $4N\times4N$ matrix. The elements can be determined by the equation \eqref{E-48} and \eqref{E-49} which are under the condition that the reflection coefficients $\rho(z)$ and $\tilde{\rho}(z)$ disappear.

\subsection{Soliton solutions for the double poles case}
In this section, the properties of the soliton solutions for the double case will be analysed. Similar to the simple pole case, we first investigate the solution when the dimensionless parameter $\epsilon$ is zero. When $N=1$, we apply the appropriate parameters and get the following images.
\\

{\rotatebox{0}{\includegraphics[width=3.6cm,height=3.0cm,angle=0]{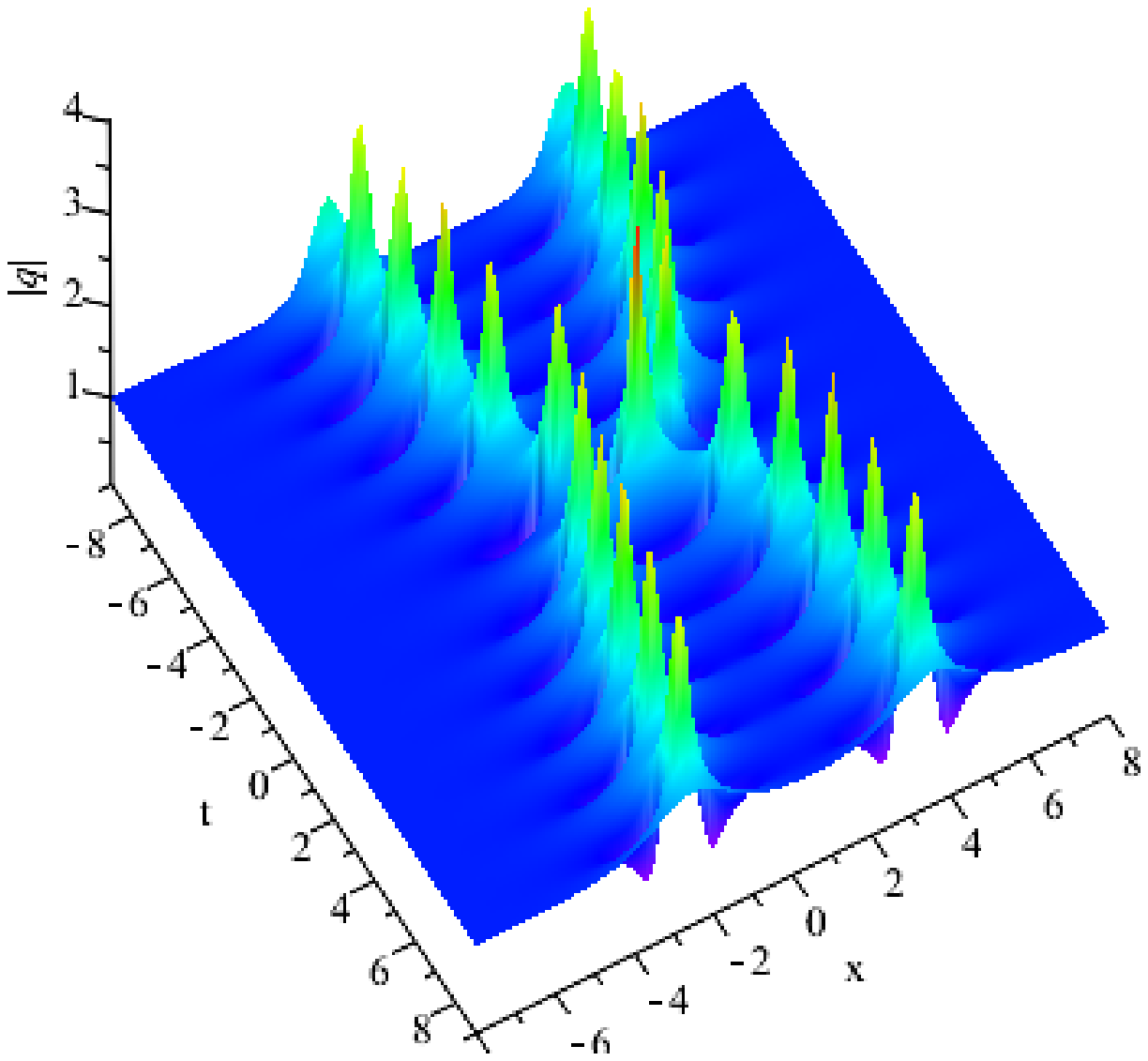}}}
~~~~
{\rotatebox{0}{\includegraphics[width=3.6cm,height=3.0cm,angle=0]{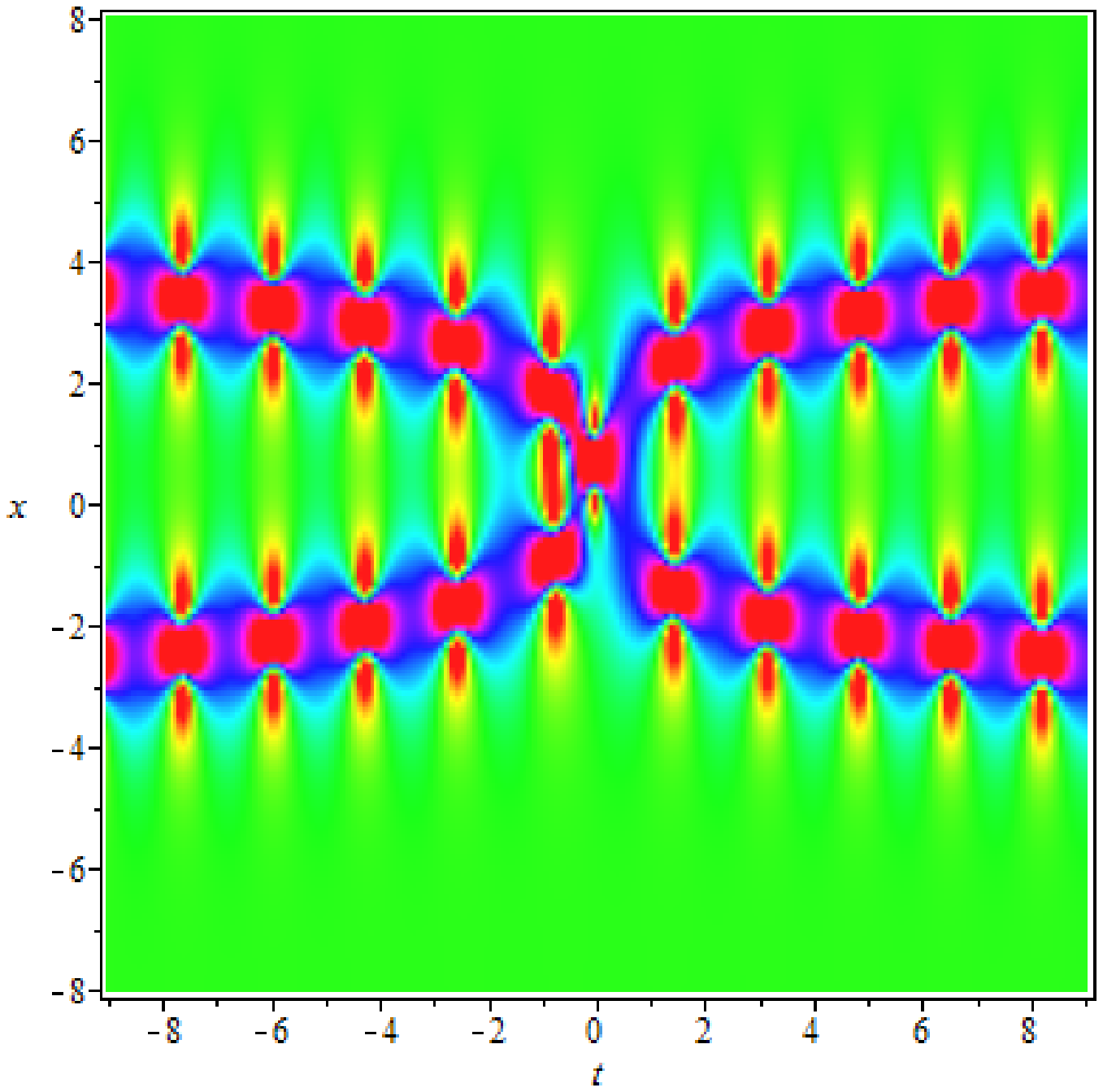}}}
~~~~
{\rotatebox{0}{\includegraphics[width=3.6cm,height=3.0cm,angle=0]{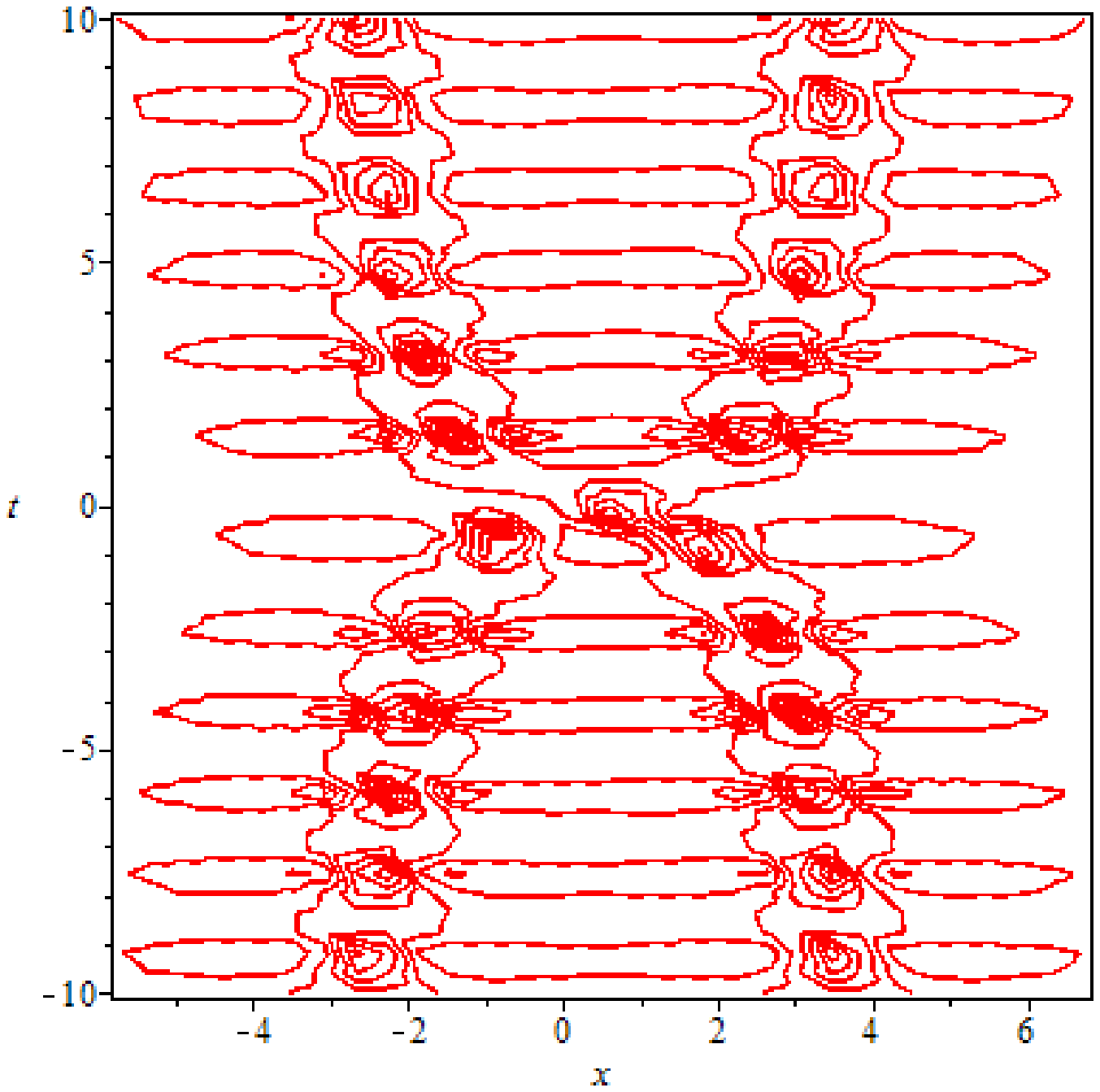}}}

$\qquad~~~~~~~~~(\textbf{a})\qquad \ \qquad\qquad\qquad\qquad~~~(\textbf{b})
\qquad\qquad\qquad\qquad\qquad~(\textbf{c})$\\
\noindent { \small \textbf{Figure 6.} (Color online) Plots of the soliton solution of the equation  with the parameters $\epsilon=0$, $q_{-}=1$, $\xi_{1}=-2i$ and $e_{1}=h_{1}=e^{1+i}$.
$\textbf{(a)}$: the soliton solution ,
$\textbf{(b)}$: the density plot ,
$\textbf{(c)}$: the contour line of the soliton solution.} \\

From the Figure 6, the interesting phenomenon that there are two columns of breather solutions are shown. It is worth noting that two columns of breather solutions interact in the process of propagation. Then, we change the boundary value $q_{-}$ and get the following graphs.\\

{\rotatebox{0}{\includegraphics[width=3.6cm,height=3.0cm,angle=0]{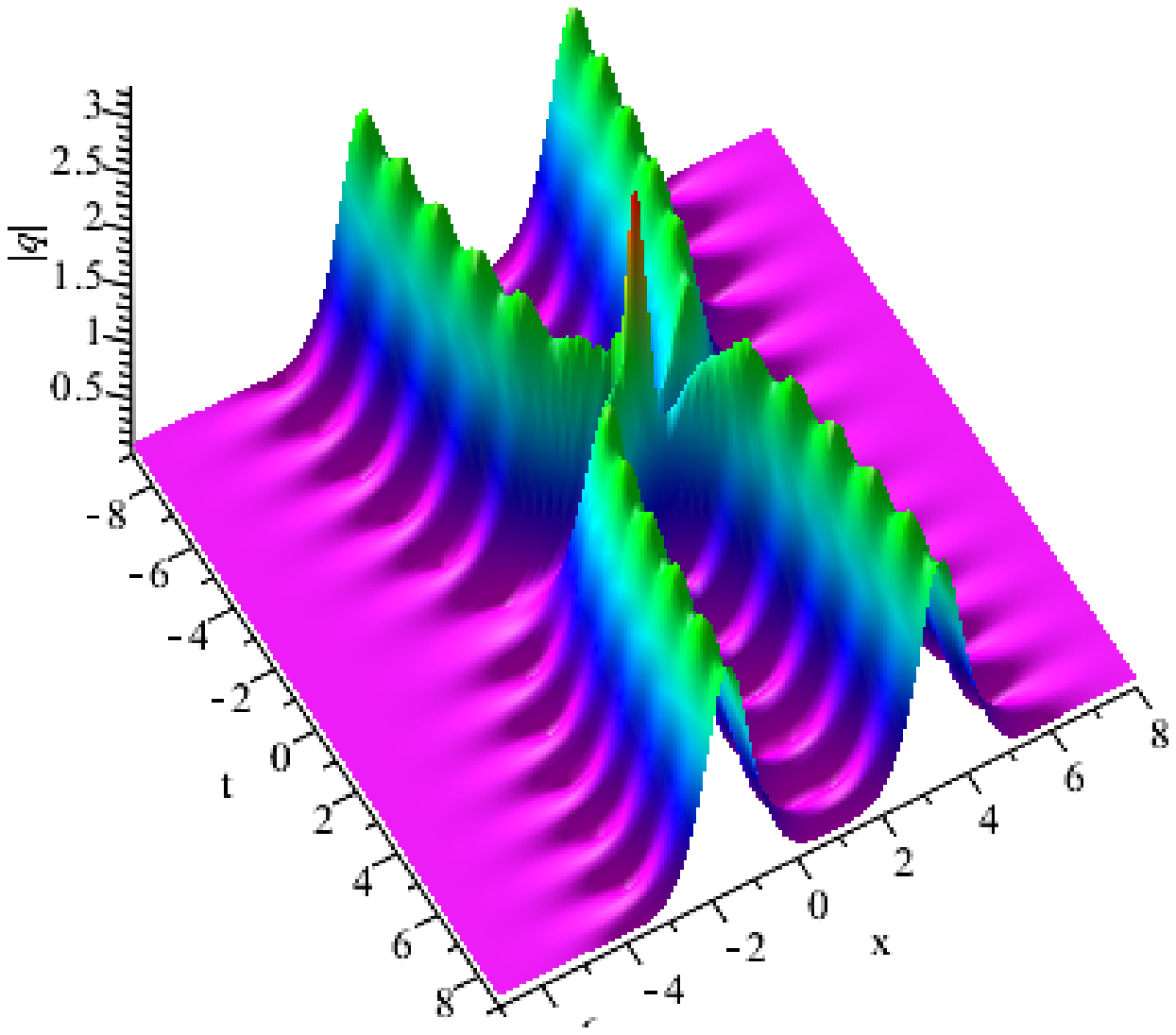}}}
~~~~
{\rotatebox{0}{\includegraphics[width=3.6cm,height=3.0cm,angle=0]{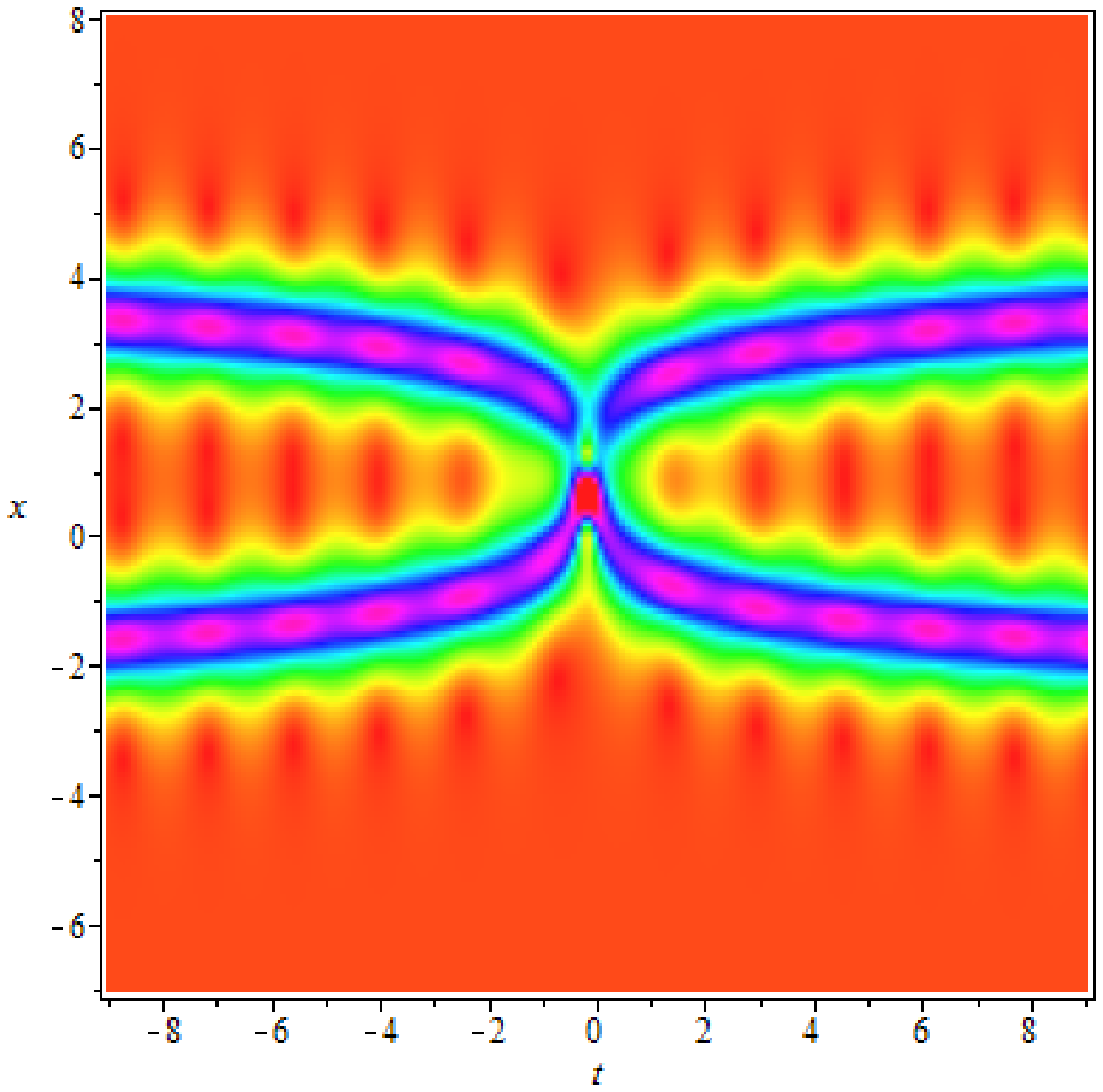}}}
~~~~
{\rotatebox{0}{\includegraphics[width=3.6cm,height=3.0cm,angle=0]{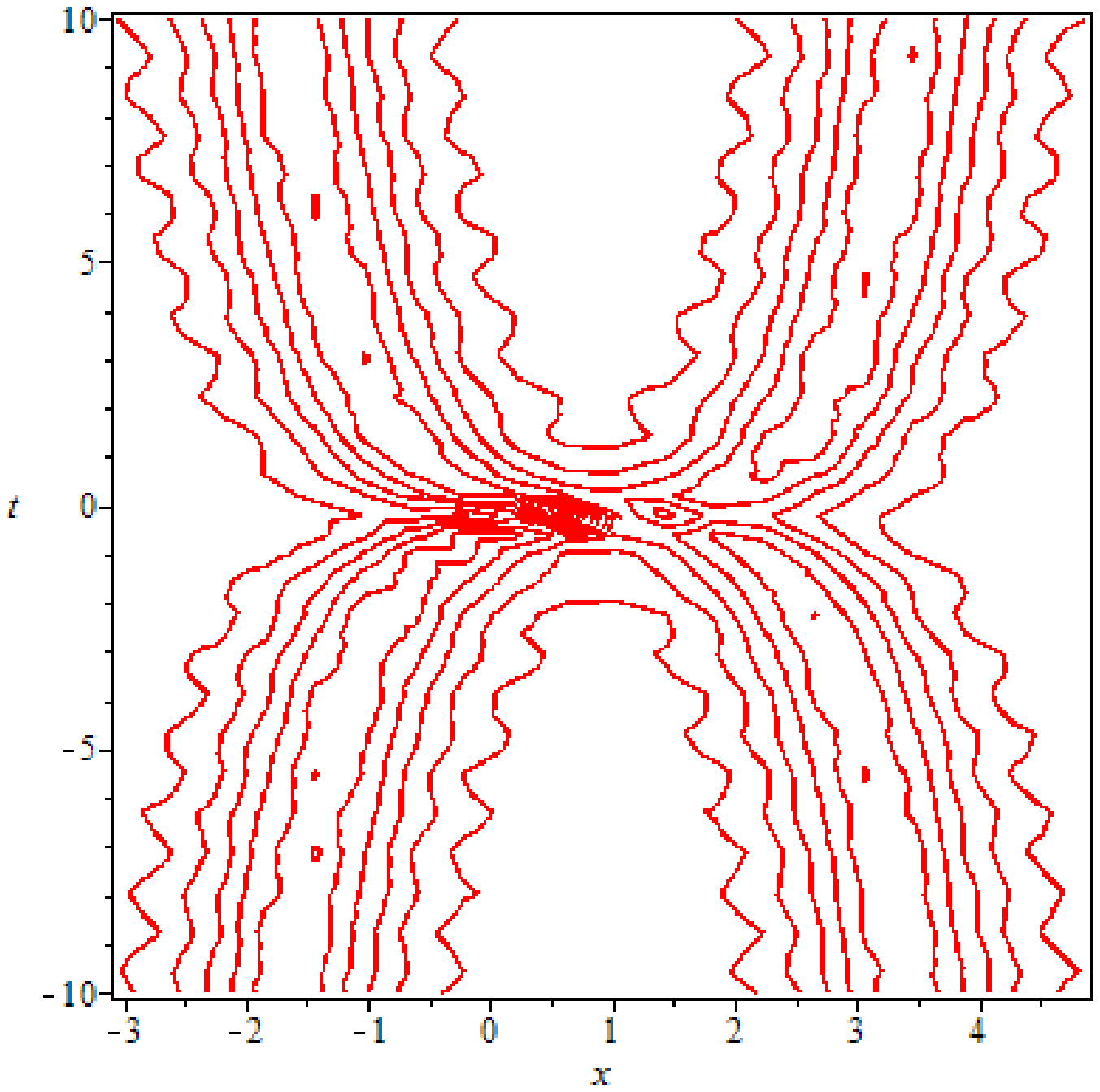}}}

$\qquad~~~~~~~~~(\textbf{a})\qquad \ \qquad\qquad\qquad\qquad~~~(\textbf{b})
\qquad\qquad\qquad\qquad\qquad~(\textbf{c})$\\

{\rotatebox{0}{\includegraphics[width=3.6cm,height=3.0cm,angle=0]{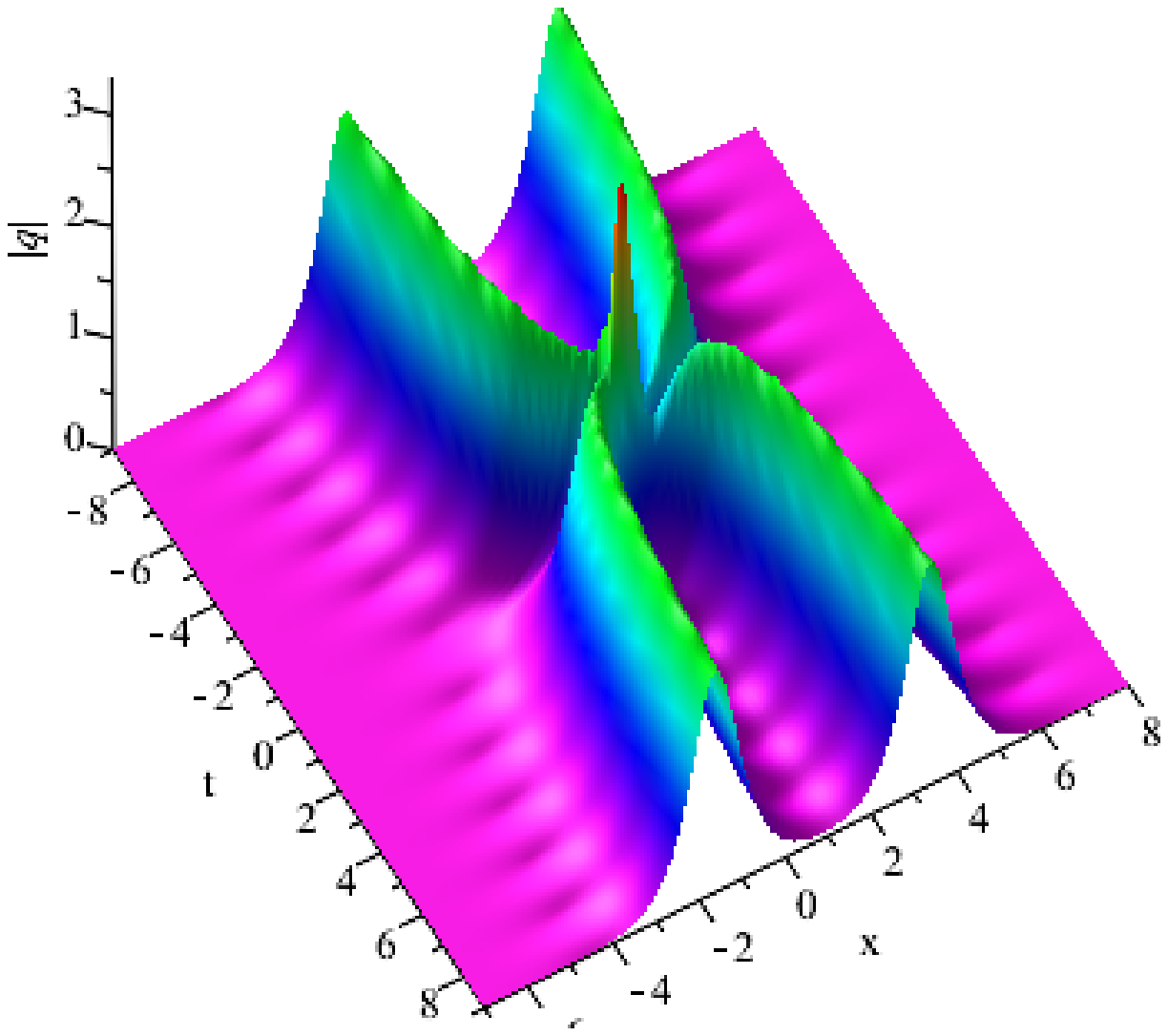}}}
~~~~
{\rotatebox{0}{\includegraphics[width=3.6cm,height=3.0cm,angle=0]{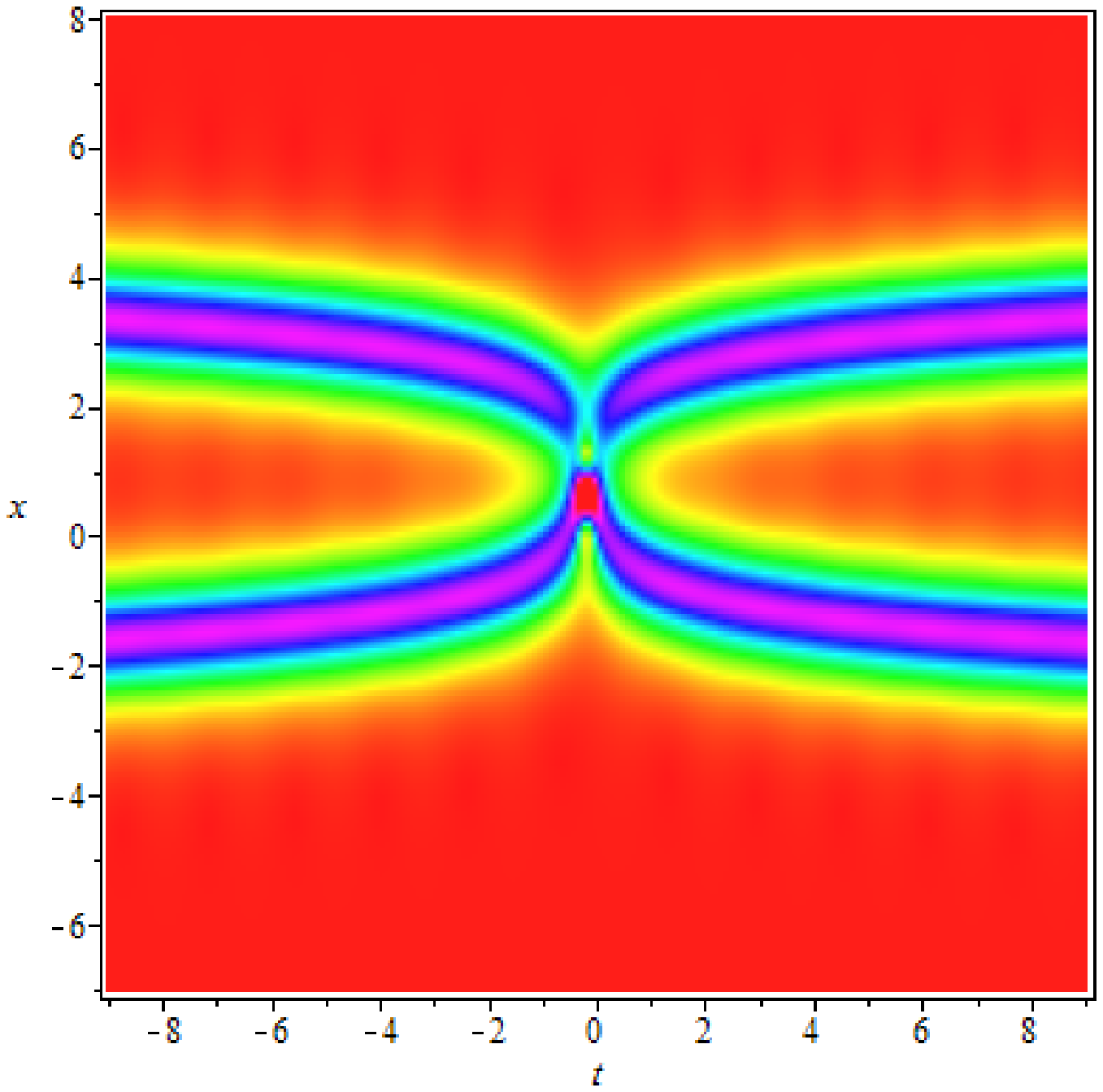}}}
~~~~
{\rotatebox{0}{\includegraphics[width=3.6cm,height=3.0cm,angle=0]{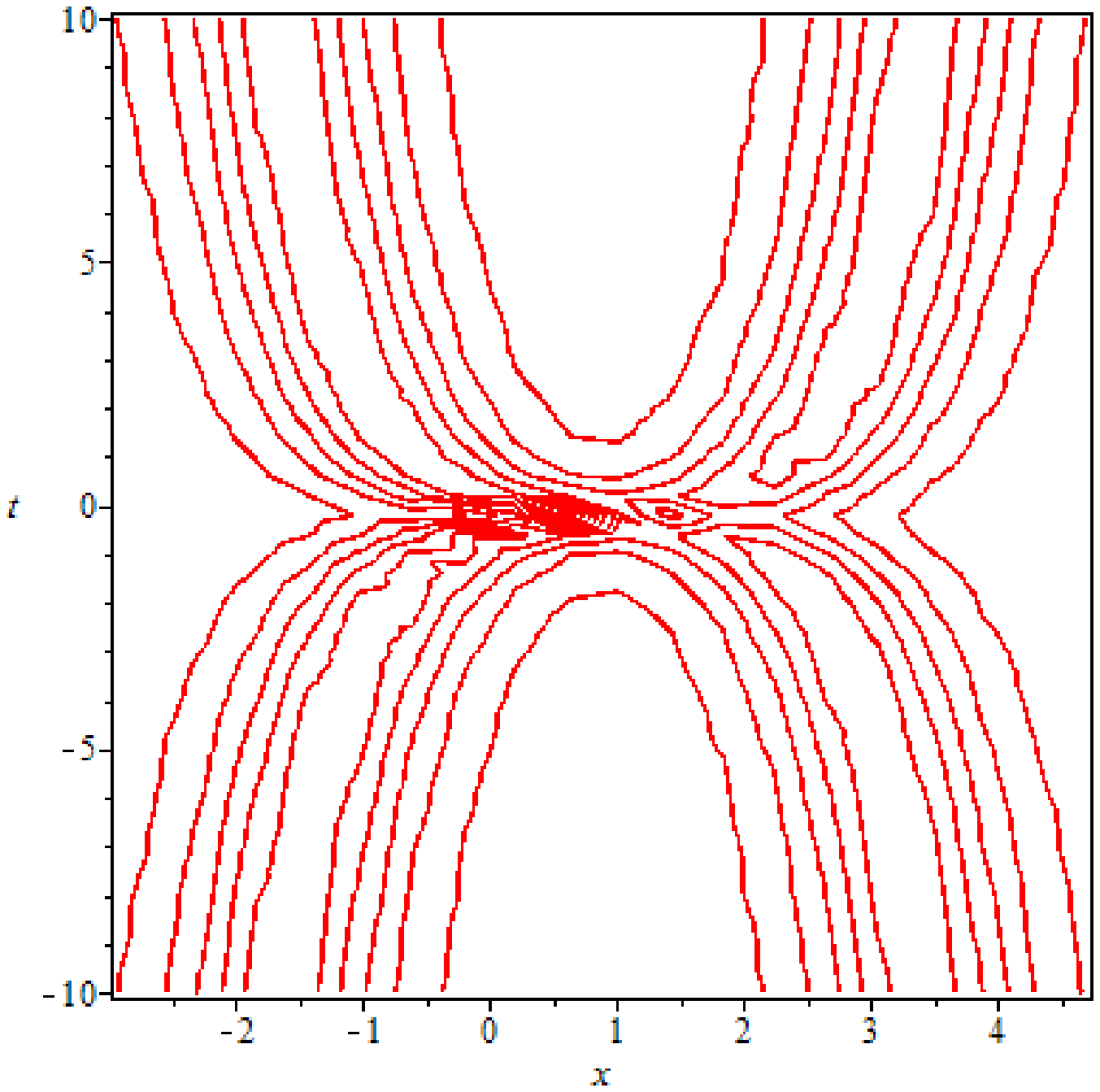}}}

$\qquad~~~~~~~~~(\textbf{d})\qquad \ \qquad\qquad\qquad\qquad~~~(\textbf{e})
\qquad\qquad\qquad\qquad\qquad~(\textbf{f})$\\
\noindent { \small \textbf{Figure 7.} (Color online) Plots of the soliton solutions of the equation  with the parameters $\epsilon=0$, $\xi_{1}=-2i$ and $e_{1}=h_{1}=e^{1+i}$.
$\textbf{(a)}$: the soliton solution with $q_{-}=0.1$,
$\textbf{(b)}$: the density plot corresponding to $(a)$,
$\textbf{(c)}$: the contour line of the soliton solution corresponding to $(a)$,
$\textbf{(d)}$: the soliton solution with $q_{-}=0.01$,
$\textbf{(e)}$: the density plot corresponding to $(d)$,
$\textbf{(f)}$: the contour line of the soiton solution corresponding to $(d)$.} \\

What we can learn from the Figure 7. are that there is only a sharp soliton exists in the place where the two columns of waves interact during the propagation process as $q_{-}$ gradually becomes smaller. While, the breathing phenomenon gradually disappears. Furthermore, if we change the dimensionless parameter $\epsilon$, more interesting phenomena will be generated. Now, we consider the case that the dimensionless parameter $\epsilon$ is non-zero.\\

{\rotatebox{0}{\includegraphics[width=3.6cm,height=3.0cm,angle=0]{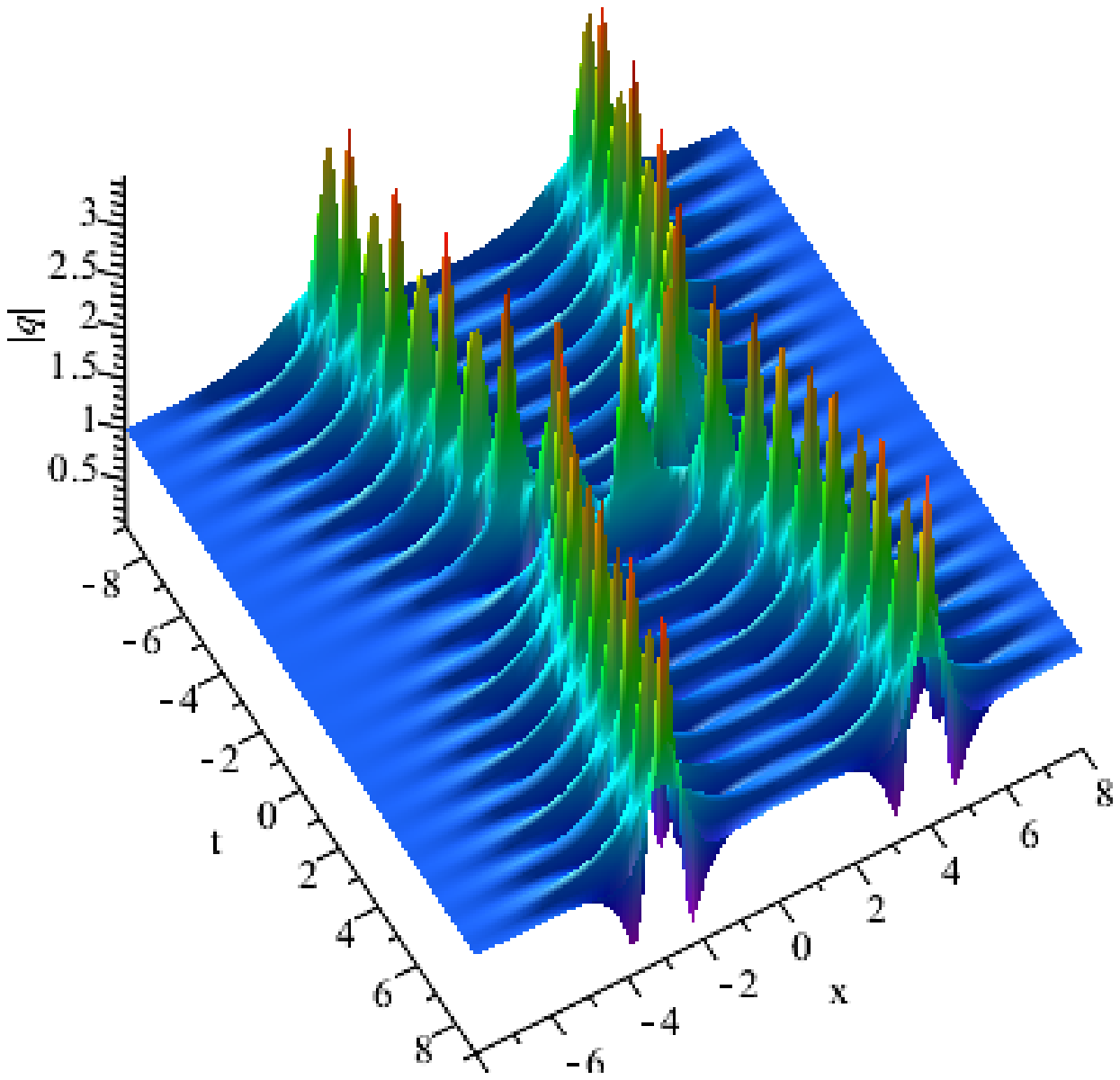}}}
~~~~
{\rotatebox{0}{\includegraphics[width=3.6cm,height=3.0cm,angle=0]{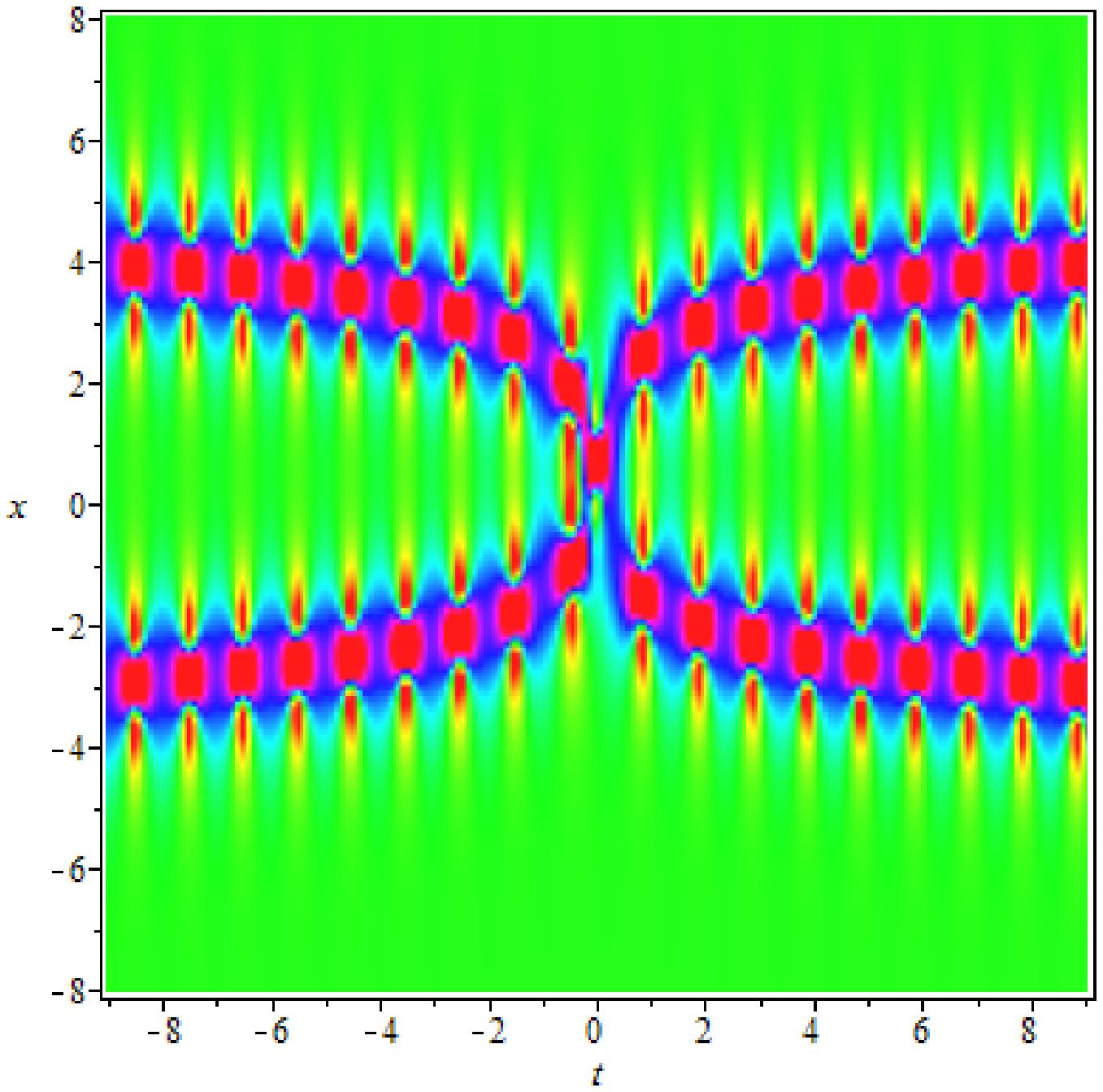}}}
~~~~
{\rotatebox{0}{\includegraphics[width=3.6cm,height=3.0cm,angle=0]{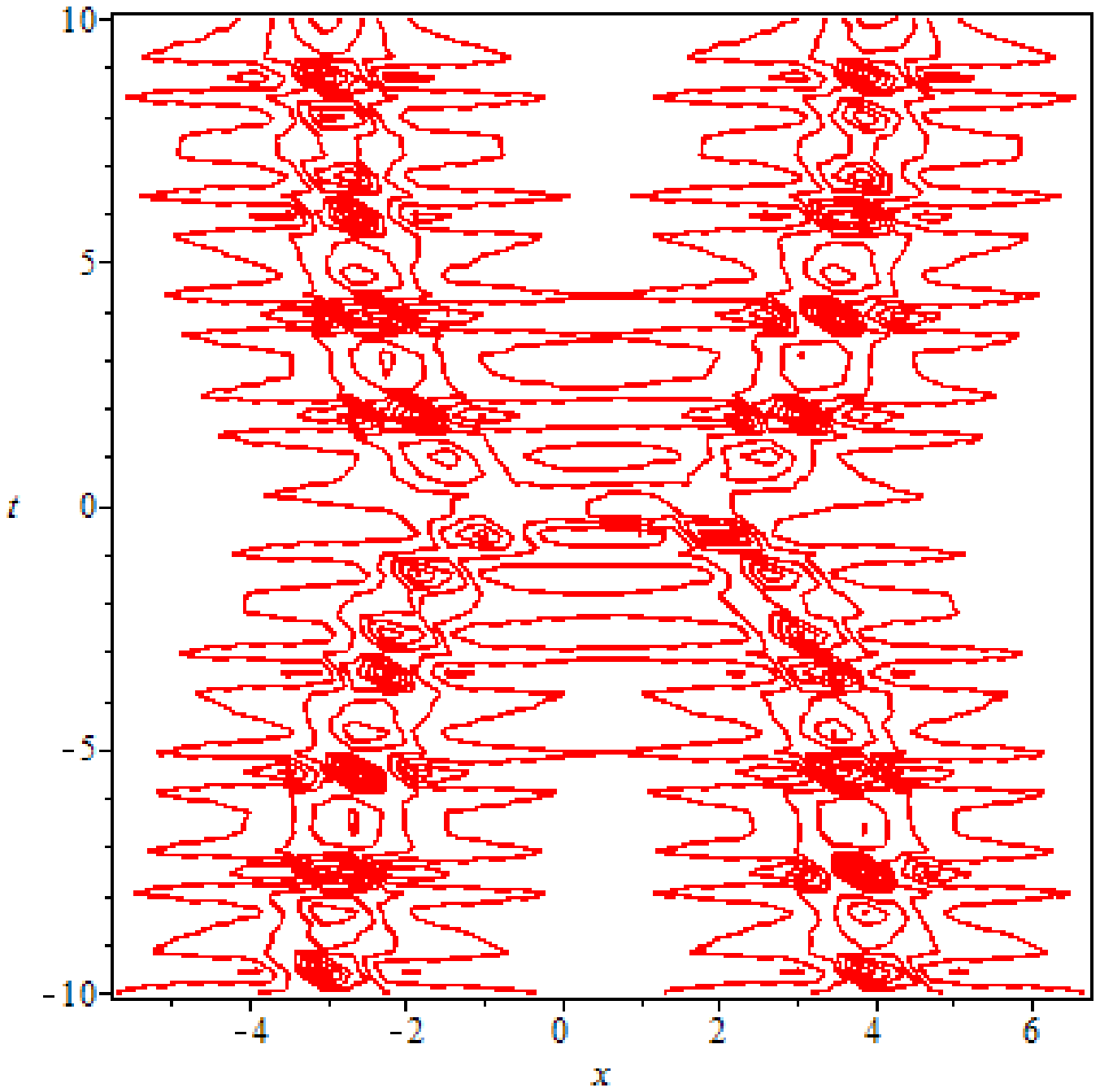}}}

$\qquad~~~~~~~~~(\textbf{a})\qquad \ \qquad\qquad\qquad\qquad~~~(\textbf{b})
\qquad\qquad\qquad\qquad\qquad~(\textbf{c})$ \\

{\rotatebox{0}{\includegraphics[width=3.6cm,height=3.0cm,angle=0]{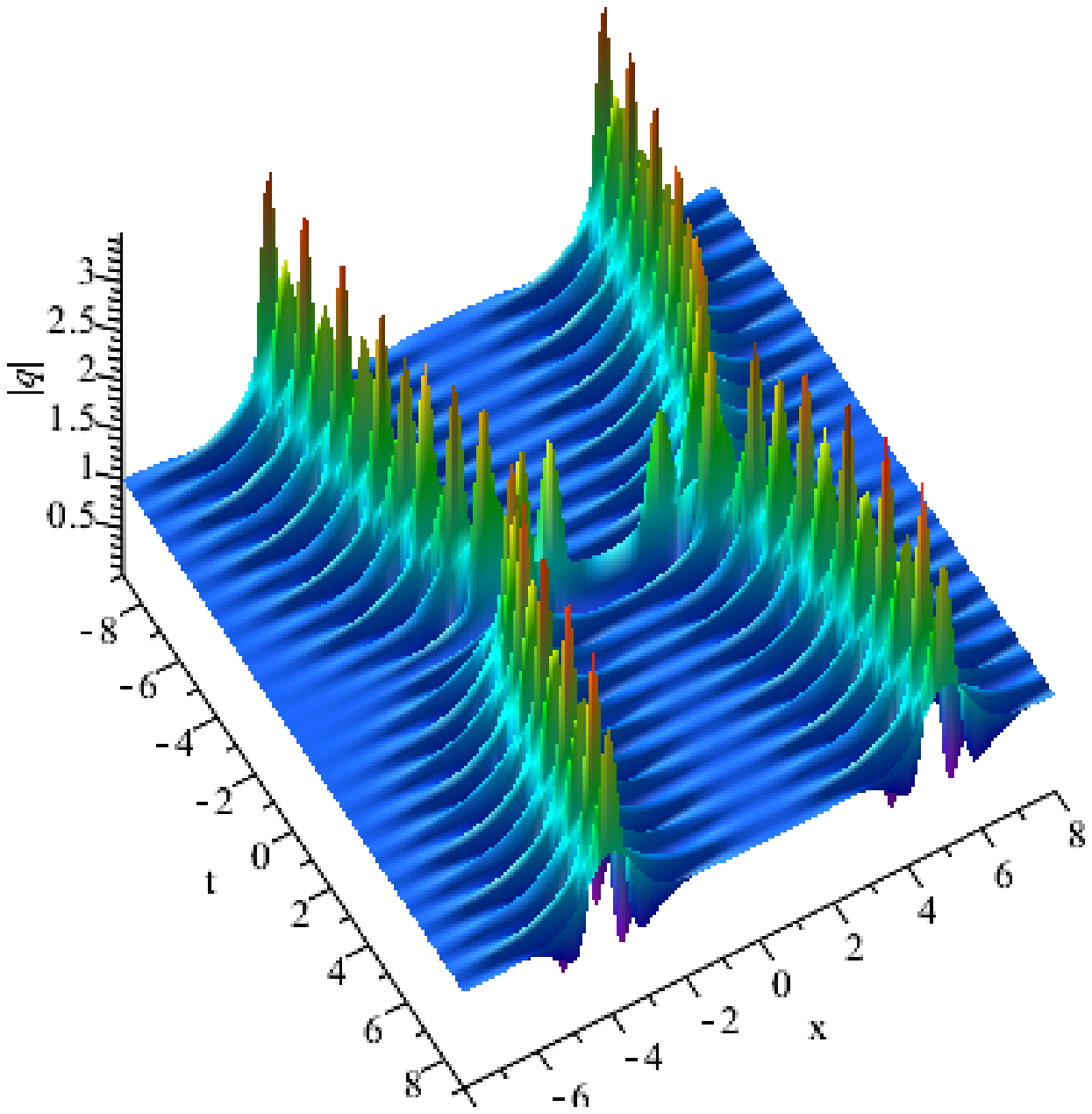}}}
~~~~
{\rotatebox{0}{\includegraphics[width=3.6cm,height=3.0cm,angle=0]{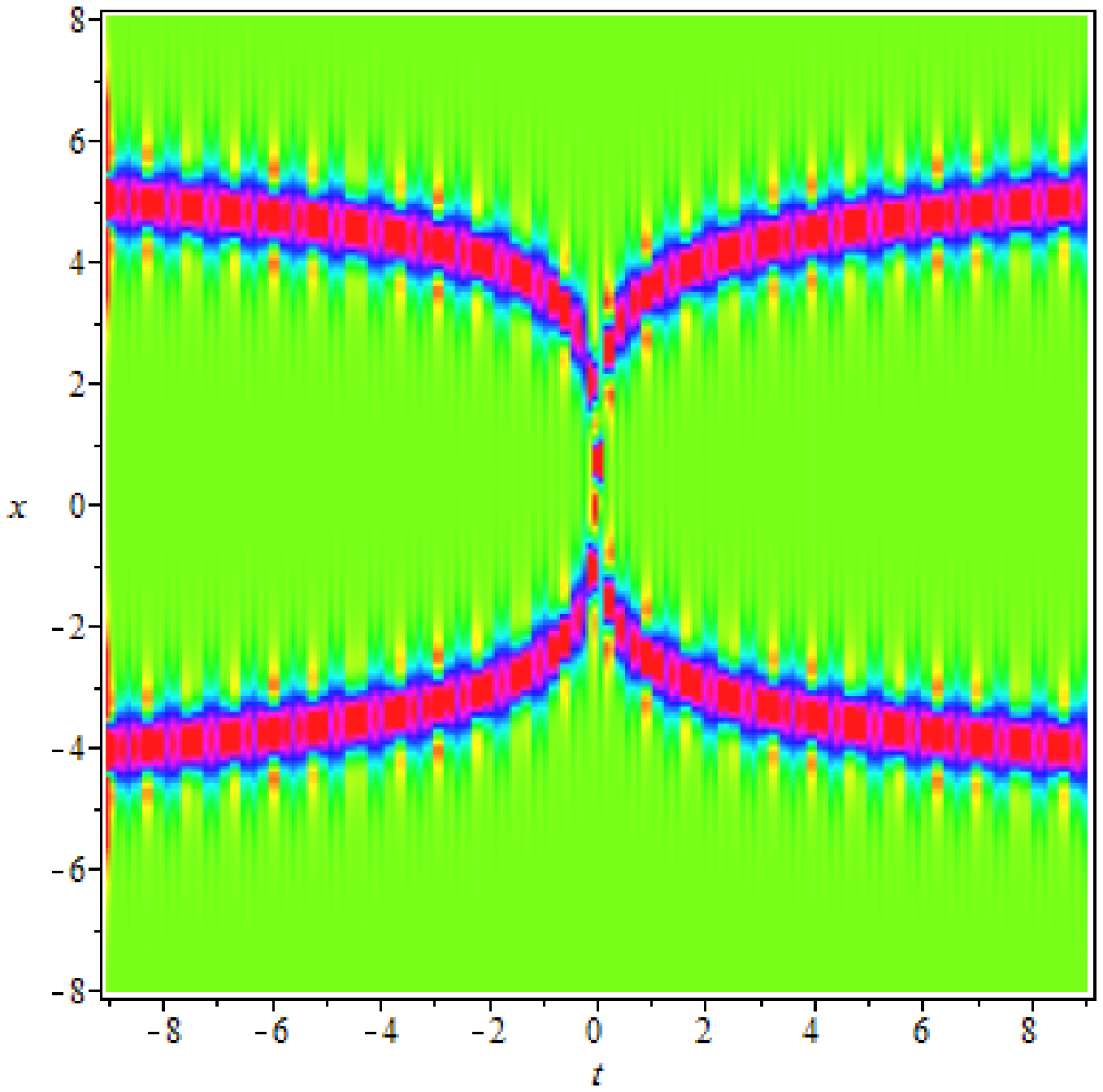}}}
~~~~
{\rotatebox{0}{\includegraphics[width=3.6cm,height=3.0cm,angle=0]{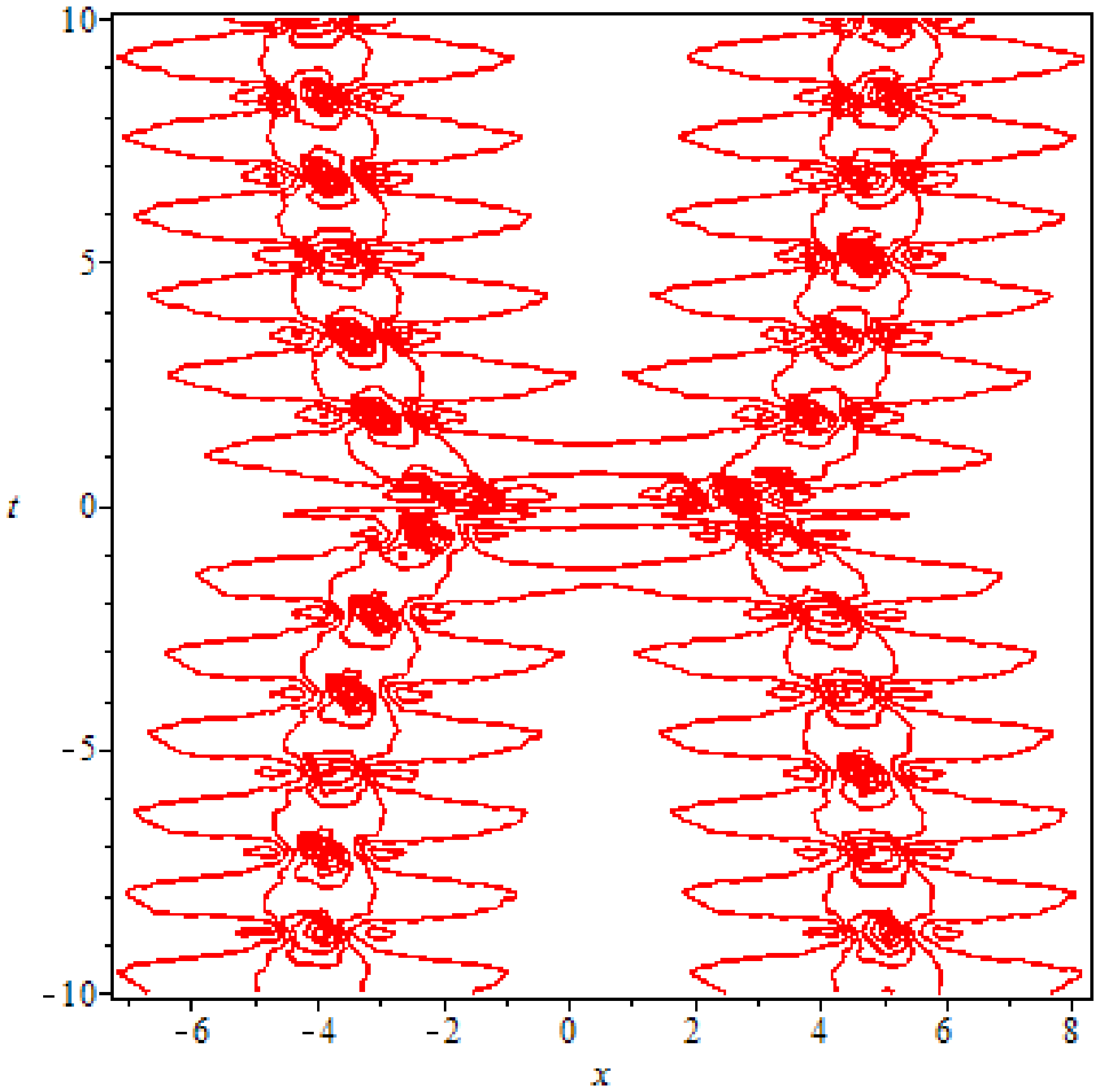}}}

$\qquad~~~~~~~~~(\textbf{d})\qquad \ \qquad\qquad\qquad\qquad~~~(\textbf{e})
\qquad\qquad\qquad\qquad\qquad~(\textbf{f})$\\
\noindent { \small \textbf{Figure 8.} (Color online) Plots of the soliton solutions of the equation  with the parameters $q_{-}=1$, $\xi_{1}=-2i$ and $e_{1}=h_{1}=e^{1+i}$.
$\textbf{(a)}$: the soliton solution with $\epsilon=1$,
$\textbf{(b)}$: the density plot corresponding to $(a)$,
$\textbf{(c)}$: the contour line of the soliton solution corresponding to $(a)$,
$\textbf{(d)}$: the soliton solution with $\epsilon=3$,
$\textbf{(e)}$: the density plot corresponding to $(d)$,
$\textbf{(f)}$: the contour line of the soliton solution corresponding to $(d)$.} \\

From the Figure 8. and comparing it with Figure 6, it is interesting that the both of the two column breather soliton solutions are more closely arranged as the dimensionless parameter $\epsilon$ becomes larger. For this phenomenon, the density plot shows more clearly. Similar to the simple pole, we can see that the appearance of $\epsilon$  will disturb the form of the soliton solution and has no effect on the structure of the soliton solution.

\section{ Conclusions and discussions}
In  \cite{I-20}, Yang, et al. have studied the rogue wave solutions and obtained rogue waves dynamics and several new spatial-temporal structures of the HDNLS equation by using the generalized Darboux transformation method. However, we investigated the HDNLS equation \eqref{1.1} with non-zero boundary conditions by applying the inverse scattering transform via RH approach which is quite different from the generalized Darboux transformation method. Meanwhile, the specific form of the analytical solution we obtained is different from the results in \cite{I-20}. In addition, we fully discuss the effects of different boundary conditions, discrete spectral points, and disturbance $\epsilon$ on the soliton solutions. Also, some interesting phenomena are obtained   when the spectrum points tend to singular points   by choosing appropriate parameters.

In this work, we studied the HDNLS equation with NBCs at infinity and presented the ISTs. Firstly, we overcome difficulties that the double-valued functions occur in the process of direct scattering through through introducing a appropriate Riemann surface and uniformization variable. After the discussion of the direct scattering problem and the inverse scattering problem, the solutions of the HDNLS equation with NBCs are presented. Meanwhile, special soliton solutions under the condition of reflection-less potentials are given for both of the two case i.e. simple and double poles. In addition, based on the concrete expression of the solution, some graphic analysis are presented via selecting some appropriate parameters.

\section*{Acknowledgements}
This work was supported by the Postgraduate Research and Practice of Educational Reform for Graduate students in CUMT under Grant No. 2019YJSJG046, the Natural Science Foundation of Jiangsu Province under Grant No. BK20181351, the Six Talent Peaks Project in Jiangsu Province under Grant No. JY-059, the Qinglan Project of Jiangsu Province of China, the National Natural Science Foundation of China under Grant No. 11975306, the Fundamental Research Fund for the Central Universities under the Grant Nos. 2019ZDPY07 and 2019QNA35, and the General Financial Grant from the China Postdoctoral Science Foundation under Grant Nos. 2015M570498 and 2017T100413.

\renewcommand{\baselinestretch}{1.2}

\end{document}